\documentclass[a4paper,twocolumn,superscriptaddress,11pt,accepted=2020-03-19]{quantumarticle}
\pdfoutput=1
\usepackage[utf8]{inputenc}
\usepackage[english]{babel}
\usepackage[T1]{fontenc}
\usepackage{amsmath}
\let\doi\relax
\RequirePackage{doi}
\usepackage{hyperref}
\usepackage[protrusion=true,expansion=true]{microtype}

\makeatletter
\newcommand\xleftrightarrow[2][]{%
  \ext@arrow 9999{\longleftrightarrowfill@}{#1}{#2}}
\newcommand\longleftrightarrowfill@{%
  \arrowfill@\leftarrow\relbar\rightarrow}
\makeatother

\usepackage{amsthm}
\usepackage{amssymb}
\usepackage{cleveref}

\usepackage{todonotes}

\usepackage{stmaryrd}

\usepackage{tikz-cd}
\usepackage{tikz}
\usetikzlibrary{quantikz}

\newcommand{\su}{\mathfrak{su}}

\newcommand{\A}{\mathfrak A}

\newcommand{\C}{\mathbb C}
\newcommand{\CP}{\C \mathrm P}
\newcommand{\Gr}{\mathrm{Gr}}
\renewcommand{\L}{\mathcal{L}}
\newcommand{\M}{\mathcal{M}}
\renewcommand{\P}{\mathcal P}
\newcommand{\Q}{\mathcal Q}
\newcommand{\R}{\mathbb R}
\renewcommand{\S}{\mathcal S}
\newcommand{\Z}{\mathbb Z}

\newcommand{\co}{\colon\thinspace}
\newcommand{\mmod}{/\!\!/}

\newcommand{\LogCoords}{\Pi}

\newcommand{\B}{\mathrm{B}}
\newcommand{\CNOT}{\mathrm{CNOT}}
\newcommand{\CZ}{\mathrm{CZ}}
\newcommand{\CPHASE}{\mathrm{CPHASE}}
\newcommand{\DB}{\mathrm{DB}}
\newcommand{\Haar}{\mathrm{Haar}}
\newcommand{\ISWAP}{i\mathrm{SWAP}}
\newcommand{\PSWAP}{\mathrm{PSWAP}}
\newcommand{\X}{\mathrm{X}}
\newcommand{\Y}{\mathrm{Y}}
\renewcommand{\Z}{\mathrm{Z}}
\newcommand{\I}{\mathrm{I}}
\newcommand{\SWAP}{\mathrm{SWAP}}
\newcommand{\XY}{\mathrm{XY}}
\newcommand{\CAN}{\mathrm{CAN}}

\renewcommand{\epsilon}{\varepsilon}

\renewcommand{\th}{\textsuperscript{th}}
\newcommand{\avg}{\mathrm{avg}}

\let\Im\undefined

\DeclareMathOperator{\Ad}{Ad}

\DeclareMathOperator*{\Im}{Im}
\DeclareMathOperator{\Mon}{Mon}

\DeclareMathOperator{\Spec}{Spec}
\DeclareMathOperator{\LogSpec}{LogSpec}
\DeclareMathOperator{\tr}{tr}
\DeclareMathOperator{\vol}{vol}

\newcommand{\<}{\langle}
\renewcommand{\>}{\rangle}

\theoremstyle{plain}
\newtheorem{theorem}{Theorem}
\newtheorem*{theorem*}{Theorem}
\newtheorem{lemma}[theorem]{Lemma}
\newtheorem*{lemma*}{Lemma}
\newtheorem{corollary}[theorem]{Corollary}
\newtheorem*{corollary*}{Corollary}
\theoremstyle{remark}
\newtheorem{remark}[theorem]{Remark}
\newtheorem{example}[theorem]{Example}
\theoremstyle{definition}
\newtheorem{definition}[theorem]{Definition}
\newtheorem{problem}[theorem]{Problem}
\newtheorem*{problem*}{Problem}

\begin{document}

\title{Two-Qubit Circuit Depth and the Monodromy Polytope}
\author{Eric C. Peterson}
\author{Gavin E. Crooks}
\author{Robert S. Smith}
\affiliation{Rigetti Quantum Computing, 2919 Seventh St, Berkeley, CA 94710}

\maketitle

\begin{abstract}
  For a native gate set which includes all single-qubit gates, we apply results from symplectic geometry to analyze the spaces of two-qubit programs accessible within a fixed number of gates.  These techniques yield an explicit description of this subspace as a convex polytope, presented by a family of linear inequalities themselves accessible via a finite calculation.  We completely describe this family of inequalities in a variety of familiar example cases, and as a consequence we highlight a certain member of the ``$\XY$--family'' for which this subspace is particularly large, i.e., for which many two-qubit programs admit expression as low-depth circuits.
\end{abstract}


\section{Introduction}\label{Introduction}

Compilers for quantum computers have two primary tasks.  One is to convert a hardware-agnostic description of an algorithm to a hardware-aware description suitable for execution on a particular physical device.  This is an involved process, owing both to the idiosyncratic limitations of quantum computational devices and to the extremely large space of quantum programs.  Ideal, ``pure'' quantum programs, which do not interact with the outside world until termination, can be interpreted as points in the projective unitary group $PU(2^q)$ (i.e., unitaries neglecting the effects of global phase), where $q$ is the number of qubits in the system.  For all $q > 0$, the group $PU(2^q)$ is infinite, and so quantum compilers must draw on methods from continuous mathematics to accomplish their task.

Optimized expression of a program, a compiler's second task, is of particular interest to programmers of quantum devices which do not yet enjoy fault tolerance.  If each instruction has the potential to introduce error into the computation, then after sufficiently many instructions are enacted, the state of the quantum device will no longer even approximate the programmer's intent.  Correspondingly, optimization passes in a quantum compiler which lower circuit depth provide a form of noise mitigation, and hence they contribute not just to expedience but to correctness.

In light of this, optimality results for decompositions are of interest to quantum compiler designers.  In the presence of recursive compilation schemes (e.g., Quantum Shannon Decomposition~\cite[Theorem 13]{SBMSynthesis}), such results for low numbers of qubits are particularly interesting, as they improve the overall output of the compilation routine by a scalar factor.  Some of the most advanced such results to date include: Shende, Bullock, and Markov~\cite{SBM} showed (using an existing framework, see e.g.\ \cite{KrausCirac}) that all two-qubit programs can be expressed using three applications of the $\CZ$--gate,\footnote{Equivalently: three applications of $\CNOT$ gates.} interleaved with single-qubit rotations; Zhang, Vala, Sastry, and Whaley~\cite{ZVSWMinimum} showed that all two-qubit programs can be expressed using two applications of the $\B$--gate, interleaved with single-qubit rotations; the same group showed that a wide class of exponential families of two-qubit gates can be used to implement an arbitrary two-qubit program, using three applications interleaved with single-qubit rotations~\cite{ZVSWGeometric}; and the same group gave near-optimal decompositions into controlled prescribed gates, ultimately relying on a convexity theorem~\cite{ZVSWControlled}.

This first set of results has two particularly interesting features.  First, the methods they describe are computationally tractable: one can actually construct their circuits by using standard algorithms in linear algebra.  Second, they use the same techniques in a follow-up paper~\cite{SMB} to analyze the subspace of programs which take no more than two $\CZ$s to implement, and they conclude that ``almost all'' two-qubit programs take three invocations of $\CZ$ to implement.  Comparing this with the results of the other authors suggests an interesting quality metric on native gate sets: let $\L_\S(U)$ be the minimum number of two-qubit gates needed to implement $U$ as a circuit with gates drawn from $\S$.  We then conclude that the expected values satisfy $\<\L_\CZ\> = 3$, $\<\L_\B\> = 2$, and $\<\L_{\mathcal H}\> \le 3$ for most gate sets of the form $\mathcal{H} = \{H_t = \exp(t h) \mid t \in \R\}$ with $h$ a fixed anti-Hermitian matrix.

Although physical devices with native multi-qubit operations other than $\CZ$ have been implemented, optimality results and expected gate counts for these other native gate sets have not yet appeared.  In this paper, we offer tools for the analysis of these problems at the following level of generality:

\begin{theorem*}[{\Cref{PnAreAllPolytopes}}]
For $\S$ a finite set\footnote{The finiteness assumption on $S$ may be relaxed to account for such families as $\CPHASE_\alpha$.} of two-qubit operations and for $n \ge 0$, let $P^n_\S \subseteq PU(4)$ be the following set of two-qubit programs
\[P^n_\S = \left\{
\begin{tikzcd}[column sep=1em]
& \gate{A_n} & \gate[2]{S_n} & \qw \cdots & \gate{A_1} & \gate[2]{S_1} & \gate{A_0} & \qw \\
& \gate{B_n} & \qw & \qw \cdots & \gate{B_1} & \qw & \gate{B_0} & \qw
\end{tikzcd} \right\}\]
for $S_j \in \S$ and $A_j, B_j \in PU(2)$.  In a certain coordinate system to be described below, $P^n_\S$ can be expressed as a union of $(2|\S|)^{n}$ convex polytopes, each described by a (typically highly redundant) family of linear inequalities of size exponential in $n$.
\end{theorem*}

We use these results to explore the space of possible choices for the native gate set, with an emphasis on those appearing via Rigetti's choice of interaction Hamiltonian~\cite{CaldwellEtAl} (cf.\ also~\cite{SchuchSiewert}): the gates $\CZ$, $\ISWAP$, $\CPHASE$, and $\XY$, where by $\XY$ we intend the unitary family
\begin{align*}
    \XY_\alpha & = \exp\left( -i \alpha (\sigma_X^{\otimes 2} + \sigma_Y^{\otimes 2}) \right) \\
    & = \left( \begin{array}{cccc} 1 & 0 & 0 & 0 \\ 0 & \cos\left(\frac{\alpha}{2}\right) & -i\sin\left(\frac{\alpha}{2}\right) & 0 \\ 0 & -i\sin\left(\frac{\alpha}{2}\right) & \cos\left(\frac{\alpha}{2}\right) & 0 \\ 0 & 0 & 0 & 1 \end{array}\right).
\end{align*}
Our methods in the case of $\S = \{\CZ\}$ recover results of Shende, Bullock, and Markov~\cite{SBM,SMB}.  In the other cases, we make the following conclusions:

\begin{corollary*}[{\Cref{ISWAPSameAsCZ}}]
The sets $P^2_{\ISWAP}$ and $P^3_{\ISWAP}$ are the same as the corresponding sets for $\S = \{\CZ\}$.  Hence, $P^2_{\ISWAP}$ occupies 0\% of the volume of all two-qubit programs, and $\<\L_{\ISWAP}\> = 3$.
\end{corollary*}

\begin{corollary*}[{\Cref{P2CPHASEDescription}}]
Allowing the parameter of $\CPHASE$ to range freely in $0 \le \alpha \le 2\pi$, the sets $P^2_{\CPHASE}$ and $P^3_{\CPHASE}$ are the same as the corresponding sets for $\S = \{\CZ\}$.  Hence, $P^2_{\CPHASE}$ occupies 0\% of the volume of all two-qubit programs, and $\<\L_\CPHASE\> = 3$.
\end{corollary*}

\noindent We find the situation to be quite different for $\XY$:

\begin{corollary*}[{Somewhat informal\footnote{In particular, the notion of ``volume'' is different from the usual Haar volume, and so ``randomly sampled'' also changes meaning.}; \Cref{LXYComputation}, \Cref{LDBComputation}}]
As a function of $\alpha$, the volume of the set $P^2_{\XY_\alpha}$ is maximized at $\alpha = 3\pi/4$, where it contains $75\%$ of randomly sampled two-qubit programs.  Correspondingly, $\<\L_{\XY_\alpha}\>$ is minimized as $\<\L_{\XY_\alpha}\> = 9/4$.  Allowing the parameter of $\XY$ to range freely, the set $P^2_{\XY}$ contains ${\approx}96\%$ of randomly sampled all two-qubit programs, with corresponding value $\<\L_{\XY}\> \approx 2.04$.
\end{corollary*}

\noindent The two most interesting features of this result are that the availability of gates in the $\XY$--family can have a dramatic effect on the optimal gate depth of a generic two-qubit program, and that the bulk of this effect is seen from a \emph{single} gate from this family.  At the magic value of $\alpha = 3 \pi / 4$, we provide an explicit routine for checking membership in this preferred subspace.

Our methods also lend themselves to an analysis of problems in approximate compilation.  Each of the sets $P^n_\S$ described above is a proper subset of the space of all two-qubit programs, and for a two-qubit program $U$ it is natural to search for the two-qubit program within $P^n_\S$ ``closest'' to $U$ and to give a precise expression for this distance.  We give a protocol describing the use of our techniques in this situation, and we give explicit computations of the best approximant and its minimum distance for certain interesting gates (e.g., $\SWAP$) and interesting gate sets (e.g., $\XY_{\frac{3\pi}{4}}$).

We include as appendices an introduction to the mathematics underpinning these results as well as a simpler viewpoint that yields similar qualitative results but is quantitatively inexplicit.

\section{The geometry of two-qubit programs and the canonical decomposition}\label{Geometryof2QSection}

As motivation, we include a brief treatment of the Euler decomposition of single-qubit programs into triples of rotations.  We begin by fixing notation:\footnote{N.B.: This normalization of the Pauli matrices is nonstandard.}
\begin{align*}
    \X_\alpha & = \exp\left( -i\alpha \sigma_X \right) = \left( \begin{array}{cc} \cos \frac{\alpha}{2} & -i \sin \frac{\alpha}{2} \\ -i \sin \frac{\alpha}{2} & \cos \frac{\alpha}{2} \end{array} \right), \\
    \Y_\alpha & = \exp\left( -i\alpha \sigma_Y \right) = \left( \begin{array}{cc} \cos \frac{\alpha}{2} & -\sin \frac{\alpha}{2} \\ \sin \frac{\alpha}{2} & \cos \frac{\alpha}{2} \end{array} \right), \\
    \Z_\alpha & = \exp\left( -i\alpha \sigma_Z \right) = \left( \begin{array}{cc} e^{- i \frac{\alpha}{2}} & 0 \\ 0 & e^{i \frac{\alpha}{2}} \end{array} \right).
\end{align*}

\begin{theorem}[{YZY--Euler decomposition, \cite[{pg.\ 189--207}]{Euler}}]\label{YZYEulerTheorem}
Any single-qubit program $U \in PU(2)$ can be expressed as a triple of rotations: \[U = \Y_\alpha \cdot \Z_\beta \cdot \Y_\delta,\] where $0 \le \beta \le \pi$.  (Moreover, for $0 < \beta < \pi$, the values $0 \le \alpha, \delta \le 2 \pi$ are essentially unique.)
\end{theorem}
\begin{proof}
This follows as a special case of the \textit{Cartan decomposition}, which we briefly recount for a generic Lie group $G$.  Such a decomposition rests on the choice of an involution $\theta$ on $G$, referred to as a \textit{Cartan involution} on $G$.  The function $\theta$ breaks the Lie algebra $\mathfrak g$ of $G$ into two parts: the derivative $D_1 \theta$ has nontrivial eigenspaces of weights $1$ and $-1$, denoted $\mathfrak e$ and $\mathfrak o$ respectively, and $\mathfrak g = \mathfrak e \oplus \mathfrak o$.  The Cartan decomposition refers to a corresponding decomposition at the level of Lie \emph{groups}: by selecting a maximal abelian subalgebra $\mathfrak t \le \mathfrak e$ and exponentiating $\mathfrak t$ and $\mathfrak o$ to the Lie subgroups $T$ and $O$, one has \[G = O \cdot T \cdot O = \{o_1 \cdot t \cdot o_2 \mid o_1, o_2 \in O, t \in T\}.\]  In this context, ``abelian'' means $[t_1, t_2] = 0$ for any $t_1, t_2 \in \mathfrak t$.

Rather than simply remit this decomposition to the literature, we make wholly concrete its application to $PU(2)$.  Begin by selecting the involution $\theta_1(U) = U^T$, whose associated eigenspaces $\mathfrak e$ and $\mathfrak o$ are given by
\begin{align*}
\mathfrak e & = \<\sigma_\Z\> \oplus \<\sigma_\X\>, &
\mathfrak o & = \<\sigma_\Y\>.
\end{align*}
Further selecting the maximal abelian subalgebra $\mathfrak t = \<\sigma_\Z\>$, the Cartan decomposition associated to this data asserts that any $U \in PU(2)$ can be decomposed as \[U = \Y_\alpha \cdot \Z_\beta \cdot \Y_\delta.\]

Finally, we turn to the algorithmic determination of these parameter values.  Inspired by this product form, we consider the \textit{Cartan double} of $U$: \[\gamma_1(U) = U \cdot \theta_1(U) = UU^T.\]  In terms of the decomposition, this gives \[\gamma_1(\Y_\alpha \cdot \Z_\beta \cdot \Y_\delta) = \Y_\alpha \cdot \Z_{2\beta} \cdot \Y_{-\alpha}.\]  Indeed, $UU^T$ is a symmetric unitary matrix and hence admits a basis of real eigenvectors.  These can be used to determine the value of $\alpha$, and the eigenvalues of $\gamma_1(U)$ can be used to determine the value of $2\beta$ (and hence $\beta$); and the value of the parameter $\delta$ can then be determined from $\Y_\delta = \Z_{-\beta} \cdot \Y_{-\alpha} \cdot U$.  The claim $0 \le \beta \le \pi$ can be guaranteed by beginning with $-\pi \le \beta \le \pi$ and, in the case $\beta \le 0$ is negative, commuting $\Y_\pi$ through the expression.
\end{proof}

Before continuing on to the two-qubit case, we remark on some variations on this result which are specifically useful to quantum computing.

\begin{remark}\label{ZYZvsYZYRemark}
In the practice of microwave-driven superconducting qubits, there is some preferred rotation---say, $\Z_t$---which is a ``virtual'' operation, implemented as a frame shift, and hence is both instantaneous and immune to device error.  Because of this, it would be preferable to have a variant of the decomposition of \Cref{YZYEulerTheorem} in which $\Z_t$ appears twice, as it would produce circuits with superior execution properties.  One can accomplish this by \textit{conjugation}: given an operator $U$, its conjugation by $Q$ is defined by $U^Q = Q^\dagger U Q$.  This operation enjoys the following properties:
\begin{itemize}
  \item $(UV)^Q = U^Q V^Q$ and $(U^Q)^\dagger = (U^\dagger)^Q$.
  \item $(U^Q)^V = U^{QV}$ and $U = (U^Q)^{Q^\dagger}$.
  \item A \textit{torus} is a connected abelian subgroup of a compact Lie group (e.g., $PU(n)$), and a torus is \textit{maximal} when no larger torus contains it.  For any two maximal tori $T_1$ and $T_2$, there exists an operator $Q$ such that $T_1^Q = T_2$~\cite[Lemma 11.2]{Hall}.\footnote{However, for maximal tori of higher dimension, one may \emph{not} have element-wise control over how conjugation carries one maximal torus into another.}
\end{itemize}
For example, the families $\Y_t$ and $\Z_t$ are both maximal tori in $PU(2)$, and they are conjugate by the operator $Q_1 = \X_{\frac{\pi}{2}}$:
\begin{align*}
\Y_t^{Q_1} & = \Z_t, &
\Z_t^{Q_1} & = \Y_{-t}.
\end{align*}
One can then use $Q_1$ to transform the outer components of \Cref{YZYEulerTheorem} from $\Y_t$ to $\Z_t$, which gives the desired decomposition:
\begin{align*}
U & = (U^{Q_1^\dagger})^{Q_1} = (\Y_\alpha \Z_\beta \Y_\delta)^{Q_1} = \Z_\alpha \Y_{-\beta} \Z_\delta.
\end{align*}
\end{remark}

\begin{remark}
Alternatively, one can apply $Q_1$--conjugation to the set-up of the theorem, rather than its inputs, to reach the same conclusion.  Starting with the Cartan involution $\theta(U) = U^T$, we define a new involution\footnote{There is a similarity relation $\gamma_1^{Q_1}(U) = \gamma_1(U^{Q_1})^{Q_1^\dagger}$.}
\begin{align*}
\theta_1^{Q_1}(U) & = ((U^{Q_1})^T)^{Q_1^\dagger} = Q_1((Q_1^\dagger U Q_1)^T)Q_1^\dagger \\
& = Q_1 Q_1^T U^T (Q_1 Q_1^T)^\dagger = \X_\pi U^T \X_{-\pi} \\
\intertext{and its associated doubling}
\gamma_1^{Q_1}(U) & = U \cdot \theta^{Q_1}(U) = U \X_\pi U^T \X_{-\pi}.
\end{align*}
\end{remark}

We now turn to the analogous structure theorem for two-qubit operators:

\begin{theorem}[``Canonical decomposition'', {\cite[Section III]{KrausCirac}, \cite[Theorem 2]{Makhlin}, \cite[Section III.A.1]{ZVSWGeometric}}]\label{CanonicalDecompositionTheorem}
Any two-qubit unitary operator $U$ admits an expression as
\[U =
\begin{tikzcd}
& \gate{C} & \gate[2]{\CAN(\alpha, \beta, \delta)} & \gate{A} & \qw \\
& \gate{D} & \qw & \gate{B} & \qw
\end{tikzcd},\]
where $A$, $B$, $C$, and $D$ are single-qubit operators, where $\pi/2 \ge \alpha \ge \beta \ge |\delta|$ are certain parameters, and where \[\CAN(\alpha, \beta, \delta) = \exp\left( -i (\alpha \sigma_X^{\otimes 2} + \beta \sigma_Y^{\otimes 2} + \delta \sigma_Z^{\otimes 2}) \right).\]  The parameter values are unique, and for generic parameter values the local gates are also unique.
\end{theorem}
\begin{proof}
As before, the Cartan involution $\theta(U) = U^T$ and associated Cartan doubling $\gamma(U) = UU^T$ give a decomposition of an operator $V$ into a product $O_L D O_R$, where these factors satisfy
\begin{align*}
    O_L^T & = O_L^{-1}, &
    O_R^T & = O_R^{-1}, &
    D^T & = D,
\end{align*}
hence $O_L$ and $O_R$ are orthogonal and $D$ is symmetric, hence diagonal.  As in \Cref{ZYZvsYZYRemark}, the theorem as stated arises by conjugating this decomposition by a particular operator $Q$, \[Q = \frac{1}{\sqrt{2}} \left( \begin{array}{cccc} 1 & 0 & 0 & i \\ 0 & i & 1 & 0 \\ 0 & i & -1 & 0 \\ 1 & 0 & 0 & -i \end{array} \right),\] which satisfies
\begin{align*}
(PU(2)^{\otimes 2})^Q & = PO(4), &
\CAN^Q & = \Delta
\end{align*}
for $\Delta \le PU(4)$ the diagonal matrices.\footnote{Identifying a useful analogue of $Q$ and of $PU(2)^{\otimes 2}$ is the primary inhibitor of generalizing this to higher qubit counts.  See \cite[Proposition IV.3]{SMB} for a list of references concerning the provenance of this operator $Q$.}

To see the constraints on the parameters $\alpha$, $\beta$, and $\delta$, we require an explicit formula for $\CAN^Q$: \[\CAN(\alpha, \beta, \delta)^Q = \exp\left(i \operatorname{diag}\left(\begin{array}{c}\alpha - \beta + \delta \\ \alpha + \beta - \delta \\ -(\alpha + \beta + \delta) \\ -\alpha + \beta + \delta \end{array}\right)\right).\]  This linear system can be solved for any diagonal matrix $\operatorname{diag}(d_1, d_2, d_3, d_4)$ satisfying $d_+ = 0$, and the resulting parameters can furthermore be taken to lie in the range $[-\pi, \pi]$.  Since diagonal operators are stable under conjugation by signed permutation matrices, which are themselves members of $PO(4)$, one may insert such operators into the expression to obtain several equivalent decompositions.  One checks that the effects of these signed permutation matrices are generated by permutations and the operator \[(\alpha, \beta, \delta) \mapsto (\pi - \alpha, \pi - \beta, \delta).\]  It follows that there is a unique representative satisfying the indicated inequality.
\end{proof}

\begin{corollary}\label{CanonicalCoordinatesImplyLocalEquiv}
If \Cref{CanonicalDecompositionTheorem} yields the same canonical parameters for a pair of two-qubit operators $U$ and $V$, then there exist single-qubit operators $A$, $B$, $C$, and $D$ satisfying
\[
\begin{tikzcd}
& \gate[2]{U} & \qw \\
& \qw & \qw
\end{tikzcd}
 =
\begin{tikzcd}
& \gate{C} & \gate[2]{V} & \gate{A} & \qw \\
& \gate{D} & \qw & \gate{B} & \qw
\end{tikzcd}.\]  (In this case, $U$ and $V$ are said to be \textit{locally equivalent}.)
\end{corollary}
\begin{proof}
Apply the Theorem to produce
\[\begin{tikzcd}
& \gate[2]{U} & \qw \\
& \qw & \qw
\end{tikzcd}
=
\begin{tikzcd}
& \gate{C_U} & \gate[2]{\CAN(\alpha, \beta, \delta)} & \gate{A_U} & \qw \\
& \gate{D_U} & \qw & \gate{B_U} & \qw
\end{tikzcd},\]
\[\begin{tikzcd}
& \gate[2]{V} & \qw \\
& \qw & \qw
\end{tikzcd}
 =
\begin{tikzcd}
& \gate{C_V} & \gate[2]{\CAN(\alpha, \beta, \delta)} & \gate{A_V} & \qw \\
& \gate{D_V} & \qw & \gate{B_V} & \qw
\end{tikzcd}.\]
Solving the second equation for the canonical gate gives
\[
\begin{tikzcd}
& \gate{C_V^\dagger} & \gate[2]{V} & \gate{A_V^\dagger} & \qw \\
& \gate{D_V^\dagger} & \qw & \gate{B_V^\dagger} & \qw
\end{tikzcd}
=
\begin{tikzcd}
& \gate[2]{\CAN(\alpha, \beta, \delta)} & \qw \\
& \qw & \qw
\end{tikzcd},\]
and substituting it into the first yields
\[\begin{tikzcd}
& \gate[2]{U} & \qw \\
& \qw & \qw
\end{tikzcd}
 =
\begin{tikzcd}
& \gate{C_U C_V^\dagger} & \gate[2]{V} & \gate{A_V^\dagger A_U} & \qw \\
& \gate{D_U D_V^\dagger} & \qw & \gate{B_V^\dagger B_U} & \qw
\end{tikzcd}. \qedhere\]
\end{proof}

\begin{remark}
As in the single-qubit case, this decomposition is algorithmically effective: given a two-qubit gate $U$, by selecting angle values $\alpha$, $\beta$, and $\delta$ the operator spectrum of $\gamma(U^Q)$ can be made to agree with that of $\CAN(\alpha, \beta, \delta)$; a special-orthogonal matrix diagonalizing $\gamma(U^Q)$ recovers $A$ and $B$; and, from this, one can then solve for $C$ and $D$~\cite[Proposition IV.3]{SMB}.  The keystone of Shende, Bullock, and Markov is a process for manufacturing circuits with low $\CNOT$--count for realizing particular values of $\CAN(\alpha, \beta, \delta)$.  In general, they show that this requires three applications of $\CNOT$, and they moreover show which gates are accessible within two (resp.\ one, resp.\ zero) applications of $\CNOT$s: these are those gates whose canonical parameter $\delta$ is fixed at zero (resp.\ both $\beta$ and $\delta$ are zero, resp.\ all parameters are zero).  With these circuits in hand, they then apply \Cref{CanonicalCoordinatesImplyLocalEquiv}.
\end{remark}

\begin{remark}
In the current practice of quantum computing, single qubit operators typically experience device errors at a rate 1--2 orders of magnitude less than multi-qubit operators, which underscores a compiler designer's desire for decompositions that use two-qubit gates as sparingly as possible.  From this perspective, the canonical decomposition neatly cleaves any two-qubit interaction into outer pieces, which are completely assailable by decomposition into single-qubit gates and which hence are representable with good fidelity, and an inner piece, which is completely unassailable by single-qubit gates, for which a general theory of decomposition into native interactions is not known, and which is especially prone to poor fidelity if a decomposition is inefficient.  In reaction to this, we will focus all of our attention on the manufacture of families of canonical gates with the shortest circuits possible.
\end{remark}

\begin{figure}
\begin{center}
\begin{tabular}{lc}
canonical decomp. & $G = \Z_\alpha \Y_\beta \Z_\delta$ \\
\hline
orthogonal decomp. & $G^{Q_1} = \Y_{-\alpha} \Z_\beta \Y_{-\delta}$ \\
\hline
\begin{tabular}{c}diagonalized \\ Cartan double\end{tabular} & $\gamma^{Q_1}(G) = \Y_{-\alpha} \Z_{2\beta} \Y_{\alpha}$ \\
\hline
canonical parameter & $\beta$
\end{tabular}
\end{center}
\caption{Summary of the one-qubit invariants of \Cref{Geometryof2QSection} and their interrelations}\label{QConjTable1Q}
\end{figure}

\begin{figure}
\begin{center}\setlength{\tabcolsep}{4pt}
\begin{tabular}{lc}
canonical decomp. & \begin{tabular}{c}$G = L_1 \cdot \CAN \cdot L_2$, \\ with $L_1$, $L_2$ local. \end{tabular} \\
\hline
orthogonal decomp. & \begin{tabular}{c}$G^Q = O_1 D O_2$, \\ with $O_j = L_j^Q$ ortho., \\ $D = \CAN^Q$ diagonal. \end{tabular} \\
\hline
\begin{tabular}{c}diagonalized \\ Cartan double\end{tabular} & \begin{tabular}{c}$\gamma^Q(G) = O_1 D^2 O_1^T$, \\ with $O_1 = L_1^Q$ ortho., \\ $D = \CAN^Q$ diagonal. \end{tabular} \\
\hline
canonical parameter & \begin{tabular}{c}spectrum of $D$, or \\ spectrum of $\gamma^Q(G)$, or \\ arguments to $\CAN$. \end{tabular}
\end{tabular}
\end{center}
\caption{Summary of the two-qubit invariants of \Cref{Geometryof2QSection} and their interrelations}\label{QConjTable2Q}
\end{figure}

\section{The multiplicative eigenvalue problem and the monodromy polytope}\label{ReductionSection}

Shende, Bullock, and Markov's description of those two-qubit programs accessible within a fixed number of applications of $\CNOT$ relies on specific commutation relations and explicit computation (cf.\ \Cref{LeakyEntanglersApdx}).  We will study a more general version of this same question:

\begin{problem}\label{2QMultiDeckerProblem}
Let $\S$ be a set of two-qubit gates.
\begin{enumerate}
  \item Describe the subspace $P^n_\S \subseteq PU(4)$,
  \[\left\{
\begin{tikzcd}[column sep=1em]
& \gate{A_n} & \gate[2]{S_n} & \qw \cdots & \gate{A_1} & \gate[2]{S_1} & \gate{A_0} & \qw \\
& \gate{B_n} & \qw & \qw \cdots & \gate{B_1} & \qw & \gate{B_0} & \qw
\end{tikzcd} \right\},\]
for $S_j \in \S$ and $A_j, B_j \in PU(2)$. \label{2QMultiDeckerProblem1}
  \item Given $G \in P^n_{\S}$, algorithmically produce local gates $A_j, B_j$ which realize $G$ as such a circuit.\label{2QMultiDeckerProblem2}
\end{enumerate}
\end{problem}

\noindent A solution to the entirety of \Cref{2QMultiDeckerProblem} would yield a depth-optimal compilation algorithm targeting $\S$: given a two-qubit program $U$, one could compute its canonical parameters, use \Cref{2QMultiDeckerProblem}.\ref{2QMultiDeckerProblem1} to discern the minimal $P^n_{\S}$ to which they belong, use \Cref{2QMultiDeckerProblem}.\ref{2QMultiDeckerProblem2} to manufacture a circuit of the prescribed depth and with the prescribed canonical parameters, and finally use \Cref{CanonicalCoordinatesImplyLocalEquiv} to produce a circuit modeling $U$.  A solution only to \Cref{2QMultiDeckerProblem}.\ref{2QMultiDeckerProblem1} would yield a method for computing the following interesting value:

\begin{definition}\label{ExpectedCircuitDepthDefn}
Let $\L_\S(U)$ be the minimum number of two-qubit gates appearing in any circuit implementation of $U$ using the gate set $\S$.  We define the \textit{expected circuit depth of $\S$} to be
\begin{align*}
\<\L_\S\>^\Haar & = \int_{U \in PU(4)} \L_\S(U) d\mu \\
& = \sum_{n=0}^\infty n \cdot \vol^\Haar \left(\L_\S^{-1}(n)\right) \\
& = \sum_{n=0}^\infty n \cdot \vol^\Haar \left(P^n_\S \setminus \bigcup_{j=0}^{n-1} P^j_\S \right).
\end{align*}
This value measures the efficiency of the gate set $\S$: the smaller $\<\L_\S\>^\Haar$, the more efficient $\S$ is at encoding programs into circuits.
\end{definition}

Sufficiently enticed, we begin with the first nontrivial case where $n = 2$ and where $S_1$ and $S_2$ are fixed:

\begin{problem}\label{2QSandwichProblem}
Let $E$ and $F$ be fixed two-qubit gates.
\begin{enumerate}
    \item Describe the subspace $P \subseteq PU(4)$,
    \[\left\{\begin{tikzcd}
    & \gate{A_3} & \gate[2]{F} & \gate{A_2} & \gate[2]{E} & \gate{A_1} & \qw \\
    & \gate{B_3} & \qw & \gate{B_2} & \qw & \gate{B_1} & \qw
    \end{tikzcd}\right\},\]
    for $A_1$, $A_2$, $A_3$, $B_1$, $B_2$, $B_3 \in PU(2)$. \label{2QSandwichProblem1}
    \item Given $G \in P$, algorithmically produce local gates $A_1$, $A_2$, $A_3$, $B_1$, $B_2$, and $B_3$ which realize $G$ as such a circuit. \label{2QSandwichProblem2}
\end{enumerate}
\end{problem}

\noindent In this section, we show that this reduces to a well-known problem in representation theory, the \textit{multiplicative eigenvalue problem}, whose solution comes in the form of the \textit{monodromy polytope}.

Using \Cref{CanonicalCoordinatesImplyLocalEquiv}, we can discern whether a given program $G$ belongs to $P$ by computing all of the canonical parameters of the programs in $P$ and checking whether that set includes those of $G$.  In order to make use of this, one would require a compact description of this set of parameters.  To begin exploring whether this is feasible, let us go through the motions of computing the canonical parameters of $G \in P$---or, equivalently, computing the spectrum of the Cartan double $\gamma(G^Q)$.  Begin by applying $Q$--conjugation to the supposed decomposition of $G$ to get $G^Q = O_1 E^Q O_2 F^Q O_3$, where each $O_j$ is the orthogonal matrix $Q$--conjugate to the local gate $(A_j \otimes B_j)$.  In turn, $E^Q$ and $F^Q$ have orthogonal decompositions:
\begin{align*}
    E^Q & = O_{E,L} D_E O_{E,R}, &
    F^Q & = O_{F,L} D_F O_{F,R},
\end{align*}
where $D_E$ and $D_F$ are diagonal and $O_{E,L}$, $O_{E,R}$, $O_{F,L}$, and $O_{F,R}$ are all orthogonal.  Combining these decompositions yields \[G^Q = O_1 O_{E,L} D_E O_{E,R} O_2 O_{F,L} D_F O_{F,R} O_3,\] from which we compute  
\begin{align*}
    \gamma(G^Q) & = O_1 O_{E,L} D_E O_{E,R} O_2 O_{F,L} D_F O_{F,R} O_3 \cdot \\
    & \quad \quad (O_1 O_{E,L} D_E O_{E,R} O_2 O_{F,L} D_F O_{F,R} O_3)^T \\
    & = O_1 O_{E,L} D_E O_{E,R} O_2 O_{F,L} D_F^2 \cdot \\
    & \quad \quad O_{F,L}^T O_2^T O_{E,R}^T D_E O_{E,L}^T O_1^T \\
    & \sim D_E^2 O D_F^2 O^T,
\end{align*}
where $O = O_{E,R} O_2 O_{F,L}$ is an orthogonal operator and $\sim$ denotes unitary similarity.  Altogether, this reduces \Cref{2QSandwichProblem} to the following:

\begin{problem}\label{OrthogonalMEP}
Let $D_E$ and $D_F$ be fixed diagonal special unitary matrices.
\begin{enumerate}
    \item Calculate the possible spectra of operators of the form $D_E^2 O D_F^2 O^T$ as $O$ ranges over orthogonal matrices. \label{OrthogonalMEP1}
    \item Given a particular such operator spectrum $D_G$, calculate an orthogonal matrix $O$ such that $D_E^2 O D_F^2 O^T$ diagonalizes to give $D_G^2$. \label{OrthogonalMEP2}
\end{enumerate}
\end{problem}
\begin{proof}[Description of reduction]
The canonical parameters of $G$ are equivalent data to the spectrum of $\gamma(G^Q)$.  Because $\gamma(G^Q)$ is similar to $D^2_E O D^2_F O^T$ and because operator similarity preserves spectra, it is equivalent to compute the spectrum of $D^2_E O D^2_F O^T$.  As $G$ ranges over $P$ (i.e., as the local gates in the circuit vary), the terms $O$ range over $PO(4)$.

For the second part, suppose that we can construct an $O$ so that $D^2_E O D^2_F O^T$ has a prescribed operator spectrum.  We may then work backward:
\begin{align*}
& D^2_E O D^2_F O^T \\
& \sim D_E O D^2_F O^T D_E \\
& \sim O_{E,L} D_E O D^2_F O^T D_E O_{E,L}^T \\
& = O_{E,L} D_E O_{E,R} (O_{E,R}^T O O_{F,R}^T) O_{F,R} D_F O_{F,L} \cdot \\
& \quad \quad (O_{E,L} D_E O_{E,R} (O_{E,R}^T O O_{F,R}^T) O_{F,R} D_F O_{F,L})^T \\
& = \gamma(O_{E,L} D_E O_{E,R} (O_{E,R}^T O O_{F,R}^T) O_{F,R} D_F O_{F,L}) \\
& = \gamma^Q((1 \otimes 1) E (A_2 \otimes B_2) F (1 \otimes 1)),
\end{align*}
where the local gate $A_2 \otimes B_2$ is determined by the formula \[(A_2 \otimes B_2)^Q = O_{E,R}^T O O_{F,R}^T.\]  The circuit
\[\begin{tikzcd}
    & \gate[2]{F} & \gate{A_2} & \gate[2]{E} & \qw \\
    & \qw & \gate{B_2} & \qw & \qw
\end{tikzcd}\]
then has the prescribed canonical parameters.
\end{proof}

\Cref{OrthogonalMEP} is a restricted instance of the multiplicative eigenvalue problem:

\begin{problem}[{\cite{AgnihotriWoodward}}]\label{UnitaryMEP}
Let $U_1$, $U_2$, and $U_3$ be unitary operators.
\begin{enumerate}
    \item \textit{Multiplicative eigenvalue problem}: Describe the possible spectra of all triples $U_1$, $U_2$, $U_3$ satisfying $U_1 U_2 U_3 = 1$.\footnote{There are also variants of the multiplicative eigenvalue problem concerning products of any fixed length $k \ge 0$.} \label{UnitaryMEP1}
    \item \textit{Effective saturation problem}: Given a triple of operator spectra satisfying the conditions of \Cref{UnitaryMEP}.\ref{UnitaryMEP1}, algorithmically produce $U_1$, $U_2$, and $U_3$ realizing these spectra and satisfying $U_1 U_2 U_3 = 1$. \label{UnitaryMEP2}
\end{enumerate}
\end{problem}

What is immediately clear is that the collection of spectral quadruples satisfying \Cref{OrthogonalMEP}.\ref{OrthogonalMEP1} can be converted to satisfy \Cref{UnitaryMEP}.\ref{UnitaryMEP1}.  Supposing that we have suitable operators $D_E$, $D_F$, $D_G$, and $O$ satisfying the similarity relation above, there must exist a unitary $U$ satisfying
\begin{align*}
    \gamma^Q(G) = O_{G, L} D_G^2 O_{G, L}^T & \sim D_E^2 O D_F^2 O^T \\
    U^\dagger D_G^2 U & = D_E^2 O D_F^2 O^T,
\end{align*}
which can be rewritten as
\begin{align*}
    1 & = D_E^2 (O D_F^2 O^T) (U^\dagger (D_G^2)^\dagger U) \\
    & = U_1 U_2 U_3.
\end{align*}
This shows that the spectra of $D_E^2$ and $D_F^2$ agree with those of $U_1$ and $U_2$, and the spectrum of $D_G^2$ agrees with that of $U_3^\dagger$.  What is not obvious is that the reverse implication is also true: given $U_1$, $U_2$, $U_3$ satisfying $1 = U_1 U_2 U_3$, we would like to know whether it is always possible to factor out orthogonal matrices and rearrange terms to produce a sentence of the form \[O_{G, L} D_G^2 O_{G, L}^T \sim D_E^2 O D_F^2 O^T.\]  This turns out to be so, which is a nontrivial result in symplectic geometry:

\begin{theorem}[{\cite[Theorem 3]{FalbelWentworth}}]
Suppose that $D_1$, $D_2$, $D_3$ are diagonal operators which satisfy the conditions of \Cref{UnitaryMEP}.\ref{UnitaryMEP1}.\footnote{That is, they appear as the diagonalizations of unitaries $U_1$, $U_2$, $U_3$ satisfying $U_1 U_2 U_3 = 1$.}  Then there exist operators $V_1$, $V_2$, $V_3$ which diagonalize to $D_1$, $D_2$, $D_3$, which satisfy $V_1 V_2 V_3 = 1$, and for which $V_1$ and $V_2$ are orthogonally diagonalizable, i.e.,
\begin{align*}
V_1 & = O_1^T D_1 O_1, &
V_2 & = O_2^T D_2 O_2,
\end{align*}
for some orthogonal matrices $O_1$, $O_2$. \qed
\end{theorem}

\begin{corollary}\label{ReductionToOrthogonalLemma}
For $U_1$, $U_2$, and $U_3$ unitaries satisfying $U_1 U_2 U_3 = 1$ and with diagonalizations $D_E^2$, $D_F^2$, and $(D_G^2)^\dagger$ respectively, there exists an orthogonal $O$ and a unitary $U$ such that \[U^\dagger D_G^2 U = D_E^2 O D_F^2 O^T,\] i.e., solutions to \Cref{UnitaryMEP} give rise to solutions to \Cref{OrthogonalMEP}.
\end{corollary}
\begin{proof}
Start by applying the Theorem:
\begin{align*}
    1 & = V_1 V_2 V_3 \\
    & = (O_1^T D_E^2 O_1) (O_2^T D_F^2 O_2) V_3 \\
    O_1 V_3^\dagger O_1^T & = D_E^2 O_1 O_2^T D_F^2 O_2 O_1^T \\
    U^\dagger D_G^2 U & = D_E^2 O D_F^2 O^T,
\end{align*}
where we have written $O = O_1 O_2^T$ and $U = E O_1^T$ for $E$ a matrix with columns an eigenbasis of $V_3^\dagger$.
\end{proof}

\Cref{UnitaryMEP}.\ref{UnitaryMEP1} has been completely resolved by Agnihotri, Meinrenken, and Woodward.  Because even the statement of their result is quite technical, the newcomer may find it too opaque, and so we pause to first investigate some simple cases by hand.  We hope that this serves to motivate the full statement of \Cref{MainAWBTheorem} to follow.

\begin{example}[{\cite[Proposition 3.1]{JeffreyWeitsman}}]\label{Manual1QExample}
Let us return to the setting of \Cref{ZYZvsYZYRemark}, where we we hypothesized that $\Z$--rotations were less expensive to implement in hardware than $\Y$--rotations, and were thus led to study a decomposition of a single-qubit operator $U$ into the form $U = \Z_\epsilon \Y_\zeta \Z_\eta$.  Let us now further suppose that the supply of $\Y$--rotations is limited: we are only able to implement those whose parameters are drawn from a fixed set.  We would then be interested to know, given a program $U$ with canonical decomposition
\[\begin{tikzcd}[ampersand replacement=\&] \& \gate{U} \& \qw \end{tikzcd} = \begin{tikzcd}[ampersand replacement=\&] \& \gate{\Z_\eta} \& \gate{\Y_\zeta} \& \gate{\Z_\epsilon} \& \qw \end{tikzcd},\]
what are the conditions on $\zeta$ such that $U$ admits expression as
\[\begin{tikzcd}[ampersand replacement=\&] \& \gate{U} \& \qw \end{tikzcd} = \begin{tikzcd}[ampersand replacement=\&] \& \gate{\Z_{\eta'}} \& \gate{\Y_\delta} \& \gate{\Z_\beta} \& \gate{\Y_\alpha} \& \gate{\Z_{\epsilon'}} \& \qw \end{tikzcd}\]
with $\alpha$ and $\delta$ drawn from the fixed set of parameters?  To address this, we study the characteristic polynomial of the Cartan double $\gamma^{Q_1}(U)$, as computed from the canonical decomposition of $U$ and from its putative expression as a circuit with constrained $\Y$s:
\begin{align*}
    \chi(\gamma^{Q_1}(U)) & = z^2 - \tr(\gamma^{Q_1}(U)) z + 1, \\
    \tr(\gamma^{Q_1}(U)) & = 2 (\cos \alpha \cos \delta - \cos \beta \sin \alpha \sin \delta), \\
    \tr(\gamma^{Q_1}(U)) & = 2 \cos \zeta,
\end{align*}
hence \[\zeta = \cos^{-1} \left(\cos \alpha \cos \delta - \cos \beta \sin \alpha \sin \delta\right).\]  The role of $\beta$ is thus to modulate the interference of $\alpha$ and $\delta$, and it imposes the following linear inequality on $\epsilon$: \[|\alpha - \delta| \le \zeta \le \pi - |\alpha + \delta - \pi|.\]  However, the precise dependence of $\zeta$ on $\beta$ is decidedly nonlinear.\footnote{As a consequence, we recover the familiar fact that $\alpha = \delta = \pi/2$ generates all rotations with such a circuit.}
\end{example}

\begin{example}\label{Manual2QExample}
In extremely favorable situations, the two-qubit case can also be manually analyzed.  Let us consider the gate family \[\SWAP_t = \exp\left(-2 \pi i t(\sigma_\X^{\otimes 2} + \sigma_\Y^{\otimes 2} + \sigma_\Z^{\otimes 2}) \right),\] for which $\gamma^Q(\SWAP_t)$ has spectrum \[\left(e^{-i\frac{t}{4}}, e^{-i\frac{t}{4}}, e^{-i\frac{t}{4}}, e^{i\frac{3t}{4}}\right),\] and let us consider the instance of \Cref{2QSandwichProblem}.\ref{2QSandwichProblem1} for $E = \SWAP_{t_1}$ and $F = \SWAP_{t_2}$.  Using the reduction to \Cref{UnitaryMEP}, we will describe the possible spectra for operators of the form
\begin{align*}
U_1 & = V^\dagger D_{t_1}^2 V, &
U_2 & = D_{t_2}^2, &
U_3 & = U_1 U_2,
\end{align*}
where $D_{t_j}^2$ is the diagonal matrix formed from the spectrum of $\gamma^Q(\SWAP_{t_j})$ and $V$ is an arbitrary unitary.

Our first observation is that the eigenspace of weight $\exp(-i\frac{t_j}{4})$ of $U_j$ is $3$--dimensional, and hence the intersection of these two eigenspaces for $U_1$ and $U_2$ is at least $2$--dimensional.  It follows that there exists an orthonormal $2$--frame $\{b_1, b_2\}$ of eigenvectors for $U_3$: \[U_3 b_i = U_1 U_2 b_i = U_1 e^{-i\frac{t_2}{4}} b_i = e^{-i\frac{t_1 + t_2}{4}} b_i =: \lambda b_i.\]  In particular, we find $U_3$ to be block-diagonal in this partial basis: no matter what full orthonormal basis $\{b_1, b_2, b_3, b_4\}$ is chosen, we have \[U_3 = \left( \begin{array}{cccc} \lambda \\ & \lambda \\ & & * & * \\ & & * & * \end{array} \right)\] in this basis.

It will be convenient to choose the full orthonormal basis so that $b_3$ spans the complement of $\{b_1, b_2\}$ in the $3$--dimensional eigenspace for $U_2$ and $b_4$ spans the $1$--dimensional eigenspace.  In this basis, the subblock governing the action of $U_3$ on $B = \operatorname{span}\{b_3, b_4\}$ becomes
\begin{align*}
(U_3)|_B & = \left(\begin{array}{cc} e^{-i\frac{t_1}{4}} & 0 \\ 0 & e^{i\frac{3t_1}{4}}\end{array}\right)^W \left(\begin{array}{cc} e^{-i\frac{t_2}{4}} & 0 \\ 0 & e^{i\frac{3t_2}{4}}\end{array}\right) \\
& = \lambda^{-1} \left(\begin{array}{cc} e^{-i\frac{t_1}{2}} & 0 \\ 0 & e^{i\frac{t_1}{2}}\end{array}\right)^W \left(\begin{array}{cc} e^{-i\frac{t_2}{2}} & 0 \\ 0 & e^{i\frac{t_2}{2}}\end{array}\right)
\end{align*}
for some auxiliary unitary matrix $W$.  The resulting operator spectra are shifts by $\lambda^{-1}$ of those those studied in \Cref{Manual1QExample}, hence subject to the same inequalities also shifted by $\lambda^{-1}$.\footnote{We leave it to the reader to imagine how complex the elementary analysis of the generic case (i.e., without degenerate eigenspaces) can become.}
\end{example}

We now turn to the vocabulary needed to enunciate the full solution to \Cref{UnitaryMEP}.\ref{UnitaryMEP1} given by Agnihotri, Meinrenken, and Woodward.  Based on our experience with the linear inequality families above, there is a particular presentation of operator spectra which we are likely to find useful:
\begin{definition}[{cf.\ \cite[Chapter 4]{Humphreys}}]\label{LogSpecDefn}
For a special unitary matrix $U$, we may uniquely present its spectrum \[\Spec U = (e^{2 \pi i a_1}, \ldots, e^{2 \pi i a_n})\] as \[\LogSpec U = (a_1, \ldots, a_n),\] where
\begin{align*}
a_1 \ge \cdots \ge a_n & \ge a_1 - 1, &
a_+ & := \sum_j a_j = 0.    
\end{align*}
We refer to the collection of all such $n$--tuples as $\A$, the \textit{fundamental alcove} of $SU(n)$, and we will write $\LogSpec U$ for the associated point in $\A$.
\end{definition}

The above definition is standard for special unitary matrices, but we will also require the following nonstandard modification:

\begin{definition}\label{LogSpecDefnForSU4byC2}
Let $C_2 \le SU(4)$ be the finite central subgroup $\{\pm 1\}$,\footnote{One could make a similar definition for $SU(n) / C_m$ with $m \mid n$, but we will need only this particular case.} and let $U \in SU(4) / C_2$ be a member of the quotient group, which we may present as a coset \[\{\tilde U, -\tilde U\} \subset SU(4).\]  The logarithmic spectra of these matrices
\begin{align*}
a_* & = \LogSpec \tilde U, &
b_* & = \LogSpec (- \tilde U)    
\end{align*}
are related by a form of rotation:
\begin{align*}
a_* & = \rho\left( b_* \right) \\
& := \left(b_3 + \frac{1}{2}, b_4 + \frac{1}{2}, b_1 - \frac{1}{2}, b_2 - \frac{1}{2}\right).    
\end{align*}
By appropriately picking either $\LogSpec U = \LogSpec \tilde U$ or $\LogSpec -\tilde U$, we see that we may uniquely specify a sequence $\LogSpec U$ which further satisfies either \[(\LogSpec U)_3 + 1/2 > (\LogSpec U)_1\] or \[\left\{ \begin{array}{c}(\LogSpec U)_3 + 1/2 = (\LogSpec U)_1 \\ \text{and} \\ (\LogSpec U)_4 + 1/2 \le (\LogSpec U)_2\end{array}\right\},\] where $(\LogSpec U)_j$ denotes the $j$\textsuperscript{th} component of the quadruple $\LogSpec U$.  We similarly refer to the collection of all such quadruples as $\A_{C_2}$, the fundamental alcove of $SU(4) / C_2$.
\end{definition}

\begin{remark}
\Cref{LogSpecDefn} is the natural target of the logarithmic spectrum of a special unitary operator, and it forms a convex polytope.  \Cref{LogSpecDefnForSU4byC2} is useful because it is the natural target of the logarithmic spectrum of the Cartan double of a \emph{projective} unitary operator: \[\LogSpec \gamma(-) \co PU(4) \to \A_{C_2}.\]  However, it does not quite form a convex polytope: the closure $\overline{\A_{C_2}}$ is a convex polytope, but the values $a_*$ satisfying $a_3 + 1/2 = a_1$ and $a_4 + 1/2 > a_2$, which constitute half a face of $\overline{\A_{C_2}}$, are missing from $\A_{C_2}$.
\end{remark}

\begin{definition}
In line with the program outlined above, we will be especially interested in the assignment
\begin{align*}
\LogCoords\co PU(4) & \to \A_{C_2}, \\
U & \mapsto \LogSpec \gamma(U^Q),
\end{align*}
which we abbreviate to $\LogCoords(U)$.
\end{definition}

\begin{definition}
The extremal points of the polytope $\overline{\A_{C_2}}$ lie at the following coordinates:
\begin{align*}
    \LogCoords(I) & = e_1 = \left(0, 0, 0, 0\right), \\
    \LogCoords(\CZ) & = e_2 = \left(\frac{1}{4}, \frac{1}{4}, -\frac{1}{4}, -\frac{1}{4}\right), \\
    \LogCoords(\ISWAP) & = e_3 = \left(\frac{1}{2}, 0, 0, -\frac{1}{2}\right), \\
    \LogCoords(\SWAP) & = e_4 = \left(\frac{1}{4}, \frac{1}{4}, \frac{1}{4}, -\frac{3}{4}\right), \\
    \LogCoords(\sqrt{\SWAP}) & = e_5 = \left(\frac{3}{8}, \frac{3}{8}, -\frac{1}{8}, -\frac{5}{8}\right), \\
    \rho(e_5) & = e_6 = \left(\frac{3}{8}, -\frac{1}{8}, -\frac{1}{8}, -\frac{1}{8}\right).
\end{align*}
The subspace $\A_{C_2}$ of $\overline{\A_{C_2}}$ is given by deleting the subspace of convex combinations of $e_2$, $e_3$, and $e_6$ in which $e_6$ carries a nonzero coefficient.
\end{definition}

\begin{definition}
For $r, k > 0$ be positive integers, let $\Q_{r,k}$ denote the following set of sequences: \[\Q_{r,k} = \{(I_1, \ldots, I_r) \in \mathbb Z^r \mid k \ge I_1 \ge \cdots \ge I_r \ge 0\}.\]
\end{definition}

We are now in a position to state the solution to \Cref{UnitaryMEP}.\ref{UnitaryMEP1}.  We will give a more complete exposition of this result, its context, and its proof in \Cref{AWBSection}, but for our intended application we need only the following statement:

\begin{theorem}[see \Cref{AWBTheoremRephrased}]\label{MainAWBTheorem}
\pushQED{\qed}
Let $U_1$, $U_2$, $U_3 \in SU(n)$ satisfy $U_1 U_2 = U_3$, and let
\begin{align*}
    \alpha_* & = \LogSpec U_1, \\
    \beta_*  & = \LogSpec U_2, \\
    \delta_* & = \LogSpec U_3
\end{align*}
be the associated logarithmic spectra.

Select $r, k > 0$ satisfying $r+k = n$, select $a, b, c \in \Q_{r,k}$, and take $d \ge 0$ so that the associated \textit{quantum Littlewood--Richardson coefficient} $N_{ab}^{c,d}(r,k)$ (see \Cref{QuantumLittlewoodRichardsonCoeffs} and \Cref{QLWCoeffsForC2}) satisfies $N_{ab}^{c,d}(r,k) = 1$.  The following inequality then holds:
\[
d - \sum_{i=1}^r \alpha_{k+i-a_i} - \sum_{i=1}^r \beta_{k+i-b_i} + \sum_{i=1}^r \delta_{k+i-c_i} \ge 0.
\tag{*}\label{Eqn}
\]
We define the \textit{monodromy polytope} (for $SU(n)$) to be the polytope determined by all such inequalities.

Conversely, given alcove sequences $\alpha_*$, $\beta_*$, $\delta_*$ which belong to the monodromy polytope, then there must exist $U_1$, $U_2$, $U_3$ satisfying $U_1 U_2 = U_3$ and
\begin{align*}
    \alpha_* & = \LogSpec U_1, \\
    \beta_* & = \LogSpec U_2, \\
    \delta_* & = \LogSpec U_3. \qed
\end{align*}
\end{theorem}

Before proceeding to use this result in any complex way, we show how it can be used to recover the contents of \Cref{Manual1QExample}.\footnote{We will re-do \Cref{Manual2QExample} as part of \Cref{SectionCPHASEandPSWAP}.}

\begin{example}[{cf.\ \Cref{Manual1QExample}}]\label{AWB1QExample}
In \Cref{QLWCoeffsForC2}, we provide a table of quantum Littlewood--Richardson coefficients relevant for $PU(2)$.
\begin{figure}
    \begin{center}
        \begin{tabular}{ccccccc}
$r$ & $k$ & $a$ & $b$ & $c$ & $d$ & $N_{ab}^{c,d}(r,k)$ \\
\hline
$1$ & $1$ & $(0)$ & $(0)$ & $(0)$ & $0$ & $1$ \\
        & & $(1)$ & $(0)$ & $(1)$ & $0$ & $1$ \\
              & & & $(1)$ & $(0)$ & $1$ & $1$
        \end{tabular}
    \end{center}
    \caption{Structure constants in $qH^* \mathrm{Gr}(1, 1)$.  There is a further symmetry relation $N_{ab}^{c,d}(r,k) = N_{ba}^{c,d}(r,k).$}\label{QLWCoeffsForC2}
\end{figure}
Coupling these to \Cref{MainAWBTheorem} then gives the following family of inequalities:
\begin{align*}
    \delta_2 & \ge \alpha_2 + \beta_2, &
    \delta_1 & \ge \alpha_1 + \beta_2, \\
    \delta_2 & \ge -1 + \alpha_1 + \beta_1, &
    \delta_1 & \ge \alpha_2 + \beta_1, \\
\end{align*}
or, in terms of $\alpha_1, \beta_1, \delta_1 \ge 0$ alone,
\begin{align*}
    \delta_1 & \le \alpha_1 + \beta_1, &
    \delta_1 & \ge \alpha_1 + -\beta_1, \\
    \delta_1 & \le 1 - (\alpha_1 + \beta_1), &
    \delta_1 & \ge -\alpha_1 + \beta_1.
\end{align*}
The resulting polytope is portrayed in \Cref{SU2Polytope}.  Isolating $\delta_1$---or, equivalently, studying a vertical ray in the Figure above a particular choice of $(\alpha_1, \beta_1)$---yields \[|\alpha_1 - \beta_1| \le \delta_1 \le \frac{1}{2} - \left|\alpha_1 + \beta_1 - \frac{1}{2}\right|.\]  A geometrical interpretation of this restriction is shown in \Cref{SU2PolytopeRestricted}.

\begin{figure}
    \centering
    \includegraphics[width=0.35\textwidth]{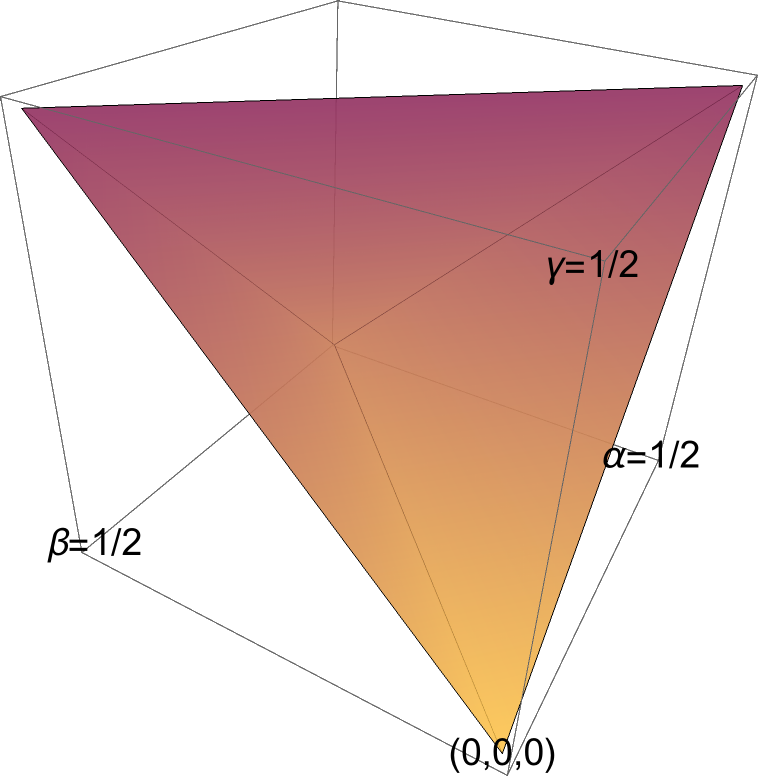}
    \caption{Full monodromy polytope for $SU(2)$, as in \Cref{AWB1QExample}, shaded along the coordinate $\delta_1$.}
    \label{SU2Polytope}
\end{figure}

\begin{figure}
    \centering
    \includegraphics[width=0.35\textwidth]{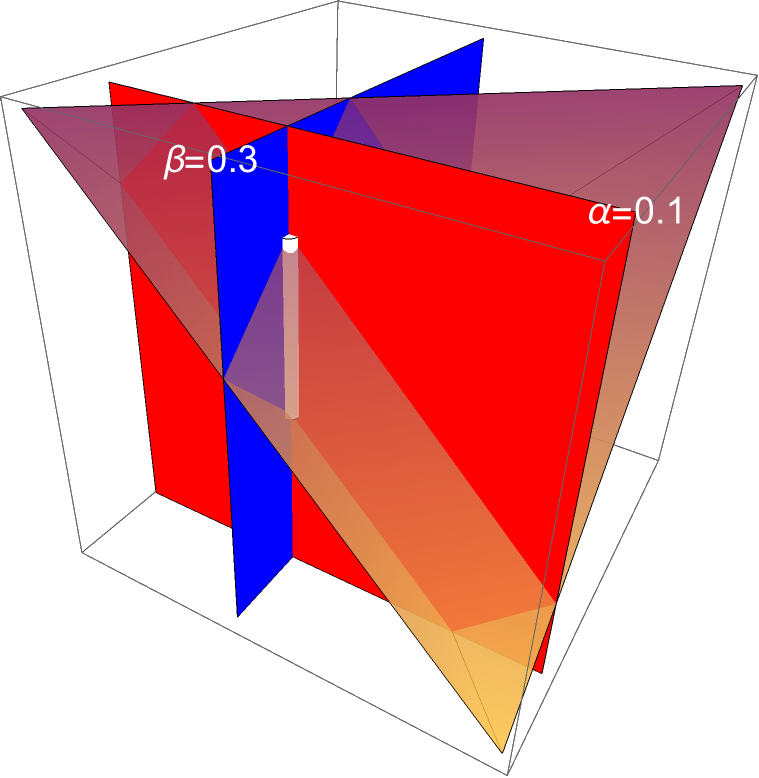}
    \caption{The polytope of \Cref{AWB1QExample} and \Cref{SU2Polytope} intersected with the planes $\alpha_1 = 0.1$, shaded red, and $\beta_1 = 0.3$, shaded blue, resulting in restrictions on $\delta_1$, shaded white.}
    \label{SU2PolytopeRestricted}
\end{figure}
\end{example}

In our study of two-qubit programs, we have already announced the importance of the composite \[\LogCoords\co PU(4) \xrightarrow{\gamma^Q} SU(4) / C_2 \xrightarrow{\LogSpec} \A_{C_2},\] and we will accordingly want to employ a variant of \Cref{MainAWBTheorem} that applies to $SU(4) / C_2$.

\begin{corollary}[{cf. \Cref{CentralExtensionsRemark}}]\label{AGWForSU4ModC2}
\Cref{MainAWBTheorem} holds for $SU(4) / C_2$, under the condition that the indicated family of inequalities all simultaneously hold either for $\delta = \LogSpec U_3$ or for $\delta = \rho(\LogSpec U_3)$.
\end{corollary}
\begin{proof}
If $U_1$, $U_2$, and $U_3$ are elements of $SU(4)$ such that $U_1 U_2 \equiv U_3$ in $SU(4) / C_2$, then there is some $j$ so that $U_1 U_2 = (-1)^j U_3$ as elements of $SU(4)$.  If $j$ is even, then \Cref{MainAWBTheorem} applies directly to give the same statement.  If $j$ is odd, then \Cref{MainAWBTheorem} applies with $U_3$ replaced by $-U_3$, which has the effect of replacing $\delta$ by $\LogSpec(-U_3) = \rho(\LogSpec U_3)$.
\end{proof}

Additionally, not only do \Cref{MainAWBTheorem} and \Cref{AGWForSU4ModC2} resolve \Cref{2QSandwichProblem}.\ref{2QSandwichProblem1}, but its form is sufficiently polite that we can also bootstrap it into a solution to \Cref{2QMultiDeckerProblem}.\ref{2QMultiDeckerProblem1}:

\begin{corollary}\label{PnAreAllPolytopes}
Let $\S$ be a gate set whose image through $\LogCoords$ is a (finite) union of convex polytopes.  The image of $P_\S^n$ through $\LogCoords$ is then also a (finite) union of convex polytopes.
\end{corollary}
\begin{proof}
We have assumed the base case: $\LogCoords(P^1_\S)$ is a union of convex polytopes.  Assuming that $\LogCoords(P^{n-1}_\S)$ is a union of convex polytopes, each constituent polytope is described by a finite collection of linear inequalities.  The monodromy polytope is itself also described by a finite collection of linear inequalities.  Select polytope constituents of $\LogCoords(P^1_\S)$ and of $\LogCoords(P^{n-1}_\S)$.  By imposing those linear inequalities describing the constituent of $\LogCoords(P^1_\S)$ on the first coordinate, imposing those linear inequalities describing the constituent of $\LogCoords(P^{n-1}_\S)$ on the second coordinate, and using Fourier--Motzkin elimination to project to the final coordinate, we produce a subset of $\LogCoords(P^n_\S)$ which is described by a finite collection of linear inequalities.  It follows that this too is a convex polytope, and the entire set $\LogCoords(P^n_\S)$ is the union of the convex polytopes formed in this way.
\end{proof}

\begin{remark}\label{NestingPolytopes}
For gate sets $\S$ of interest to us, it is often the case that $\S$ appears as a subset of $P_\S^2$.  This condition entails the nesting properties $P_\S^n \subseteq P_\S^{n+1}$ and $\LogCoords(P_\S^n) \subseteq \LogCoords(P_\S^{n+1})$ for $n \ge 2$.
\end{remark}

\begin{definition}\label{DefnExpectedAC2}
Recall from \Cref{ExpectedCircuitDepthDefn} the gate set quality metric $\<\L_\S\>^\Haar$.  By \Cref{PnAreAllPolytopes}, we now know the sets $\LogCoords(P^n_\S)$ to all be unions of polytopes, which have easily calculable volumes.  Inspired by this, and using the fact that $\L_\S$ is constant on fibers of $\LogCoords$ (i.e., \Cref{CanonicalCoordinatesImplyLocalEquiv}), we propose the following alternative definition:
\begin{align*}
\<\L_\S\>^{\A_{C_2}} & = \int_{x \in \A_{C_2}} \L_\S(\LogCoords^{-1}(x)) \; dx \\
& = \sum_{n=0}^\infty n \cdot \vol^{\A_{C_2}} \left( \LogCoords \left(\L_\S^{-1}(n) \right) \right) \\
& = \sum_{n=0}^\infty n \cdot \vol^{\A_{C_2}} \left(\LogCoords(P^n_\S) \setminus \bigcup_{j=0}^{n-1} \LogCoords(P^j_\S) \right),
\end{align*}
where $\vol^{\A_{C_2}}$ indicates the Euclidean volume as a subset of the polytope $\A_{C_2}$, normalized so that $\A_{C_2}$ itself has unit volume.  These are not generally equal: their difference is mediated by the Jacobian of $\LogCoords$, considered as a function on the subset of canonical gates, as well as the sizes of the fibers of $\LogCoords$.  The quantity $\<\L_\S\>^{\Haar}$ has the advantage of capturing a more traditional notion of average, but $\<\L_\S\>^{\A_{C_2}}$ has the advantage of being reliably computable by finite means.

Due to their similar definitions, it is often the case that we can make claims about both quantities simultaneously.  In these situations where we intend to make a statement about both, we will omit the superscript.
\end{definition}

\section{Monodromy polytope slices for the standard gates}\label{MonodromyOfStdGatesSection}

In this section, we consider the ``standard gates'' that appear in the paper of Smith, Curtis, and Zeng~\cite[Appendix A]{SCZ} and their effect as members of a native gate set.  Each of these gates or gate families specify via $\LogCoords$ either a particular point in $\A_{C_2}$ or a line segment in $\A_{C_2}$, which in either case can be specified via a family of linear inequalities.  By consequence of \Cref{NestingPolytopes}, the space of programs which are accessible via circuits of native gates of a fixed depth then appears as a projection of linear slice of the monodromy polytope.  Our goal is to give descriptions of these projections.

\subsection{The $\CZ$ gate}\label{SectionCZ}

The sets $\LogCoords(P^0_{\CZ})$ and $\LogCoords(P^1_{\CZ})$ are singletons and hence automatically convex polytopes:
\begin{align*}
    \LogCoords(P^0_{\CZ}) & = e_1, &
    \LogCoords(P^1_{\CZ}) & = e_2.
\end{align*}
In order to compute $\LogCoords(P^2_{\CZ})$, we intersect the polytope described in \Cref{MainAWBTheorem} with the six hyperplanes describing the conditions $\alpha_* = e_2$ and $\beta_* = e_2$:
\begin{lemma}[{cf. \cite[Proposition III.3]{SBM}}]\label{P2CZDescription}
The convex polytope $\LogCoords(P^2_{\CZ})$ is described by
\begin{align*}
\left\{(\delta_1, \delta_2, \delta_3, \delta_4) \in \A_{C_2} \middle| \begin{array}{c}\delta_1 = - \delta_4, \\ \delta_2 = -\delta_3\end{array}\right\}.
\end{align*}
The extremal points of $\LogCoords(P^2_{\CZ})$ are $\{e_1, e_2, e_3\}$, with circuit realizations given in \Cref{RealizationsP2CZ}.
\end{lemma}
\begin{proof}
Using the quantum Littlewood--Richardson coefficients
\begin{align*}
N_{(1,0) (1,0)}^{(2,0),0}(2,2) & = 1, &
N_{(1,0)(1,0)}^{(1,1),0}(2,2) & = 1
\end{align*}
we apply \Cref{MainAWBTheorem} to deduce the inequalities
\begin{align*}
    0 - (\alpha_{2+1-1} + \alpha_{2+2-0}) - (\beta_{2+1-1} + \beta_{2+2-0}) & \\
    + \delta_{2+1-2} + \delta_{2+2-0} & \ge 0, \\
    0 - (\alpha_{2+1-1} + \alpha_{2+2-0}) - (\beta_{2+1-1} + \beta_{2+2-0}) & \\
    + \delta_{2+1-1} + \delta_{2+2-1} & \ge 0,
\end{align*}
i.e.,
\begin{align*}
    0 - (1/4 + -1/4) - (1/4 + -1/4) + \delta_1 + \delta_4 & \ge 0, \\
    0 - (1/4 + -1/4) - (1/4 + -1/4) + \delta_2 + \delta_3 & \ge 0.
\end{align*}
Because of the additional constraint $\delta_+ = 0$, we learn that these nonnegative quantities are in fact exactly zero:
\begin{align*}
\delta_1 & = - \delta_4, &
\delta_2 & = - \delta_3.
\end{align*}
This plane passes through three of the extremal vertices of $\A_{C_2}$, hence $\LogCoords(P^2_{\CZ})$ is contained inside of the triangle formed as the convex hull of those three vertices.  In order to show that this inclusion is actually an equality, we need only produce witnesses that these three points have preimages in $P^2_{\CZ}$.  One checks that the circuits described in \Cref{RealizationsP2CZ} do the job: not only do they image to the appropriate vertex under $\LogCoords(-)$, but the mirroring value $j$ in \Cref{MainAWBTheorem} is not used, so that these vertices belong to the same polytope.  Hence, their convex hull is as claimed.
\end{proof}

\begin{remark}\label{ArduousMechanismRem}
One can avoid divine inspiration by instead producing the entire family of inequalities of \Cref{MainAWBTheorem}, adding the equalities coming from our selection of $\alpha_* = \beta_* = \LogCoords(\CZ)$, optionally pre-reducing the system, calculating all of the points of triple intersection, and throwing out those points which do not satisfy the original family of inequalities.  This will, ultimately, produce the same set of extremal points.  While this has the benefit of being mechanical, it is quite arduous---and it does not produce the realizations of the extremal points as circuits.
\end{remark}

The briefest method for accessing $P^3_{\CZ}$ follows along identical lines: if we can show that the extremal vertices of $\A_{C_2}$ have realizations within $\LogCoords(P^3_{\CZ})$ and we are allowed to apply convexity, then we can conclude the following equality:
\begin{lemma}[{cf. \cite[Section V]{SMB}}]\label{P3CZDescription}
$\LogCoords(P^3_{\CZ}) = \A_{C_2}$, with circuits realizing the extremal points given in \Cref{RealizationsP3CZ}.
\end{lemma}
\begin{proof}
It is automatic that we have $\LogCoords(P^3_{\CZ}) \subseteq \A_{C_2}$.  In order to show the opposite inclusion, we need to supply the necessary realizations (with the necessary mirroring property, as in the proof of \Cref{P2CZDescription}).  One may check directly that the circuits in \Cref{RealizationsP3CZ} will do.
\end{proof}

\begin{remark}
There is also the following alternative approach that mimics the alternative approach to $P^2_{\CZ}$.  By intersecting the full polytope described by \Cref{MainAWBTheorem} with the family of inequalities which constrain $\alpha_* \in \LogCoords(P^2_{\CZ})$ and with the equality which constrains $\beta_* = \LogCoords(\CZ)$, we arrive at a convex polytope contained in $\A_{C_2} \times * \times \A_{C_2}$.  Using Fourier--Motzkin elimination to delete the first factor yields $\LogCoords(P^3_{\CZ})$ by projection to the last factor.  This, too, is completely mechanical but is even more arduous.
\end{remark}

\begin{remark}
In order to explain the provenance of the first three circuits in \Cref{RealizationsP3CZ}, we remark that \Cref{P2CZDescription} shows that $\CZ \in P^2_{\CZ}$, and hence we are led to the method suggested by \Cref{NestingPolytopes}.  Beginning with the realization of the extremal vertex $e_2 \in \LogCoords(P^2_{\CZ})$, we need only solve for the outer local gates to realize $\CZ$ exactly, as in:
\begin{center}
    \begin{tikzcd}[column sep=0.2em]
    & \ctrl{1} & \qw \\
    & \control{} & \qw
    \end{tikzcd}
    =
    \begin{tikzcd}
    & \qw & \ctrl{1} & \qw & \ctrl{1} & \qw & \gate{\Z_{\frac{\pi}{2}}} & \qw \\
    & \gate{\X_{\frac{\pi}{2}}} & \control{} & \gate{\Y_{\frac{-\pi}{2}}} & \control{} & \gate{\X_{-\frac{\pi}{2}}} & \gate{\Z_{\frac{\pi}{2}}} & \qw
    \end{tikzcd}
\end{center}
Using this formula, we may inflate the other realizations of the extremal vertices of $\LogCoords(P^2_{\CZ})$ into realizations in $\LogCoords(P^3_{\CZ})$ by substituting the above circuit for $\CZ$ in for, say, the left-hand $\CZ$ supplied in \Cref{P2CZDescription}.  The realization supplied for the fourth extremal vertex is a rephrasing of the usual expression of $\SWAP$ as a triple of alternating $\CNOT$s, and the fifth we produced by numerical search.
\end{remark}

\begin{remark}[{\cite[Proposition V.1]{SMB}}]\label{CZNailsCAN}
Critically for quantum compilation, a circuit realization for any point within $\LogCoords(P^3_{\CZ})$ can be exactly produced algorithmically.  The circuit proposed by Shende, Markov, and Bullock is \[\CAN(\alpha, \beta, \delta) =
\begin{tikzcd}
& \ctrl{1} & \gate{\Y_c} & \targ{} & \gate{\Y_a} & \ctrl{1} & \qw \\
& \targ{}  & \gate{\Z_{\frac{b}{2}}} & \ctrl{-1} & \gate{\Z_{\frac{b}{2}}} & \targ{} & \qw
\end{tikzcd},\]
where $a$, $b$, and $c$ are certain linear functions of $\alpha$, $\beta$, and $\delta$.  In general, exact decompositions do not seem to exist (cf.\ \Cref{NoLinearLunch} for a basic such result), and even numerical methods pose a challenge (cf.\ \Cref{OpenQuestionsSection}).
\end{remark}

\begin{corollary}[{cf.\ \cite{SBM}}]\label{CZExpectedDepth}
The expected circuit depth for $\CZ$ is $\<\L_{\CZ}\> = 3$. \qed
\end{corollary}

\subsection{The $\ISWAP$ gate}\label{SectionISWAP}

Now we prove analogous results for the gate set $S = \{\ISWAP\}$.  We will be briefer in the aspects that exactly mimic those for the gate $\CZ$.

The sets $\LogCoords(P^0_{\ISWAP})$ and $\LogCoords(P^1_{\ISWAP})$ are again singletons:
\begin{align*}
    \LogCoords(P^0_{\ISWAP}) & = e_1, &
    \LogCoords(P^1_{\ISWAP}) & = e_3.
\end{align*}

\begin{lemma}\label{P2ISWAPDescription}
$\LogCoords(P^2_{\ISWAP})$ is described by \[\LogCoords(P^2_{\ISWAP}) = \left\{\delta_* \in \A_{C_2} \middle| \begin{array}{c} \delta_1 = - \delta_4, \\ \delta_2 = -\delta_3\end{array}\right\}.\]
The extremal points of $\LogCoords(P^2_{\ISWAP})$ are $\{e_1, e_2, e_3\}$, with circuit realizations given as in \Cref{RealizationsP2ISWAP}.
\end{lemma}
\begin{proof}
This proof entirely mimics that of \Cref{P2CZDescription}, but this time the relevant quantum Littlewood--Richardson coefficients are
\begin{align*}
N_{(0,0)(2,0)}^{(2,0),0}(2,2) & = 1, &
N_{(1,0)(2,1)}^{(1,1),0}(2,2) & = 1. \qedhere
\end{align*}
\end{proof}

Moving on to $P^3_{\ISWAP}$, we have
\begin{lemma}\label{P3ISWAPDescription}
$\LogCoords(P^3_{\ISWAP}) = \A_{C_2}$, with realizations of the extremal vertices as circuits given in \Cref{RealizationsP3ISWAP}.
\end{lemma}
\begin{proof}
Again, the proof is almost identical to that of \Cref{P3CZDescription}, beginning with an exact realization of $\ISWAP \in P^2_{\ISWAP}$ by solving for the outer local gates in the realization of $e_3 \in \LogCoords(P^2_{\ISWAP})$:
\begin{align*}
    \begin{tikzcd}[ampersand replacement=\&, row sep=2em, column sep=0em] \arrow[dash,thick]{rr} \& \gate[2, style={rectangle, inner xsep=-2em}, label style={rotate=90, outer sep=-4em, inner sep=-4em}]{\ISWAP} \& \qw \\ \& \qw \& \qw \end{tikzcd}
    & =
    \begin{tikzcd}[ampersand replacement=\&, row sep=2em, column sep=0em]
    \& \gate{\Y_{\frac{\pi}{2}}} \arrow[dash,thick]{rr} \& \gate[2, style={rectangle, inner xsep=-2em}, label style={rotate=90, outer sep=-4em, inner sep=-4em}]{\ISWAP} \& \gate{\X_{\frac{\pi}{2}}} \arrow[dash,thick]{rr} \& \gate[2, style={rectangle, inner xsep=-2em}, label style={rotate=90, outer sep=-4em, inner sep=-4em}]{\ISWAP} \& \gate{\X_{\frac{\pi}{2}} \Z_{-\frac{\pi}{2}}} \& \qw \\
    \& \gate{\Y_{\frac{\pi}{2}}} \& \qw \& \gate{\X_{\frac{\pi}{2}}} \& \qw \& \gate{\X_{\frac{\pi}{2}} \Z_{-\frac{\pi}{2}}} \& \qw
    \end{tikzcd}
\end{align*}
Using this, we can inflate the left-hand $\ISWAP$ in the realizations of the extremal vertices in \Cref{P2ISWAPDescription} to produce realizations of those same vertices in $\LogCoords(P^3_{\ISWAP})$.  What remains is to produce realizations of the extremal points $\SWAP$ and $\sqrt{\SWAP}$, where we rely on a standard decomposition.
\end{proof}

\begin{remark}\label{ISWAPSameAsCZ}
Combining the results above with those from the previous subsection, we conclude
\begin{align*}
    P^2_{\CZ} & = P^2_{\ISWAP}, &
    P^3_{\CZ} & = P^3_{\ISWAP}.
\end{align*}
\end{remark}

\begin{corollary}
The expected circuit depth for $\ISWAP$ is $\<\L_{\ISWAP}\> = 3$. \qed
\end{corollary}

As in the case of $\CZ$, the compilation problem for $\ISWAP$ (i.e., \Cref{2QSandwichProblem}.\ref{2QSandwichProblem2} and its depth-three variant) admits exact solutions.  This does not appear to be in the literature, and so we include an analysis here:

\begin{corollary}\label{RealizingRealChiGammas}
For $1/2 \ge \alpha \ge \beta \ge 0$, the operator
\[U(\alpha, \beta) =
\begin{tikzcd}[row sep=2em, column sep=0em]
\arrow[dash,thick]{rr} & \gate[2, style={rectangle, inner xsep=-2em}, label style={rotate=90, outer sep=-4em, inner sep=-4em}]{\ISWAP} & \gate{\Y_{\frac{\alpha+\beta}{2} \cdot \pi}} \arrow[dash,thick]{rr} & \gate[2, style={rectangle, inner xsep=-2em}, label style={rotate=90, outer sep=-4em, inner sep=-4em}]{\ISWAP} & \qw \\
& \qw & \gate{\Y_{\frac{\alpha-\beta}{2} \cdot \pi}} & \qw & \qw
\end{tikzcd}\]
satisfies \[\LogCoords(U(\alpha, \beta)) = (\alpha, \beta, -\beta, -\alpha).\]
\end{corollary}
\begin{proof}
This is checked by direct computation.  One may ease the computation somewhat by noticing that conjugation by $\X_{\frac{\pi}{2}} \otimes \X_{\frac{\pi}{2}}$ diagonalizes the operator.
\end{proof}

\begin{remark}[{cf.\ \cite[Proposition V.2]{SBM}}]\label{ISWAPRealignment}
Our strategy for algorithmically producing circuits for points in $P^3_{\ISWAP}$ will be to isolate the troublesome extremal vertex $e_4$.  Once this vertex does not contribute to the remaining convex linear combination, the remainder is solved by \Cref{RealizingRealChiGammas}.  Selecting a gate $U \in PU(4)$, we seek local gates $A$ and $B$ so that \[
V(U, A, B) :=
\begin{tikzcd}
& \gate[2]{\ISWAP} & \gate{A} & \gate[2]{U} & \qw \\
& \qw & \gate{B} & \qw & \qw
\end{tikzcd}\]
satisfies $\LogCoords(V(U, A, B)) \in \LogCoords(P^2_{\ISWAP})$.  We apply \Cref{P2ISWAPDescription} to see that this is accomplished by finding values of $A$ and $B$ so that $\tr \gamma^Q(V(U, A, B))$ is real.  It turns out that we may take $A = \Y_{\sigma}$ and $B = 1$, which we see by manual calculation:
\begin{align*}
& -\frac{1}{2} \tr \gamma^Q\left(\begin{tikzcd}[ampersand replacement=\&]
\& \gate[2]{\ISWAP} \& \gate{\Y_\sigma} \& \gate[2]{U} \& \qw \\
\& \qw \& \qw \& \qw \& \qw
\end{tikzcd}\right) \\
& = \left(\sum_{j=1}^4 U_{2j} U_{3(j+(-1)^j)} - \sum_{k=1}^4 U_{4k} U_{1(k+(-1)^k)}\right) \cos \sigma \\
& \hspace{1em} + \left(\sum_{j=1}^2 \sum_{k=1}^4 U_{(2j)k} U_{(2j+1)(k-2)} (-1)^{(j-1)+k}\right) \sin \sigma,
\end{align*}
where we have interpreted the indices modulo $4$.  In particular, this summation formula enables us to solve the equation
\[
\Im\left(\tr \gamma^Q\left(V(U, \sigma)\right)\right) = 0
\]
by picking a value of $\sigma$ so that $\tan \sigma$ agrees with
\[
\frac{-\Im\left(\sum_j U_{2j} U_{3(j+(-1)^j)} - \sum_k U_{4k} U_{1(k+(-1)^k)}\right)}{\Im\left(\sum_{j=1}^2 \sum_k U_{(2j)k} U_{(2j+1)(k-2)} (-1)^{(j-1)+k}\right)}.
\]
Because the tangent function is surjective, there is always such a value.
\end{remark}

As one remaining case of interest, we can also describe the collection of gates accessible to a gate set that has \emph{both} $\CZ$ and $\ISWAP$ available:

\begin{lemma}
The set $\LogCoords(P^2_{\ISWAP, \CZ})$ is the union of $\LogCoords(P^2_{\ISWAP})$ and the convex polytope \[\left\{\delta_* = (\delta_1, \delta_2, \delta_3, \delta_4) \in \A_{C_2} \middle| \begin{array}{c} \delta_1 = 1/2 - \delta_2, \\ \delta_4 = -1/2 - \delta_3 \end{array} \right\},\] which has extremal vertices $\{e_2, e_3, e_4\}$. \qed
\end{lemma}

\begin{corollary}
The expected depth for the gate set with $\CZ$ and $\ISWAP$ is $\<\L_{\CZ, \ISWAP}\> = 3$. \qed
\end{corollary}

\begin{remark}
The metrics $\<\L_\S\>$ capture the efficacy of $\S$ at encoding \emph{random} two-qubit programs, but programs appearing ``in the wild'' (as well as sub-programs appearing as components in decompositions) are not random: the program $\SWAP \in PU(4)$ is a prime example of an uncommonly important two-qubit interaction.  Although the equation \[\<\L_\CZ\> = \<\L_{\ISWAP}\> = \<\L_{\CZ, \ISWAP}\>\] indicates that there is no salient difference between these different gate sets from the perspective of random programs, \[\SWAP \in P^2_{\CZ, \ISWAP} \setminus P^2_{\CZ}\] indicates that there is an important difference from the perspective of structured programs.
\end{remark}

This gate set also admits algorithmic decomposition, which one may verify by direct calculation:

\begin{lemma}
For a point \[\delta_* \in \LogCoords(P^2_{\ISWAP, \CZ}) \setminus \LogCoords(P^2_{\ISWAP}),\] there are two entries satisfying \[-1/4 \le \delta_i \le \delta_j \le 1/4.\]  Setting $\alpha = (\delta_i + \delta_j)\pi$ and $\beta = (\delta_i - \delta_j)\pi$, we then have \[\delta_* = \LogCoords \left( \begin{tikzcd} & \gate[2]{\ISWAP} & \gate{\Y_\alpha} & \ctrl{1} & \qw \\ & \qw & \gate{\Y_\beta} & \control{} & \qw \end{tikzcd}\right). \pushQED\qed\qedhere\popQED\]
\end{lemma}

\subsection{The $\CPHASE$ and $\PSWAP$ gate families}\label{SectionCPHASEandPSWAP}

As further demonstration of these techniques, we also consider some combinations of the parametric two-qubit gates which appear in the Quil standard gate set~\cite{SCZ}.

\begin{lemma}\label{P2CPHASEDescription}
The convex polytope $\LogCoords(P^2_{\CPHASE})$ agrees with $\LogCoords(P^2_{\CZ})$ and with $\LogCoords(P^2_{\ISWAP})$.
\end{lemma}
\begin{proof}
The proof is identical to that given for \Cref{P2CZDescription}: the same quantum Littlewood--Richardson coefficients impose the same symmetry relation on $\LogCoords(P^2_{\CPHASE})$, and the reverse inclusion then follows from $P^2_{\CZ} \subseteq P^2_{\CPHASE}$.
\end{proof}

\begin{corollary}
The expected circuit depth for $\CPHASE$ is $\<\L_{\CPHASE}\> = 3$. \qed
\end{corollary}

Similarly, there is no substantial gain from mixing $\CPHASE$ with $\ISWAP$ over $\CZ$ with $\ISWAP$.

\begin{lemma}
$P^2_{\CZ, \ISWAP} = P^2_{\CPHASE, \ISWAP}$, and $\<\L_{\CPHASE, \ISWAP}\> = 3$. \qed
\end{lemma}

\begin{example}\label{SqrtCZExample}
As an exercise in the application of these methods, we provide an analysis of the polytopes associated to $\sqrt{\CZ} = \CPHASE_{\frac{\pi}{2}}$.  Since $\sqrt{\CZ} \cdot \sqrt{\CZ} = \CZ$, we can deduce immediately from \Cref{CZExpectedDepth} (i.e., from previously known methods) that $\<\L_{\sqrt{\CZ}}\> \le 6$.  Coupling our methods to the software \texttt{lrs}~\cite{Avis,AvisFukuda}, we enumerate the vertices of the sets $\LogCoords(P^n_{\sqrt{\CZ}})$ for $n \le 5$ in \Cref{sqrtCZPoints}, as displayed in graphical form in \Cref{sqrtCZPicture}.  These yield the exact calculation \[\<\L_{\sqrt{\CZ}}\>^{\A_{C_2}} = 3 \cdot \frac{1}{2} + 4 \cdot \frac{1}{3} + 5 \cdot \frac{7}{48} + 6 \cdot \frac{1}{48} = 3.6875,\] a considerable improvement over the naive estimate.
\end{example}

\begin{figure}
\centering
    \includegraphics[width=0.4\textwidth]{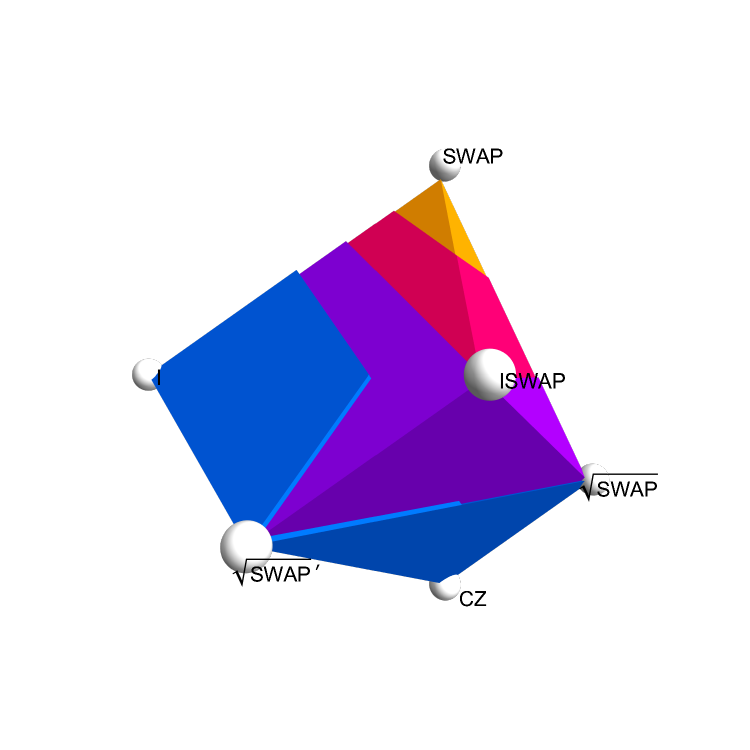}
    \includegraphics[width=0.4\textwidth]{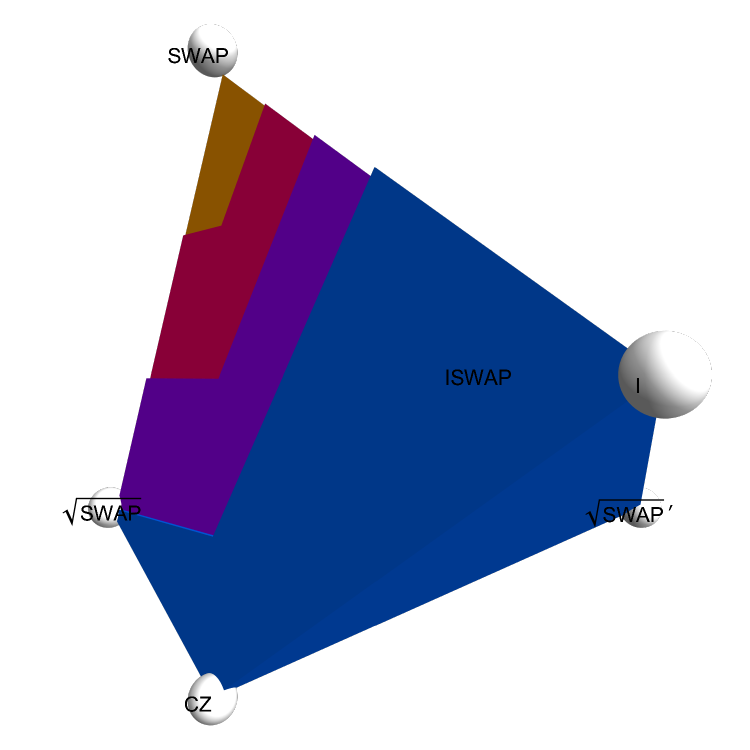}
    \includegraphics[width=0.4\textwidth,]{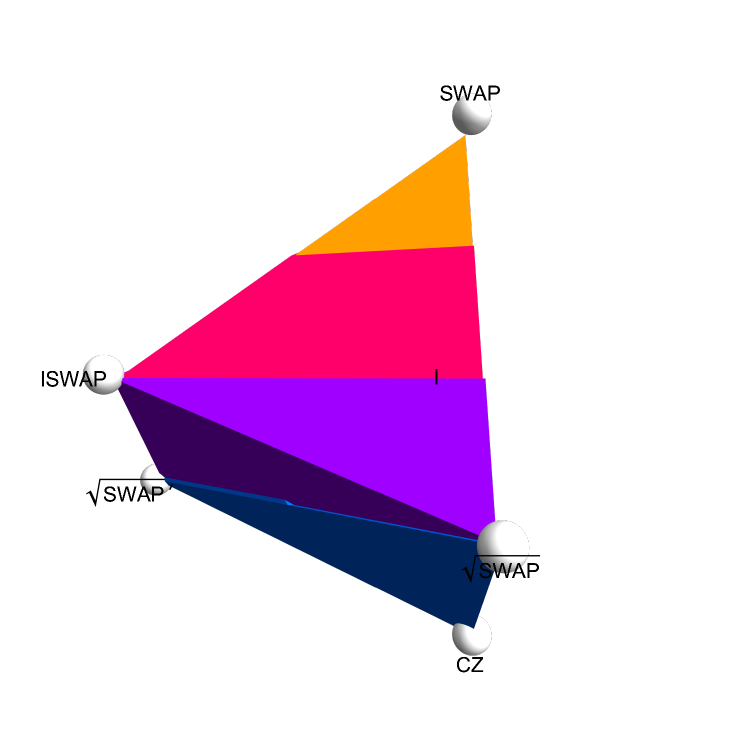}
\caption{The regions $\LogCoords(P^n_{\sqrt{\CZ}})$ within $\A_{C_2}$: blue, purple, pink, and orange respectively denote $P^3_{\sqrt{\CZ}}$, $P^4_{\sqrt{\CZ}}$, $P^5_{\sqrt{\CZ}}$, and $P^6_{\sqrt{\CZ}}$.  Not pictured are $P^0_{\sqrt{\CZ}}$ and $P^1_{\sqrt{\CZ}}$, each a point, and $P^2_{\sqrt{\CZ}}$, a flat triangle.}\label{sqrtCZPicture}
\end{figure}


The gate $\PSWAP$ is not natively available on charge-coupled superconducting hardware, so we do not explore it very thoroughly here, but for completeness we at least include a calculation of $P^2_{\PSWAP}$.
\begin{lemma}\label{P2PSWAPDescription}
$P^2_{\PSWAP}$ agrees with the other depth-two sets studied so far: \[P^2_{\PSWAP} = P^2_{\CZ} = P^2_{\ISWAP} = P^2_{\CPHASE}.\]
\end{lemma}
\begin{proof}
This proof proceeds similarly to that of \Cref{P2ISWAPDescription}.  This time the relevant quantum Littlewood--Richardson coefficients are
\begin{align*}
N_{(2,1)(2,1)}^{(2,0),1}(2,2) & = 1, &
N_{(2,1)(2,1)}^{(1,1),1}(2,2) & = 1,
\end{align*}
and $\LogCoords(\PSWAP_{2 \pi t})$ is calculated to be \[\left(\frac{3}{4} - \frac{t}{2}, -\frac{1}{4} + \frac{t}{2}, -\frac{1}{4} + \frac{t}{2}, -\frac{1}{4} - \frac{t}{2} \right).\]  An application of \Cref{MainAWBTheorem} yields inequalities which enforce the same symmetry conditions on $\LogCoords(P^2_{\PSWAP})$ as in the previous Lemmas.  Because we have $P^2_{\ISWAP} \subseteq P^2_{\PSWAP}$, we may conclude equality.
\end{proof}

\begin{corollary}
The expected circuit depth for $\PSWAP$ is $\<\L_{\PSWAP}\> = 3$. \qed
\end{corollary}

\section{Monodromy polytope slices for the $\XY$--family}

Combining the ideas which motivated $\ISWAP$ and $\CPHASE$, we are also motivated to consider the one-parameter family of native two-qubit gates given by
\begin{align*}
\XY_\alpha & = \exp\left(-i \alpha \cdot (\sigma_X^{\otimes 2} + \sigma_Y^{\otimes 2})\right) \\
& = \left( \begin{array}{cccc} 1 & 0 & 0 & 0 \\ 0 & \cos(\alpha/2) & -i \sin(\alpha/2) & 0 \\ 0 & -i \sin(\alpha/2) & \cos(\alpha/2) & 0 \\ 0 & 0 & 0 & 1 \end{array} \right).
\end{align*}
This family is interesting for a few reasons: it is one of the only\footnote{The other remaining edge is the ray connecting $\I$ to $\SWAP$, but we addressed this in \Cref{Manual2QExample}.} remaining ``edges'' of $\A$; it can arise naturally as a gate natively available to systems where $\ISWAP$ is available, as in \cite{CaldwellEtAl}; and it itself belongs to the canonical family.

Having noted that $\XY_\alpha$ belongs to the canonical family, we may compute its associated diagonal coordinates to be
\begin{align*}
\XY_\alpha^Q & = \left( \begin{array}{cccc} 1 & 0 & 0 & 0 \\ 0 & e^{i \alpha/2} & 0 & 0 \\ 0 & 0 & e^{-i \alpha/2} & 0 \\ 0 & 0 & 0 & 1 \end{array} \right), \\
\LogCoords(\XY_\alpha) & = \left( \frac{\alpha}{2\pi}, 0, 0, -\frac{\alpha}{2\pi} \right).
\end{align*}

In pursuit of an analogue of the results of \Cref{SectionCPHASEandPSWAP}, we can perform an analysis of the polytope $\LogCoords(P^2_{\XY})$.  The computation is much more involved, but we are rewarded with the following theorem:
\begin{theorem}
The set $\LogCoords(P^2_{\XY})$ is the union of the polytopes with extremal coordinates as specified in \Cref{RealizationsP2XY}.  The set $\LogCoords(P^3_{\XY})$ is the entire solid $\A_{C_2}$.
\end{theorem}
\begin{proof}
We compute along the lines of \Cref{ArduousMechanismRem}: we intersect the monodromy polytope with the hyperplane equations specifying that the first coordinate take the form $(\alpha_1, 0, 0, -\alpha_1)$ and that the second coordinate take the form $(\beta_1, 0, 0, -\beta_1)$; then we apply Fourier-Motzkin elimination to project to the third coordinate; and finally we feed the resulting system to the software package \texttt{lrs}~\cite{Avis,AvisFukuda}.  Altogether, this results in the vertex sets listed above.
\end{proof}

\begin{figure}
    \centering
    \includegraphics[width=0.3\textwidth, trim={3.5cm 3.5cm 3.5cm 3.5cm}]{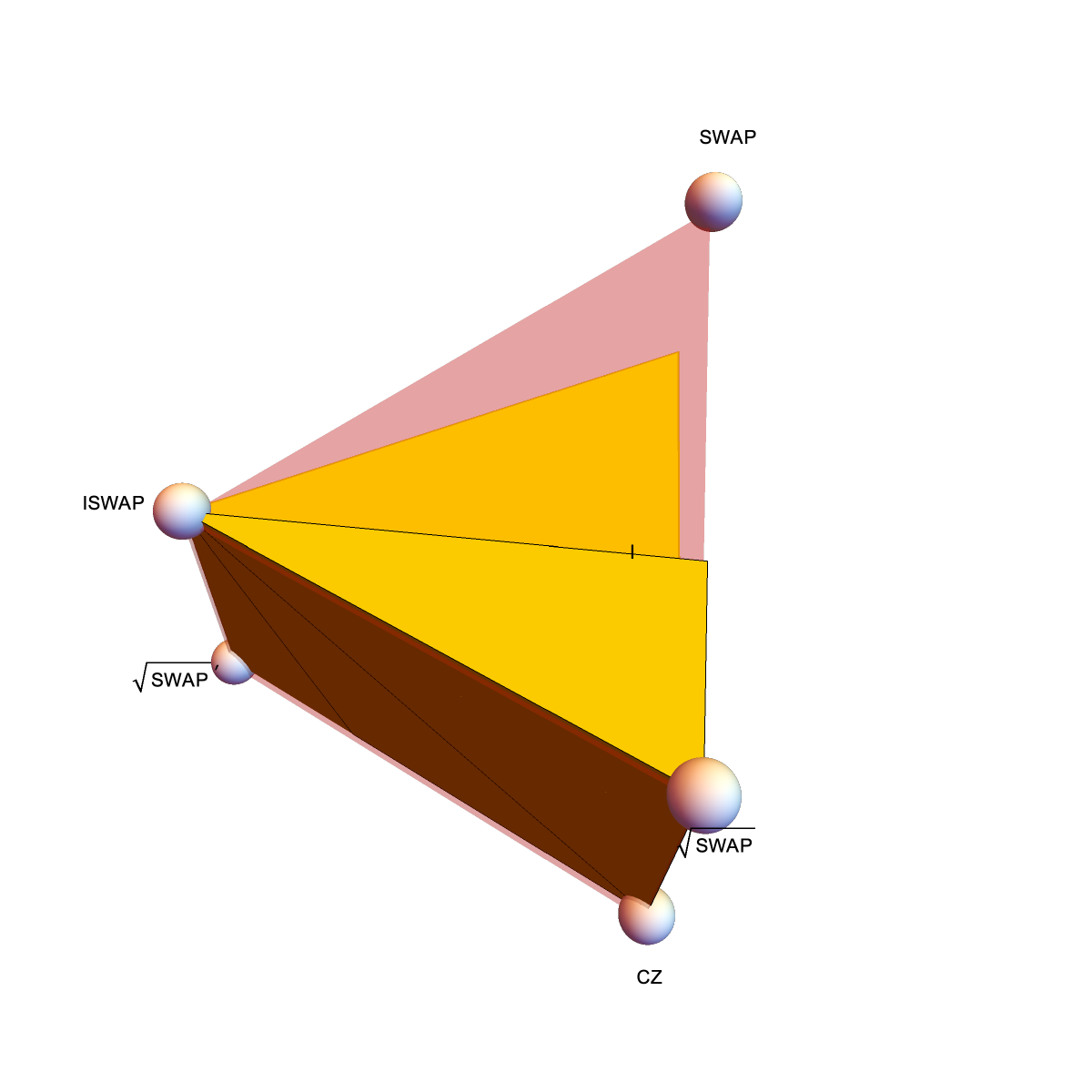}
    \includegraphics[width=0.35\textwidth,trim={3.5cm 3.5cm 3.5cm 3.5cm}]{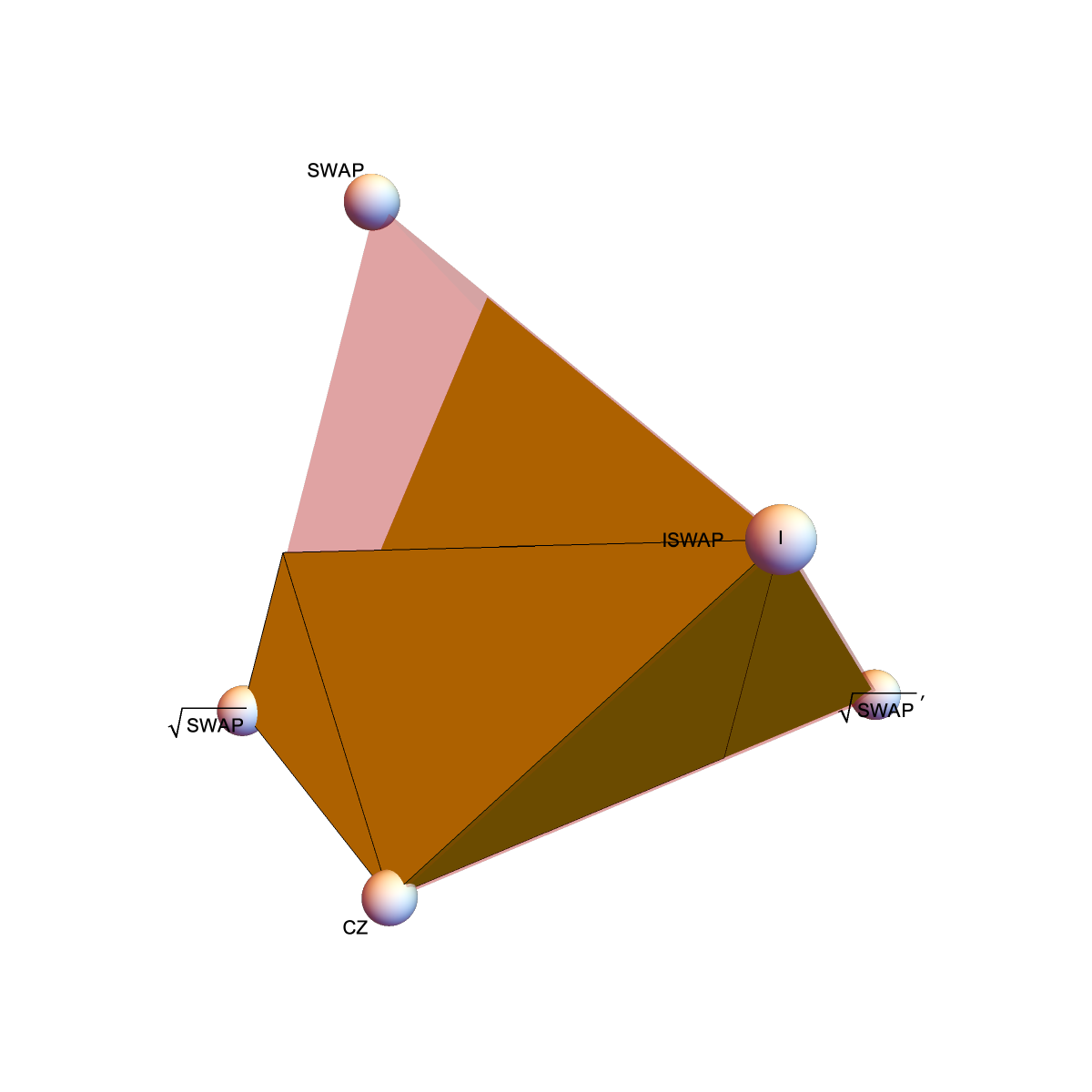}
    \includegraphics[width=0.3\textwidth, trim={3.5cm 3.5cm 3.5cm 3.5cm}]{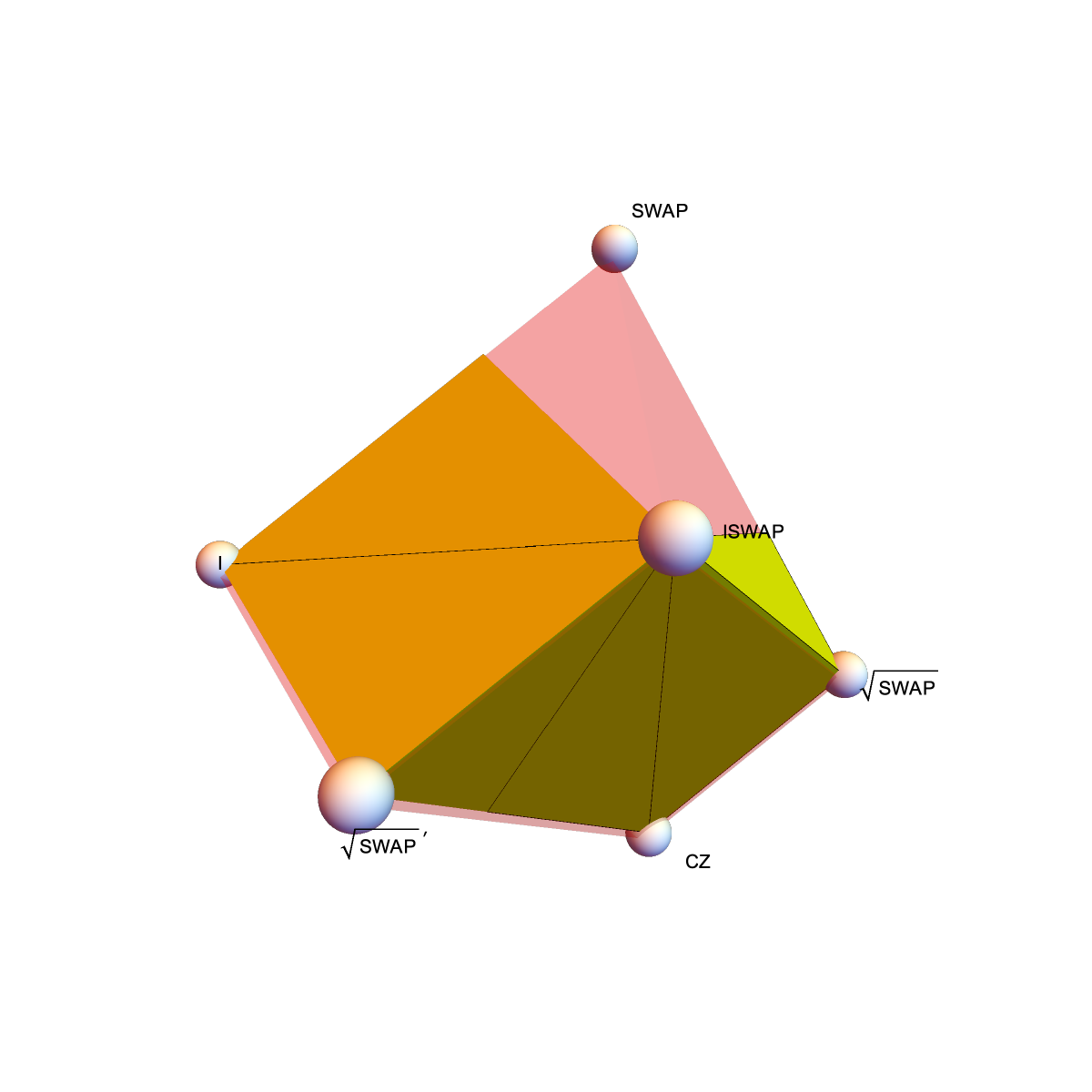} \vspace{\baselineskip}
    \caption{Three views of $\LogCoords(P^2_{\XY})$}
    \label{P2XYViews}
\end{figure}

\begin{corollary}\label{LXYComputation}
In particular, $\LogCoords(P^2_{\XY})$ is of positive volume.  More precisely, we compute the expected gate depth for $\XY$ to be\footnote{As mentioned in \Cref{Introduction}, we can approximate $\vol^{\Haar}(P^2_{\XY})$ to be $\approx 0.96$, which is quite different from $5/6$.  The skew in these two values comes from the commentary in \Cref{DefnExpectedAC2}: the remaining sixth of $\A_{C_2}$ is underdense for $\Pi_* \mu^{\Haar}$.} \[\<\L_{\XY}\>^{\A_{C_2}} = 2 \cdot \frac{5}{6} + 3 \cdot \frac{1}{6} = \frac{13}{6}. \qed\]
\end{corollary}

\begin{remark}
As an interesting aside, $\SWAP$ lies outside of this polytope.
\end{remark}

Our main observation is that $\XY$ enjoys a property that none of the other gate sets have thus far: $P^2_{\XY}$ is top-dimensional or, said otherwise, has positive volume.  In the $\CPHASE$ family, we found in \Cref{P2CPHASEDescription} that $P^2_{\CPHASE}$ had zero volume, from which we can also conclude that $P^2_{\CPHASE_\alpha}$ has zero volume for every fixed of $\alpha$, including $\CPHASE_\pi = \CZ$ (cf. \Cref{P2CZDescription}).  We here pursue the corresponding question of whether there are any fixed values of $\alpha$ for which $P^2_{\XY_\alpha}$ has nonzero volume.\footnote{Of course, this is not automatically true: these subpolytopes could form something like a ``foliation'' of $\LogCoords(P^2_{\XY})$.}  In the event that such slices exist, we can ask an additional question: which particular values of $\alpha$ maximize the volume of the slice?

Fix $0 \le \alpha < \pi$ with corresponding value $t = \alpha / \pi$ satisfying $0 \le t \le 1$.  The fundamental alcove sequences under consideration are then
\begin{align*}
\alpha_* & = (t/2, 0, 0, -t/2), \\
\beta_*  & = (t/2, 0, 0, -t/2), \\
\delta_* & = (\delta_1 \ge \delta_2 \ge \delta_3 \ge \delta_4),
\end{align*}
and the inequalities given by combining \Cref{MainAWBTheorem} with \Cref{QuantumLittlewoodRichardsonCoeffs} and the above alcove sequences are
\begin{align*}
\delta_4 + t & \ge 0, & \delta_3 + \delta_4 + t & \ge 0, \\
-\delta_1 + t & \ge 0, & \delta_3 +  t/2 & \ge 0, \\
\delta_1 + \delta_4 + t/2  & \ge 0, & -\delta_2 + t/2  & \ge 0, \\
\delta_2      & \ge 0,  & \delta_1 + \delta_4 - t & \ge -1, \\
-\delta_3      & \ge 0, & \delta_4 -  t/2 & \ge -1, \\
\delta_2 + \delta_3 - t & \ge -1, & -\delta_1 - t/2  & \ge -1.
\end{align*}
From these inequalities, we may draw the following consequence:

\begin{theorem}\label{MaximizingXYThm}
The volume of $\LogCoords(P^2_{\XY_\alpha})$ is maximized at $\alpha = 3\pi/4$.
\end{theorem}
\begin{proof}
Because the finite family of inequalities determining $\LogCoords(P^2_{\XY_\alpha})$ are linearly dependent in $\alpha$, the curve $\vol \LogCoords(P^2_{\XY_\alpha})$ is piecewise cubic in $\alpha$.  One can use this fact, together with sampling~\cite{Lawrence} and interpolation techniques, to determine a formula for $\vol \LogCoords(P^2_{\XY_\alpha})$: \[\vol \LogCoords(P^2_{\XY_\alpha}) = \begin{cases} 4t^3 & 0 \le t \le \frac{1}{2}, \\ \frac{15}{2} - 36 t + 60 t^2 - 32 t^3 & \frac{1}{2} \le t \le \frac{3}{4}, \\ -6 + 18 t - 12 t^2 & \frac{3}{4} \le t \le 1, \end{cases}\] as depicted in \Cref{SymmetricXYVolumeFig}.  From this curve, we may directly determine its maximum value.
\end{proof}

\begin{figure}
    \centering
    \includegraphics[width=0.45\textwidth]{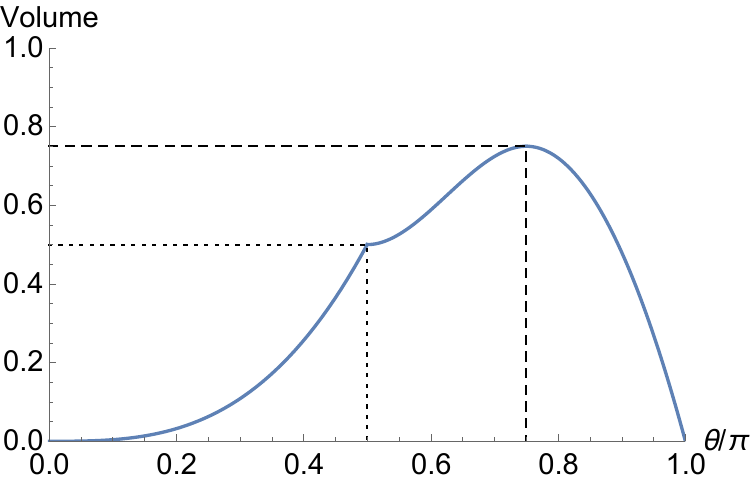}
    \caption{Volume of $\LogCoords(P^2_{\XY_\alpha})$, plotted as a fraction of the volume of $\A_{C_2}$ against $\alpha/\pi$.}
    \label{SymmetricXYVolumeFig}
\end{figure}

\begin{definition}
Motivated by \Cref{MaximizingXYThm}, we also refer to $\XY_{\frac{3 \pi}{4}}$ by the briefer synonym $\DB$.\footnote{Dagwood Bumstead is a comic strip character famous for making really big sandwiches.}
\end{definition}

\begin{remark}\label{LDBComputation}
Similarly, one may compute $\LogCoords(P^3_{\DB}) = \A_{C_2}$, from which we conclude \[\<\L_{\DB}\>^{\A_{C_2}} = 2 \cdot \frac{3}{4} + 3 \cdot \frac{1}{4} = \frac{9}{4}.\]  In fact, $\LogCoords(P^2_{\XY_{\pi/2}}) \cup \LogCoords(P^3_{\XY_{\pi/2}}) = \A_{C_2}$, so that \[\<\L_{\XY_{\pi/2}}\>^{\A_{C_2}} = 2 \cdot \frac{1}{2} + 3 \cdot \frac{1}{2} = \frac{5}{2}.\]
\end{remark}

\begin{remark}
In \Cref{P2XYFigure1}, \Cref{P2XYFigure2}, and \Cref{P2XYFigure3}, we illustrate the solids $\LogCoords(P^2_{\XY_\alpha})$ for varying values of $\alpha$, where we have projected onto the last three coordinates and shaded $\A_{C_2}$ red.  We record here (but do not prove) some interesting observations about the solids.  First, for $0 \le \alpha \le \alpha' \le 3\pi/4$, there is an inclusion of solids $\LogCoords(P^2_{\XY_\alpha}) \subseteq \LogCoords(P^2_{\XY_{\alpha'}})$, from which it follows that $\vol \LogCoords(P^2_{\DB}) \ge \vol \LogCoords(P^2_{\XY_\alpha})$ for any $0 \le \alpha \le 3\pi/4$ as in the Theorem.  However, for $3 \pi/4 \le \alpha < \alpha' \le \pi$, neither of $\LogCoords(P^2_{\XY_\alpha})$ and $\LogCoords(P^2_{\XY_{\alpha'}})$ is contained in the other: although $\LogCoords(P^2_{\XY_\alpha})$ continues to lose volume as $\alpha$ approaches $\pi$ from the left, the solid also continues to pick up ``new'' two-qubit programs as it shrinks.
\end{remark}

\begin{figure}
    \centering
    \includegraphics[width=0.2\textwidth,trim={3cm 3cm 3cm 3cm},clip]{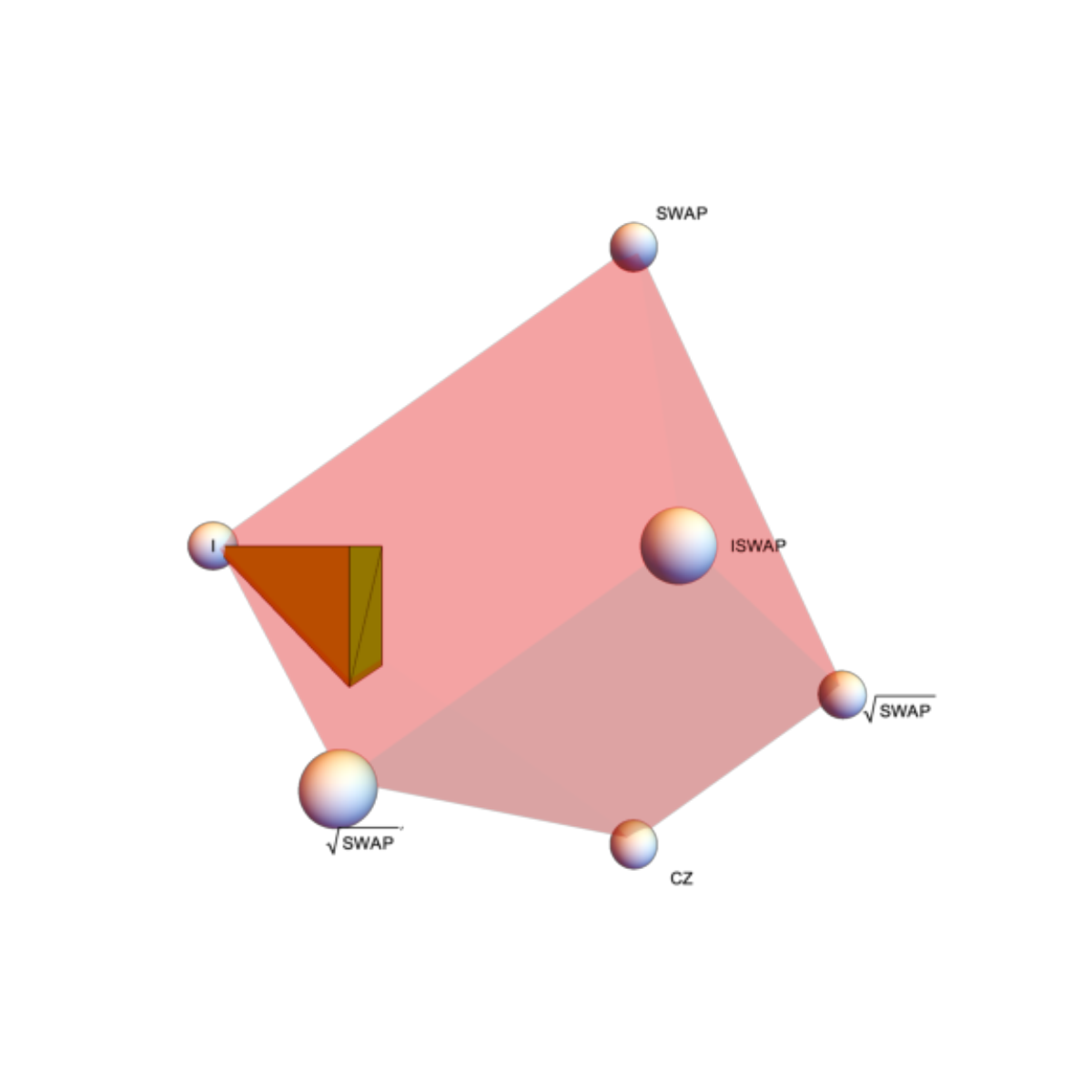}
    \includegraphics[width=0.2\textwidth,trim={3cm 3cm 3cm 3cm},clip]{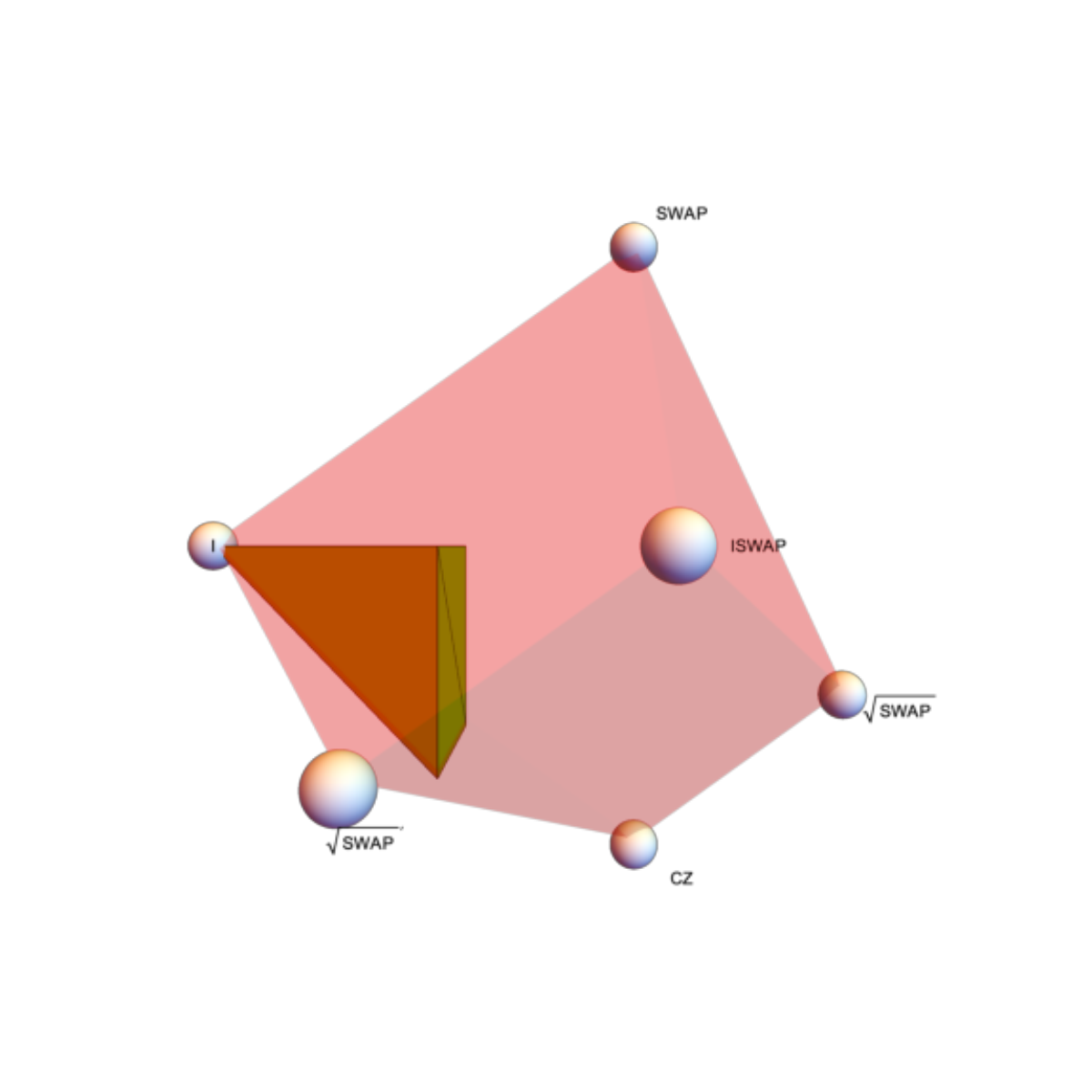}
    \includegraphics[width=0.2\textwidth,trim={3cm 3cm 3cm 3cm},clip]{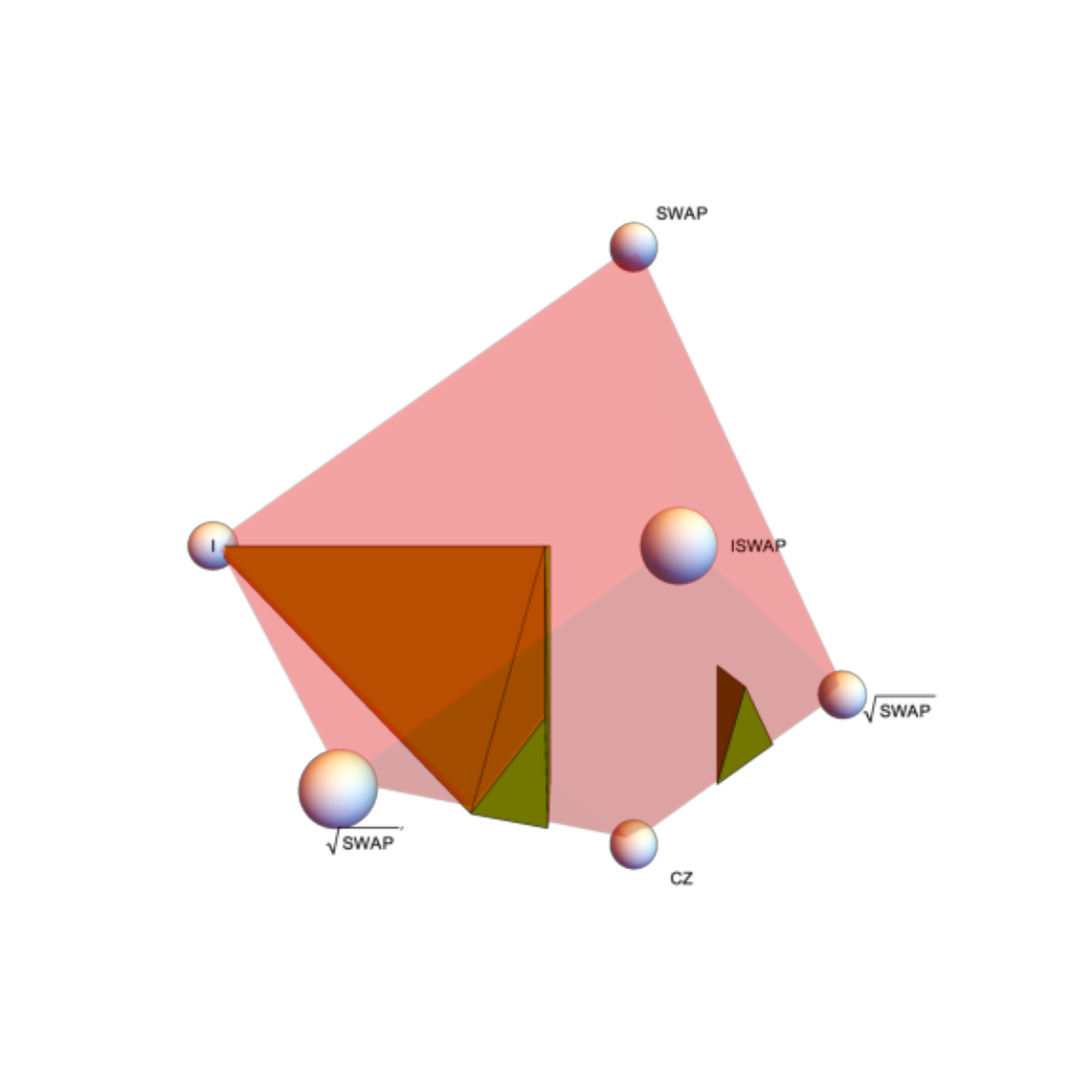}
    \includegraphics[width=0.2\textwidth,trim={3cm 3cm 3cm 3cm},clip]{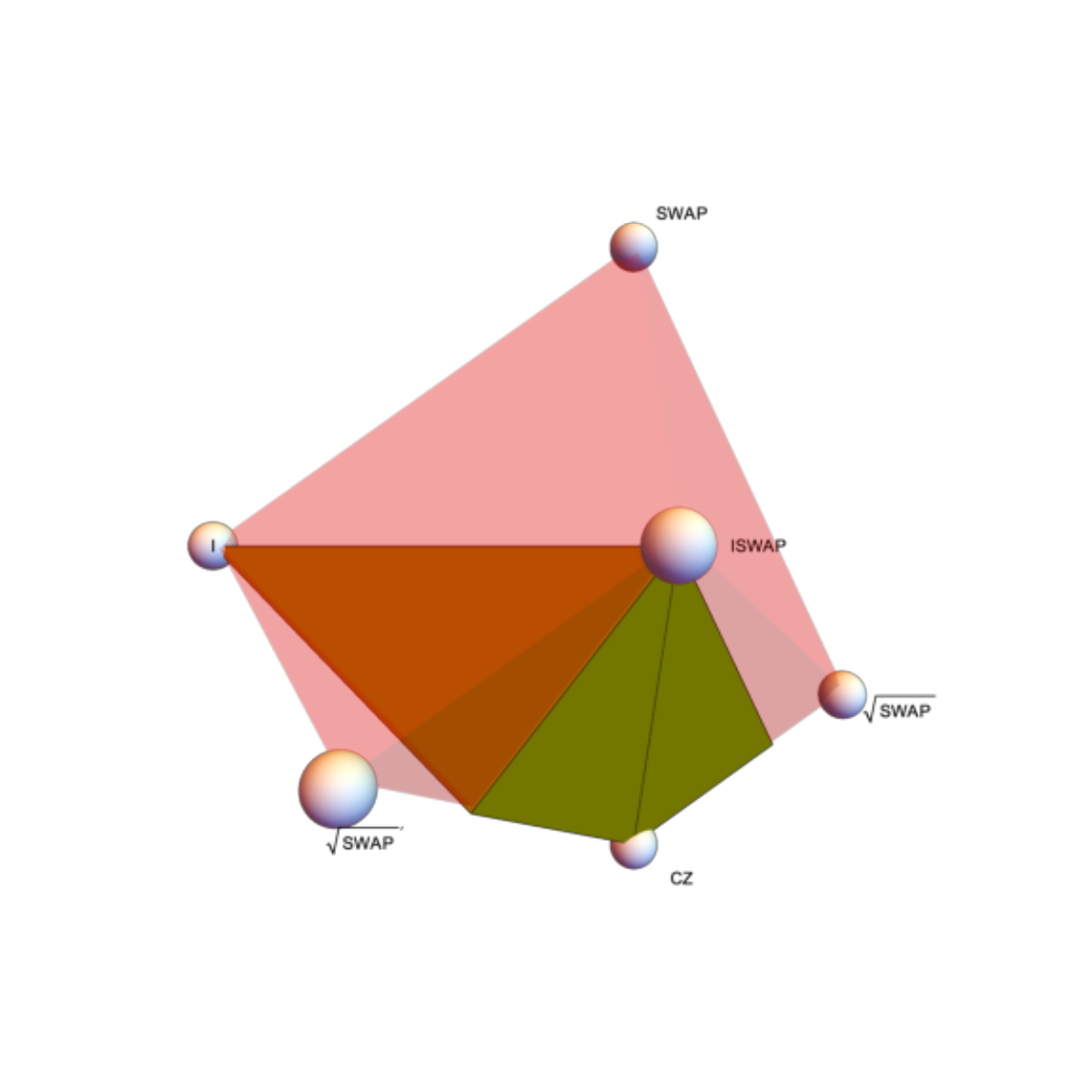}
    \includegraphics[width=0.2\textwidth,trim={3cm 3cm 3cm 3cm},clip]{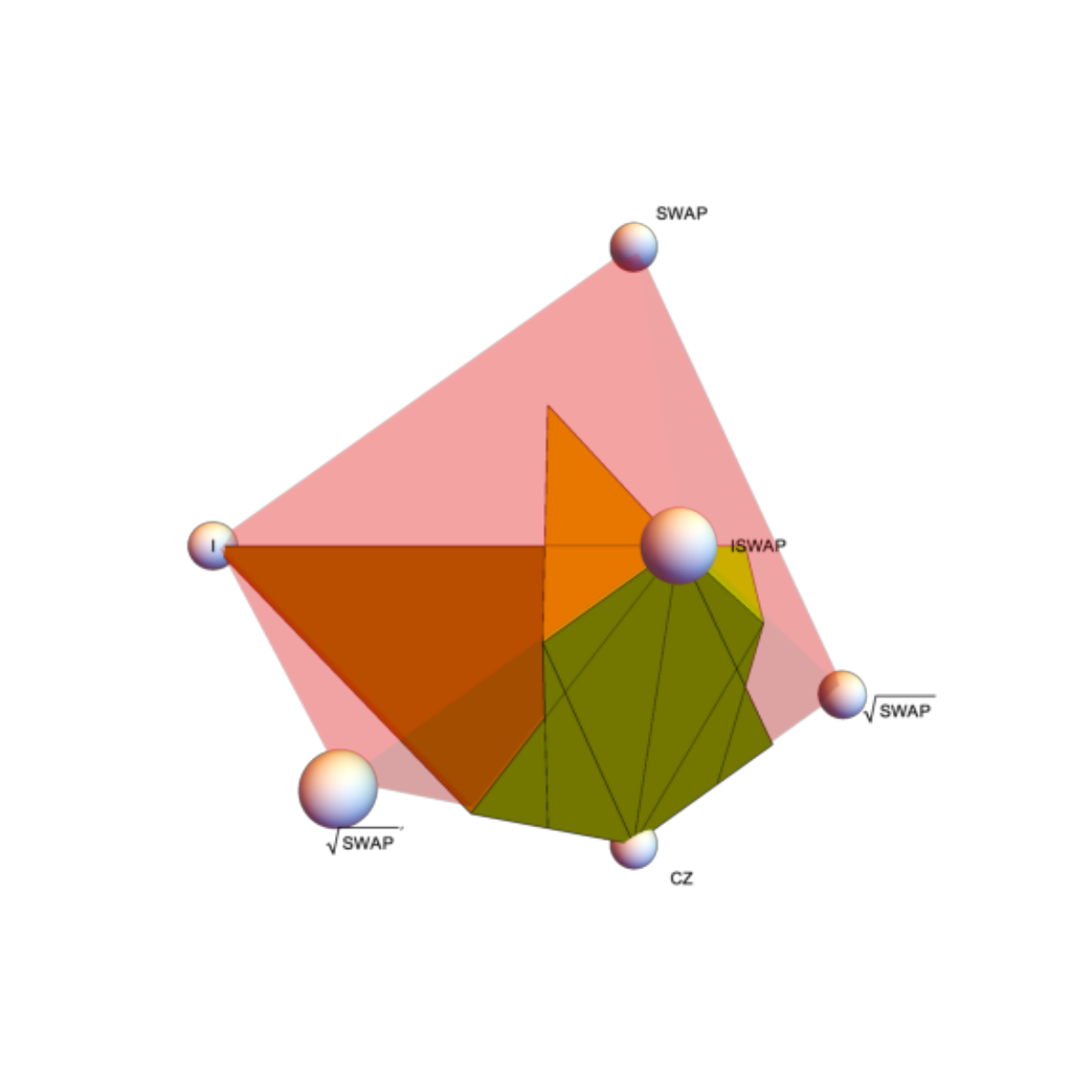}
    \includegraphics[width=0.2\textwidth,trim={3cm 3cm 3cm 3cm},clip]{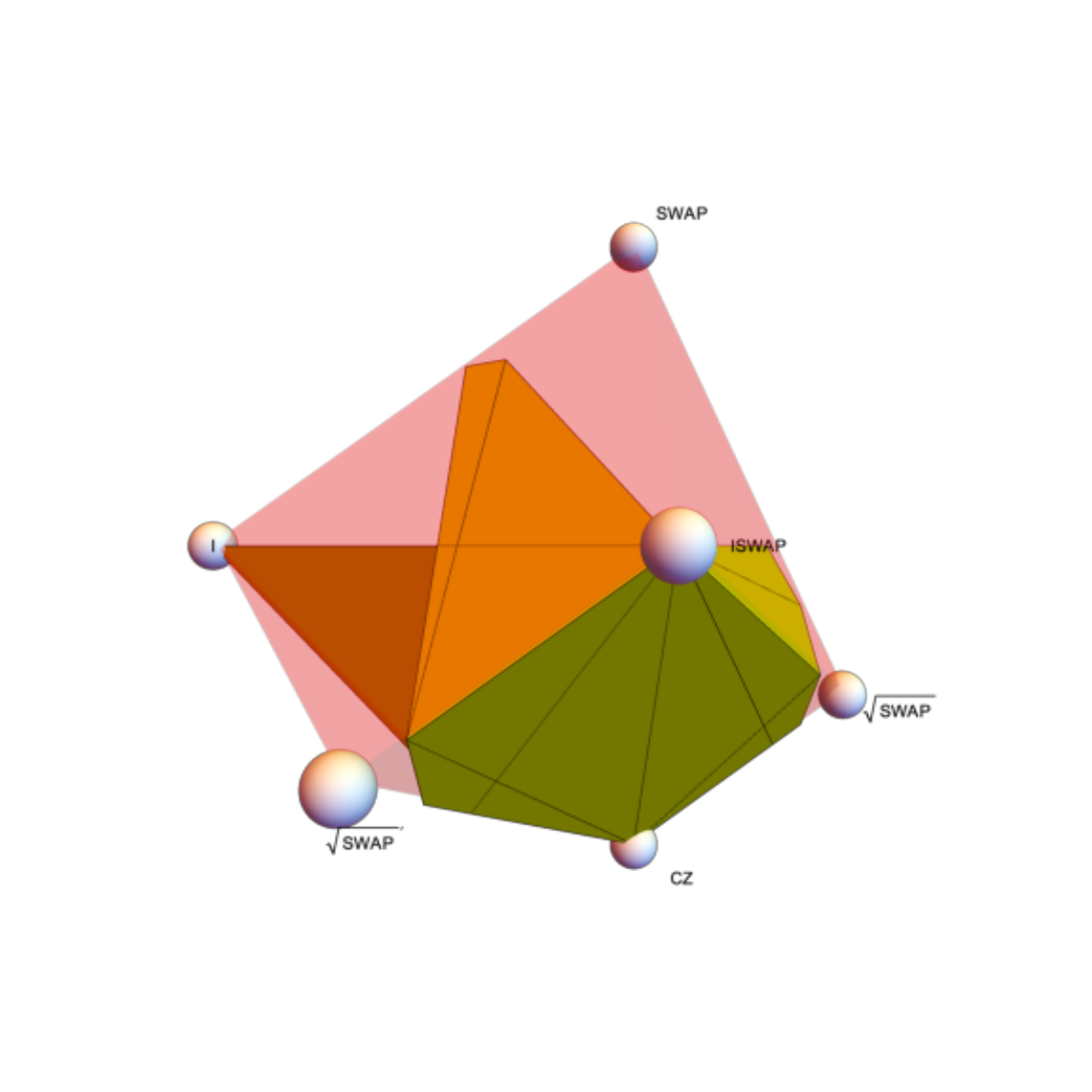}
    \includegraphics[width=0.2\textwidth,trim={3cm 3cm 3cm 3cm},clip]{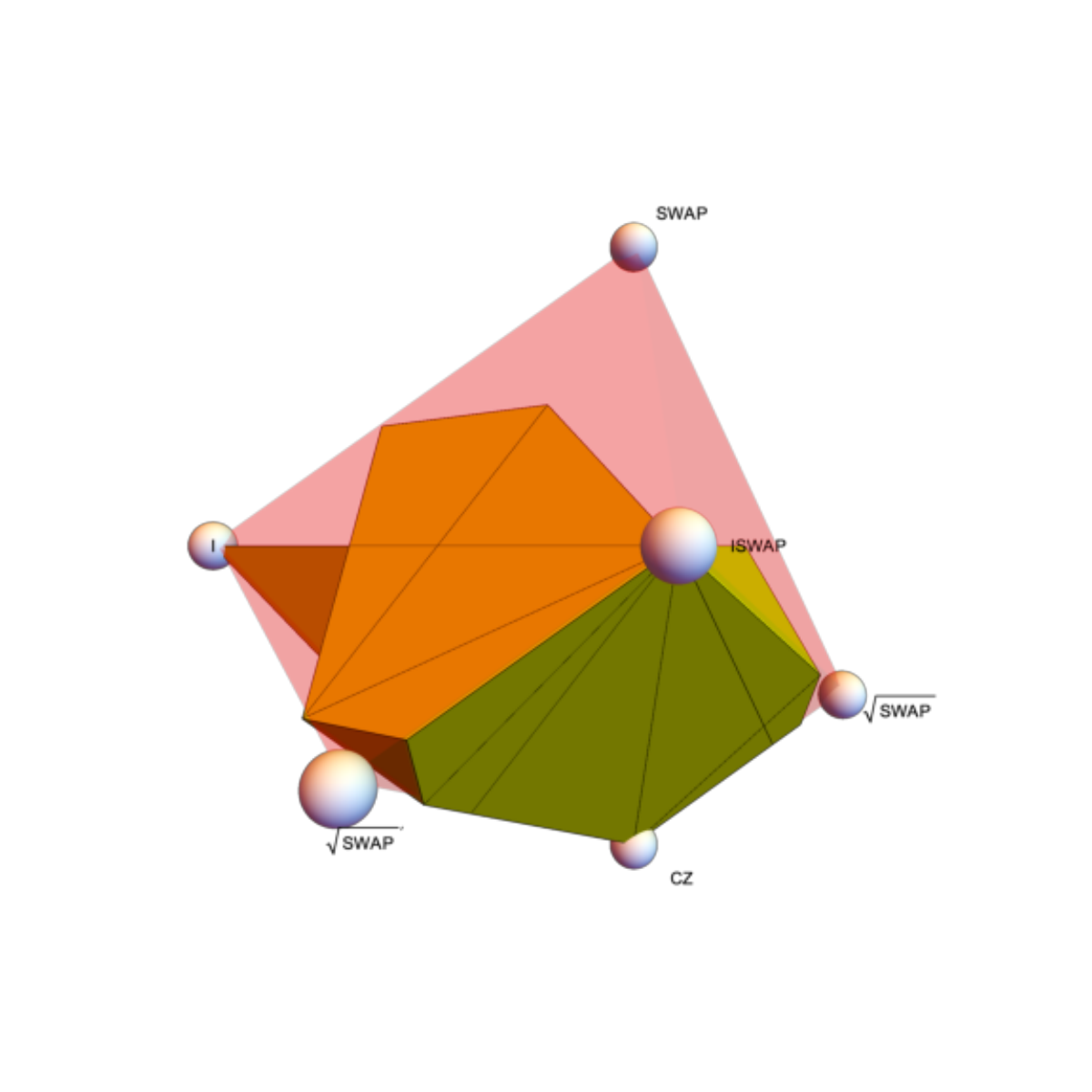}
    \includegraphics[width=0.2\textwidth,trim={3cm 3cm 3cm 3cm},clip]{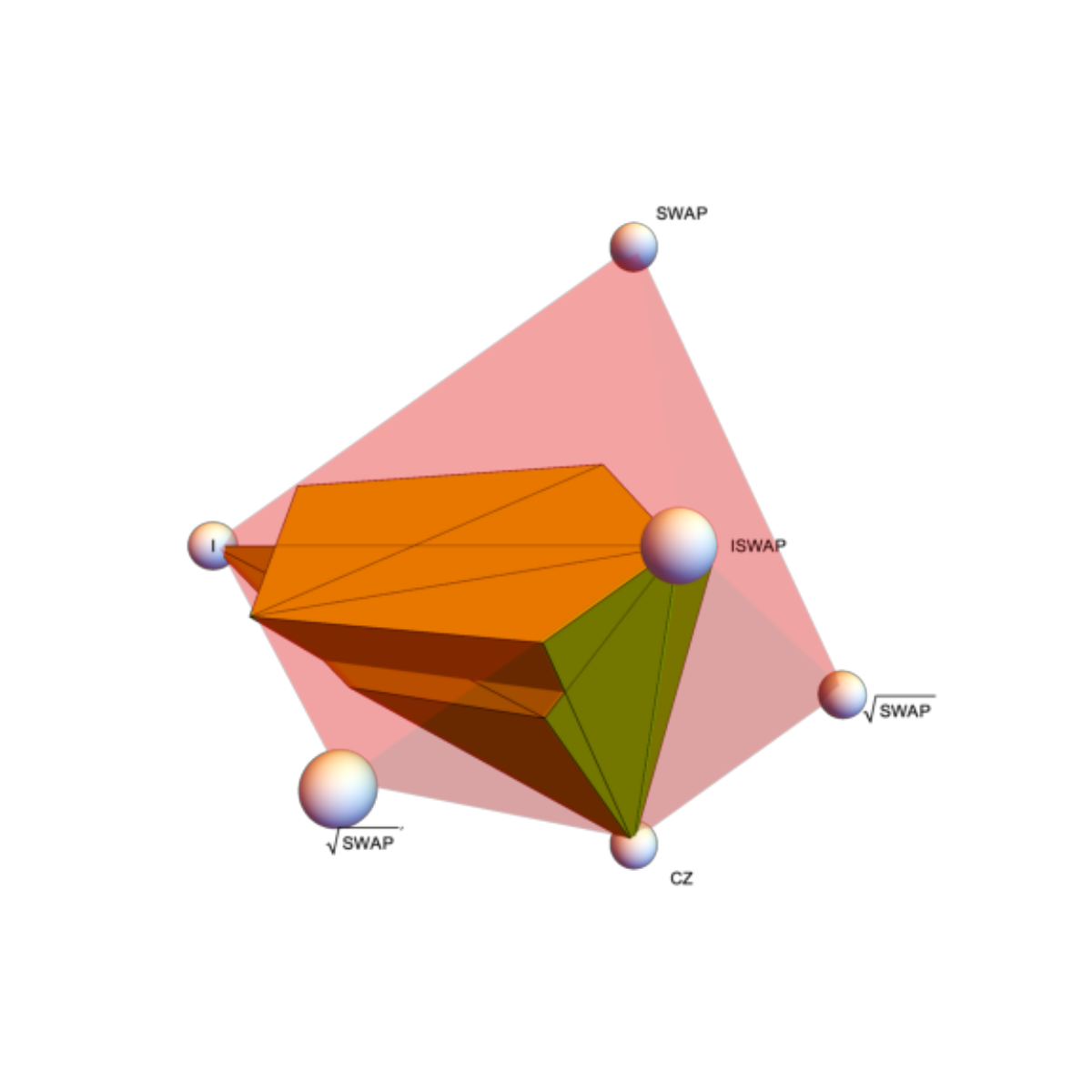}
    \caption{The solids $\LogCoords(P^2_{\XY_{\pi t}})$ for differing values of $t$: $2/10$, $3/10$, \ldots, $9/10$.}
    \label{P2XYFigure1}
\end{figure}

\begin{figure}
    \centering
    \includegraphics[width=0.2\textwidth,trim={3cm 1.5cm 3cm 1.5cm},clip]{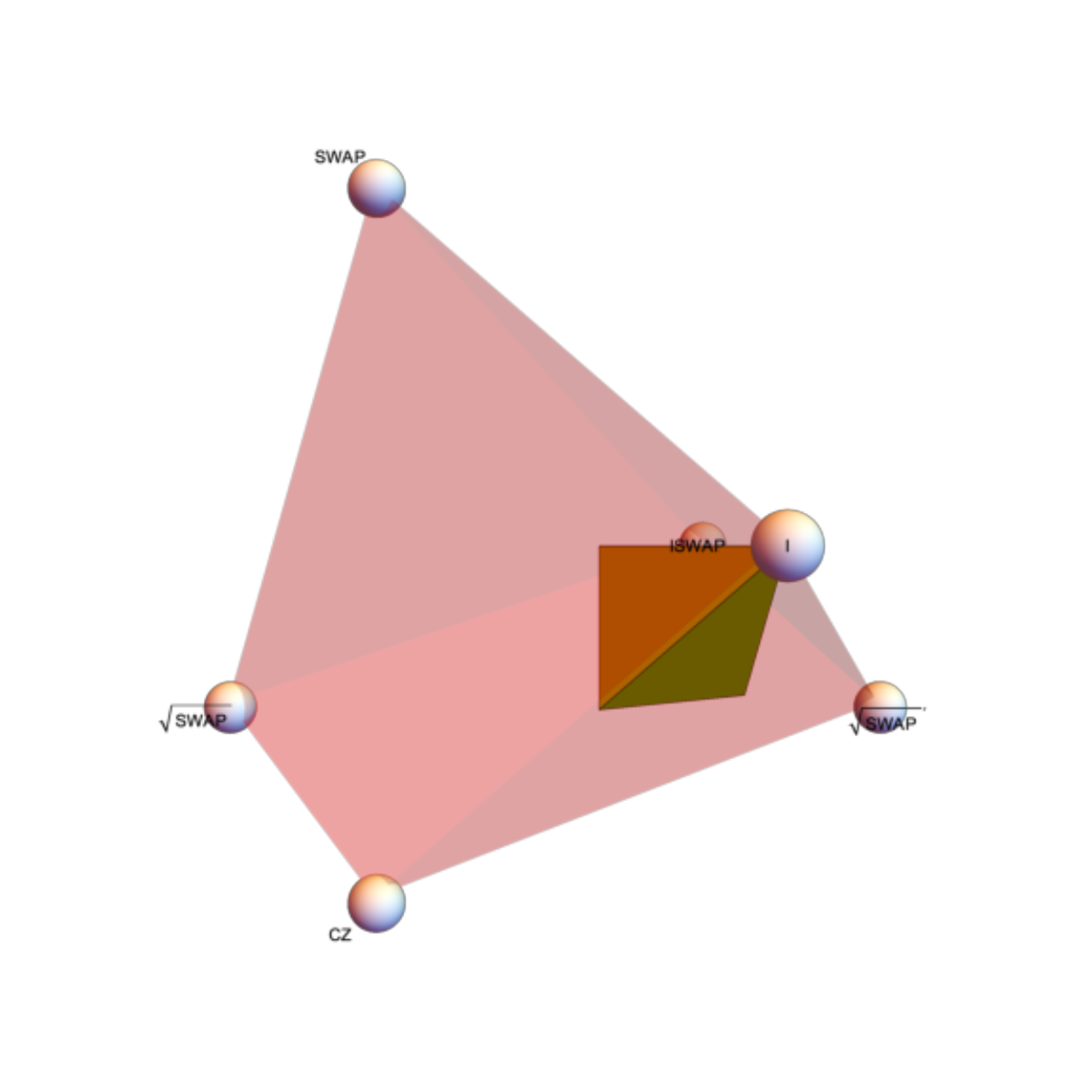}
    \includegraphics[width=0.2\textwidth,trim={3cm 1.5cm 3cm 1.5cm},clip]{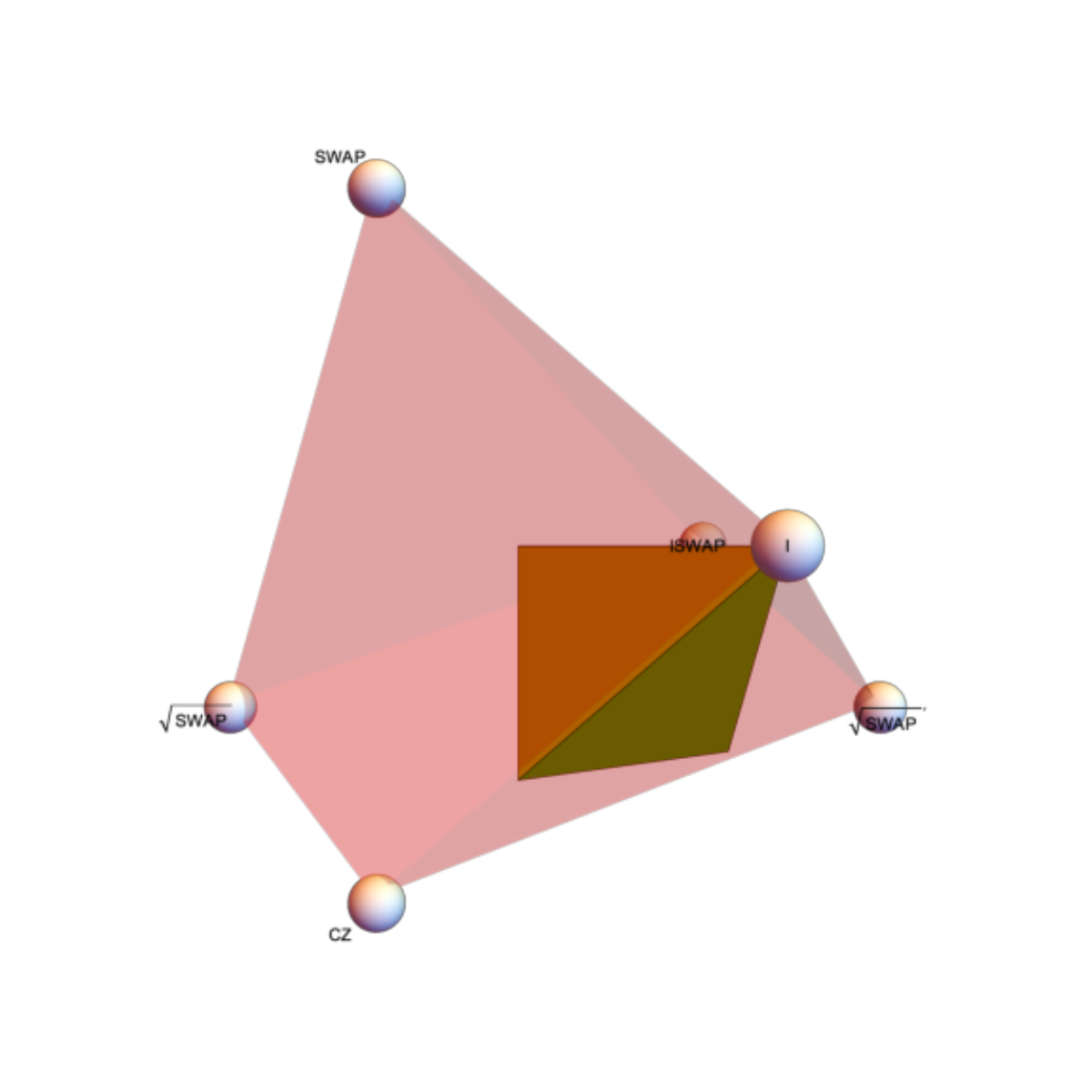}
    \includegraphics[width=0.2\textwidth,trim={3cm 1.5cm 3cm 1.5cm},clip]{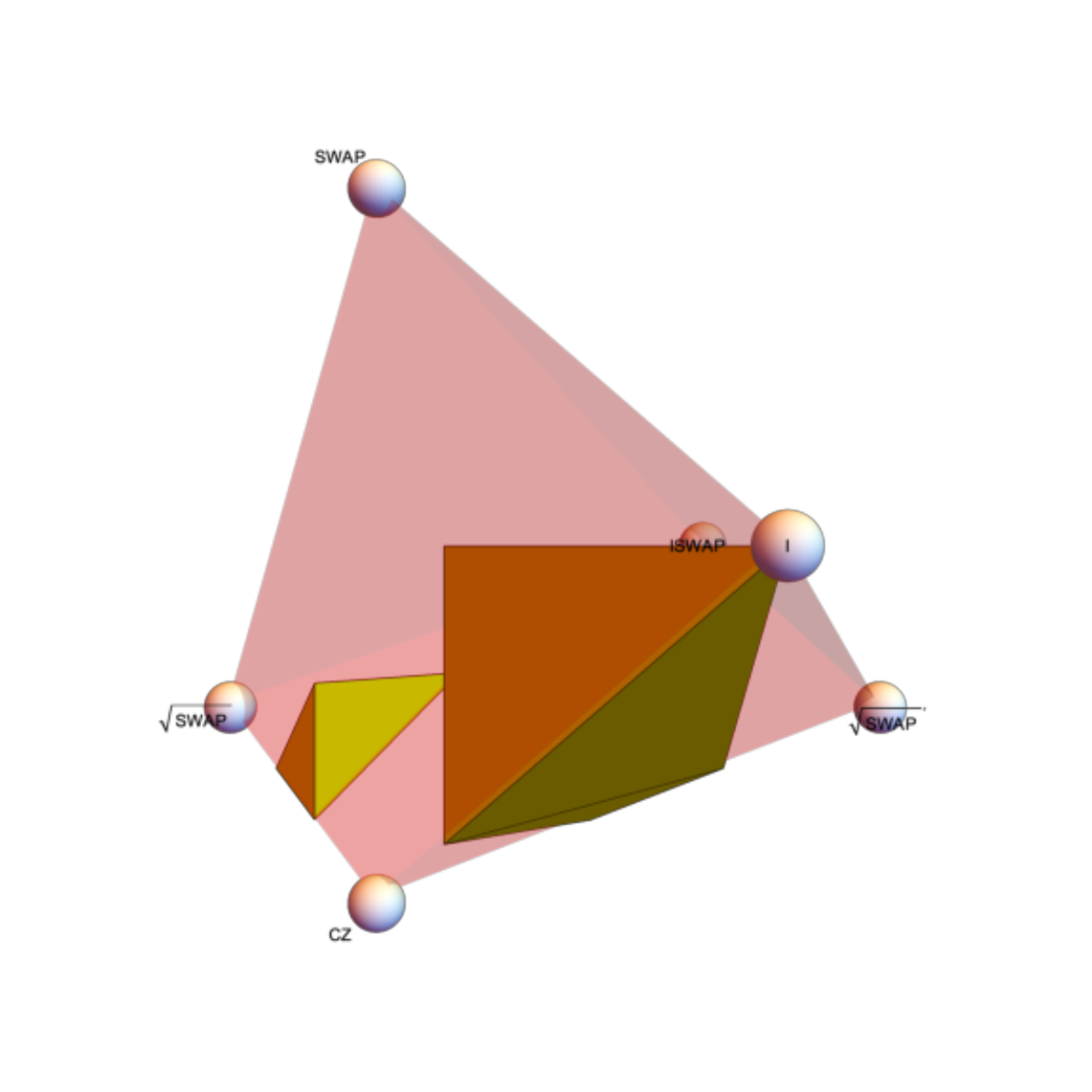}
    \includegraphics[width=0.2\textwidth,trim={3cm 1.5cm 3cm 1.5cm},clip]{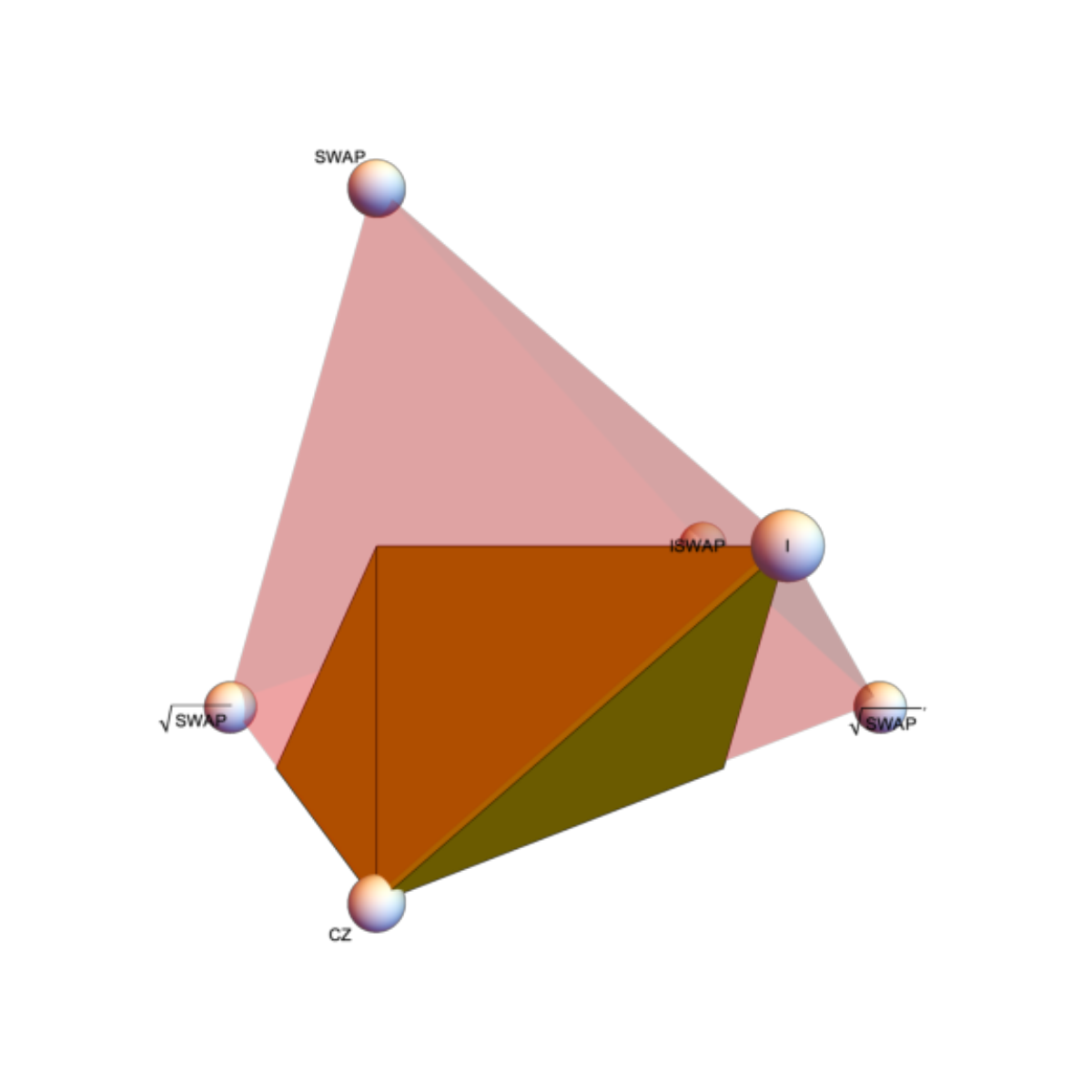}
    \includegraphics[width=0.2\textwidth,trim={3cm 1.5cm 3cm 1.5cm},clip]{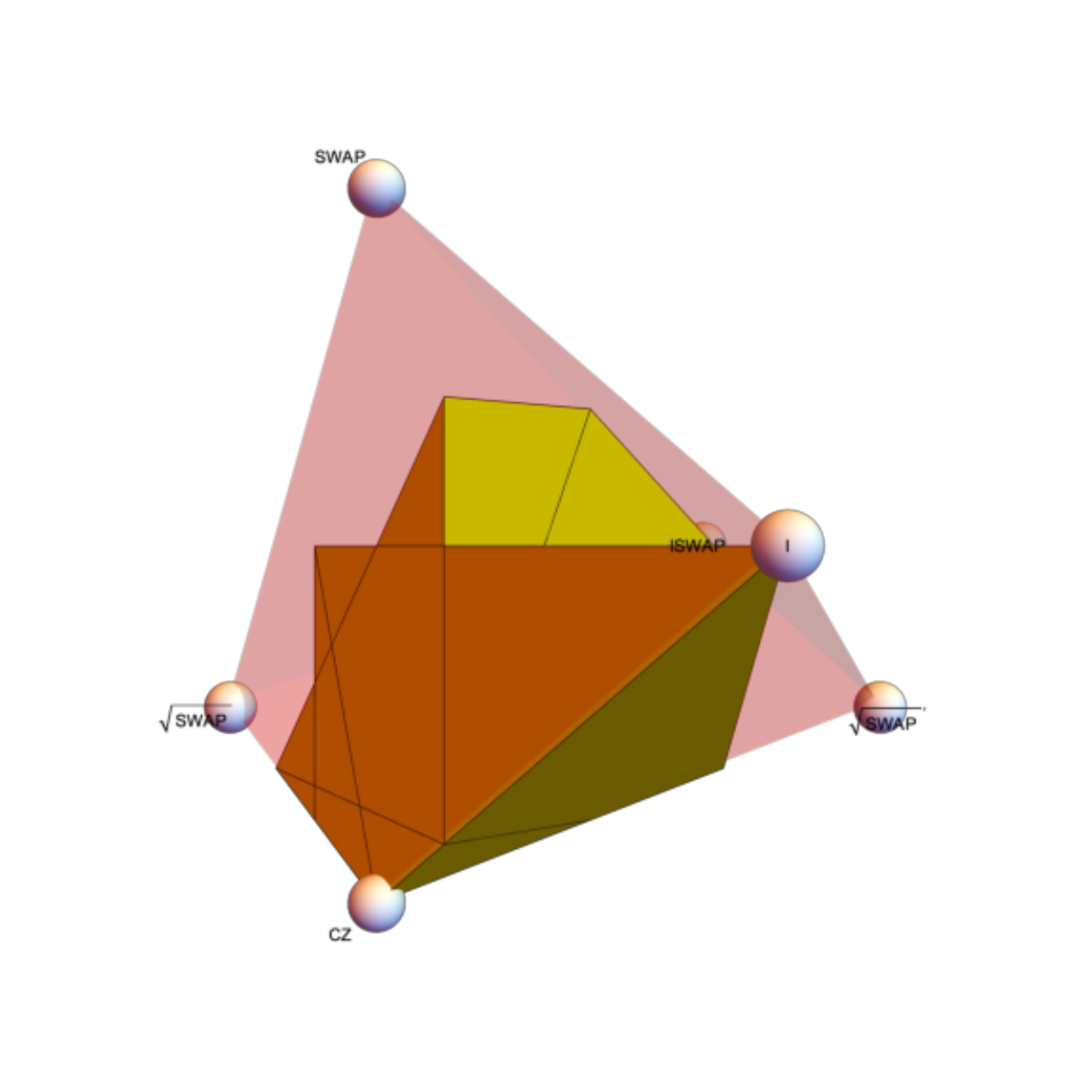}
    \includegraphics[width=0.2\textwidth,trim={3cm 1.5cm 3cm 1.5cm},clip]{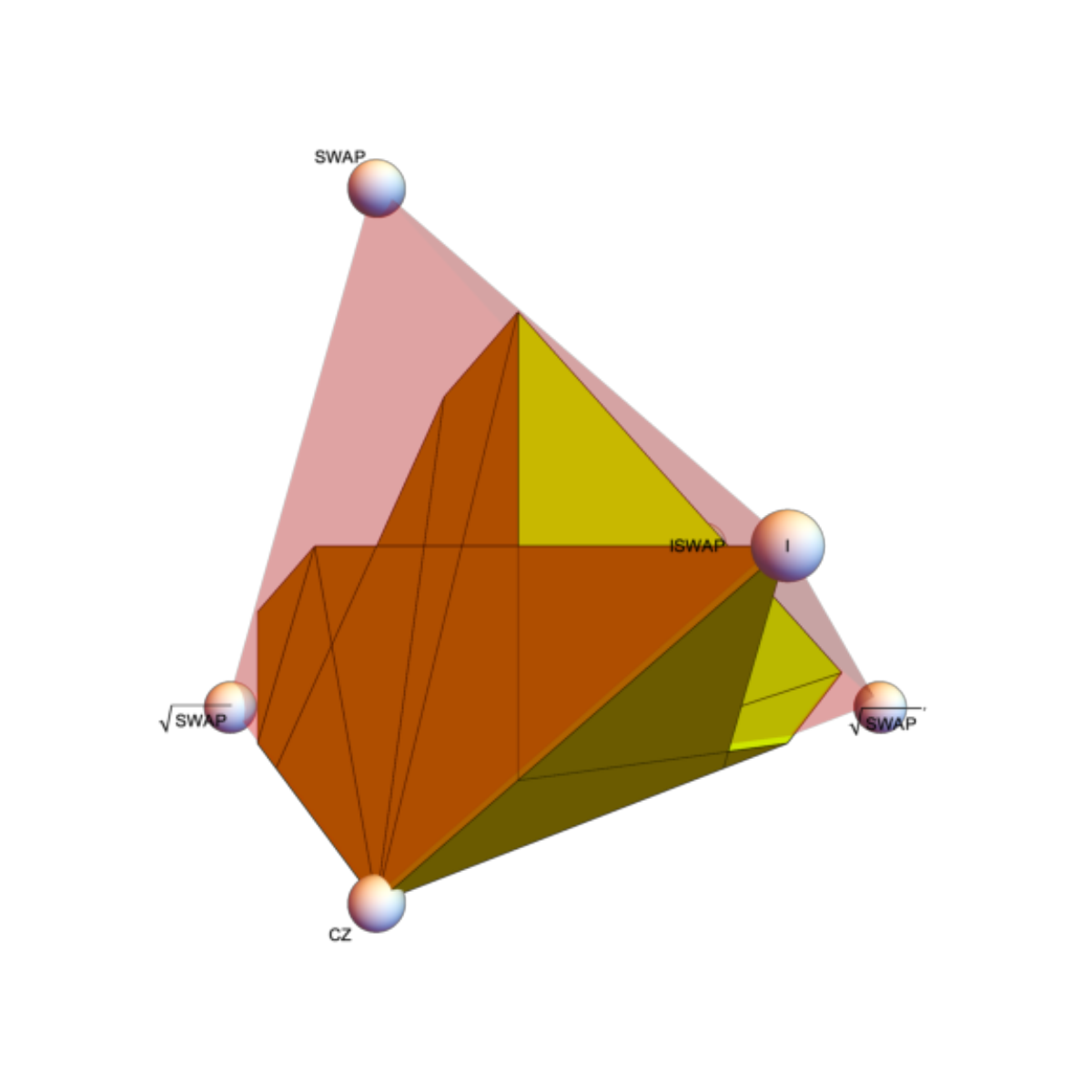}
    \includegraphics[width=0.2\textwidth,trim={3cm 1.5cm 3cm 1.5cm},clip]{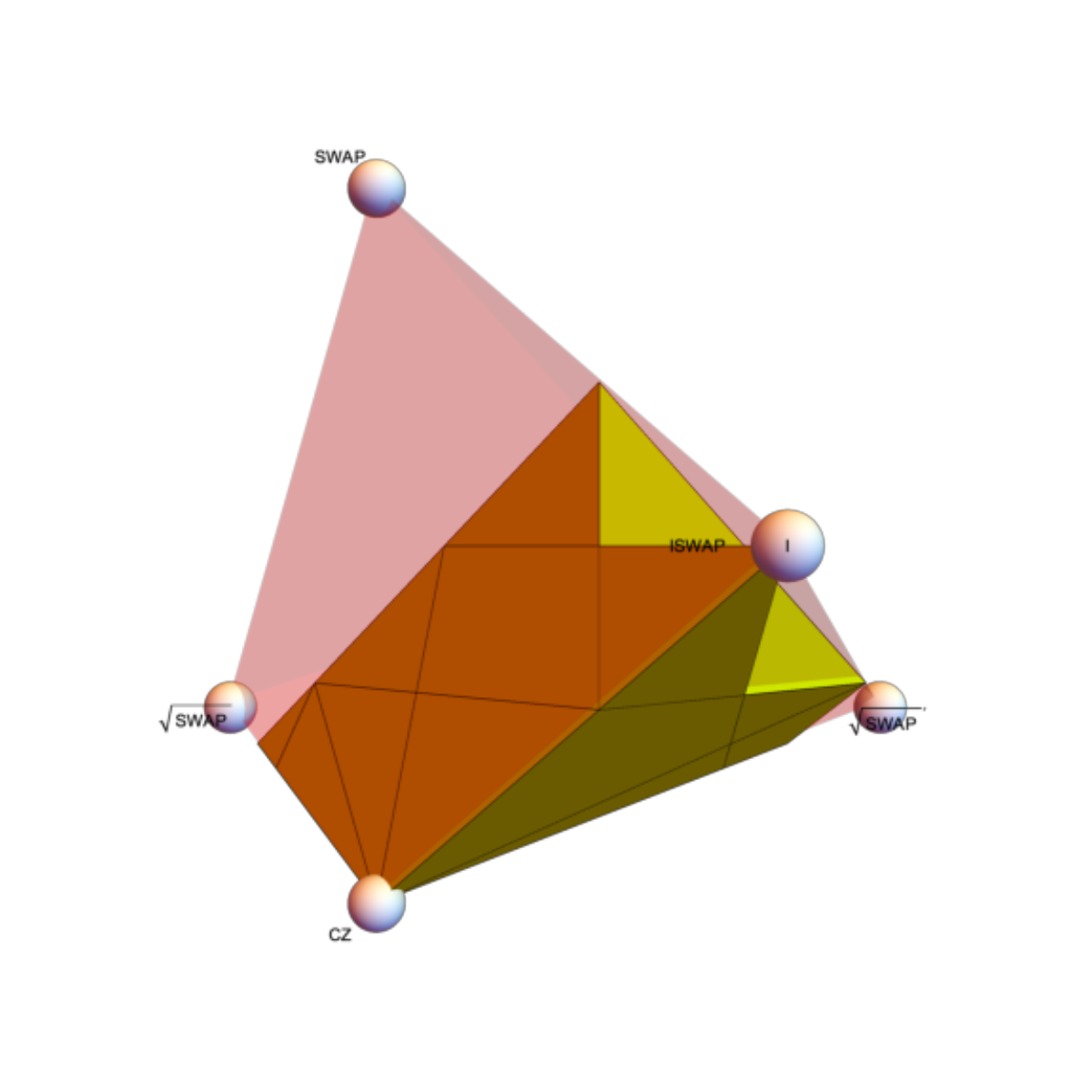}
    \includegraphics[width=0.2\textwidth,trim={3cm 1.5cm 3cm 1.5cm},clip]{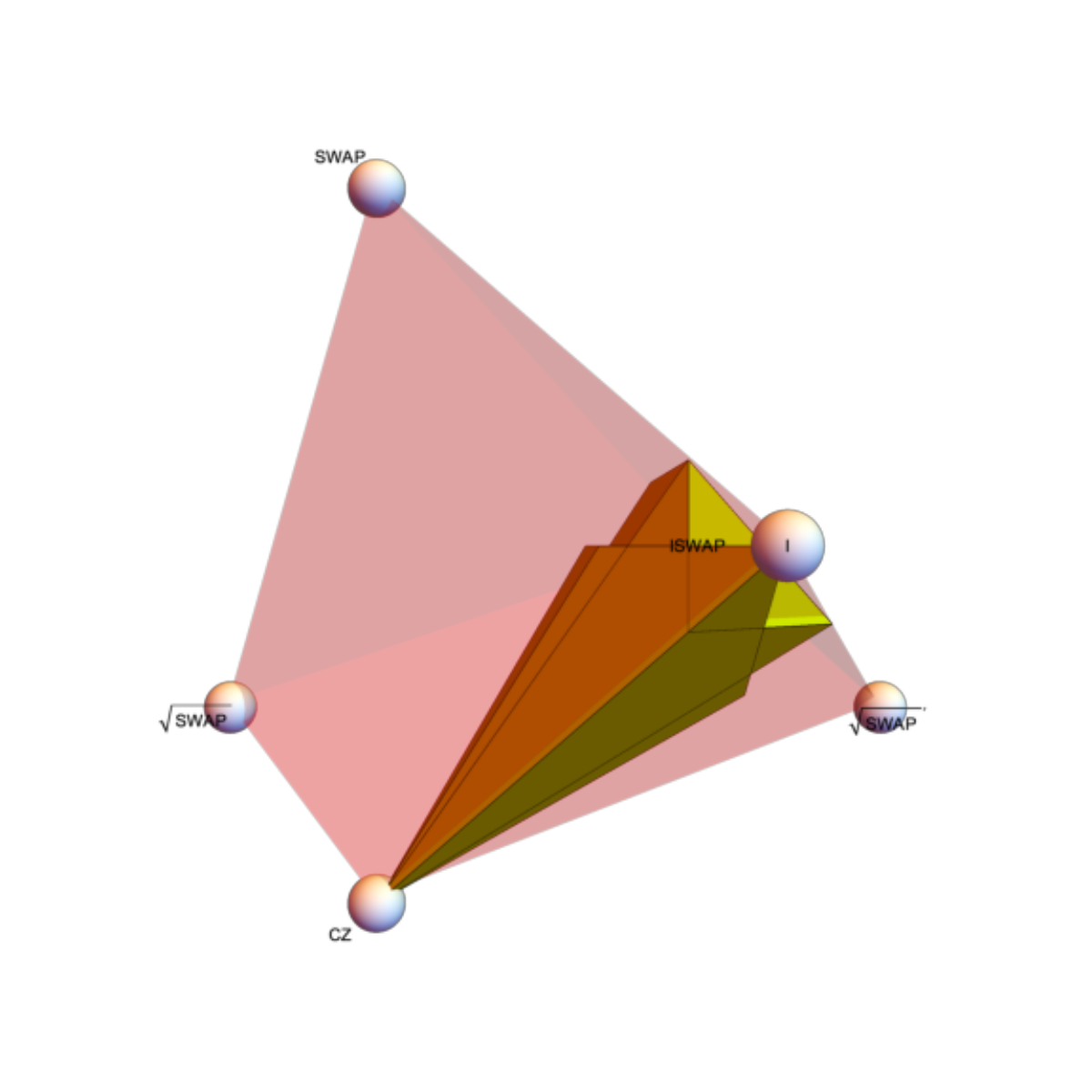}
    \caption{The solids $\LogCoords(P^2_{\XY_{\pi t}})$, as shown from a second perspective, for differing values of $t$: $2/10$, $3/10$, \ldots, $9/10$.}
    \label{P2XYFigure2}
\end{figure}

\begin{figure}
    \centering
    \includegraphics[width=0.2\textwidth,trim={1.2cm 2cm 5cm 2cm},clip]{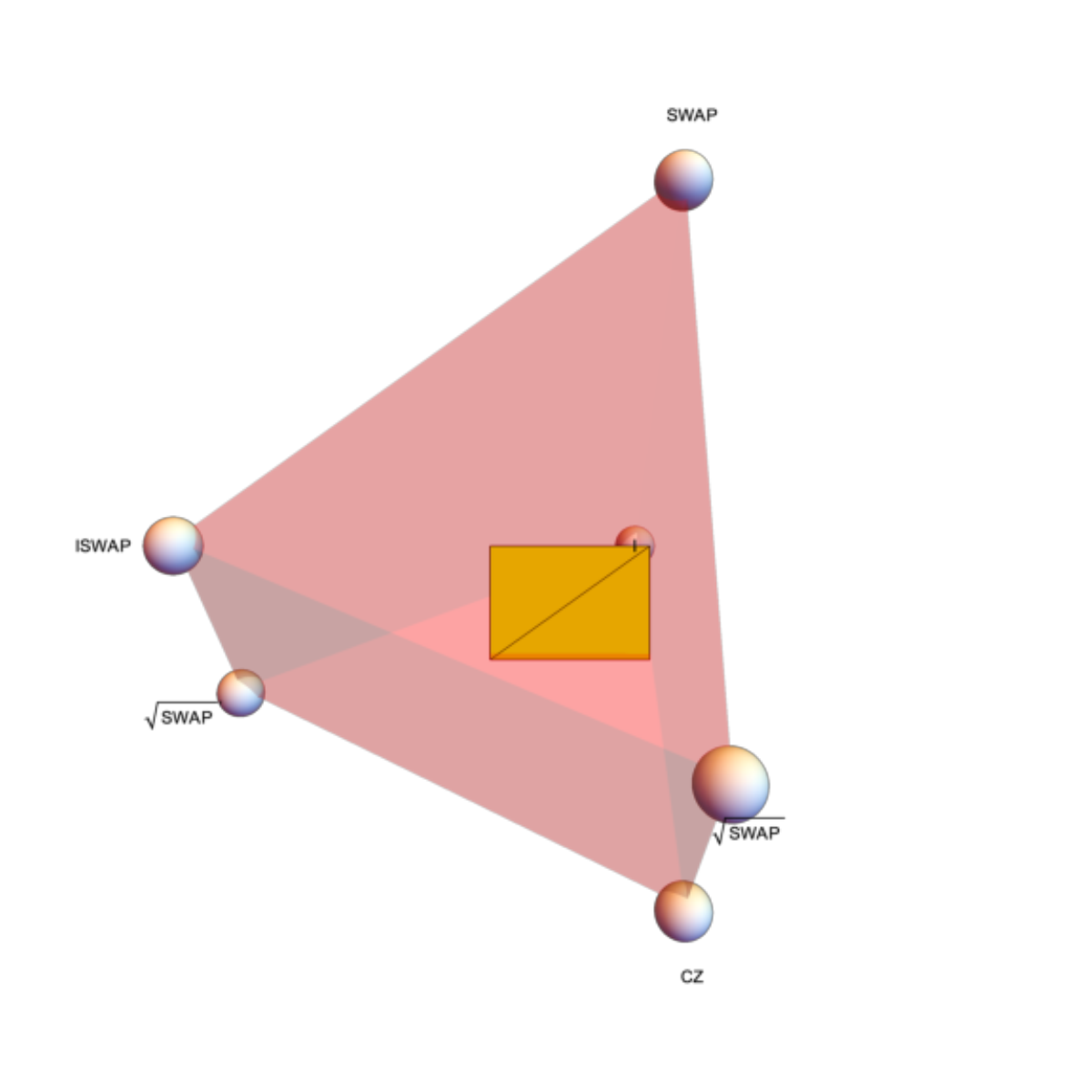}
    \includegraphics[width=0.2\textwidth,trim={1.2cm 2cm 5cm 2cm},clip]{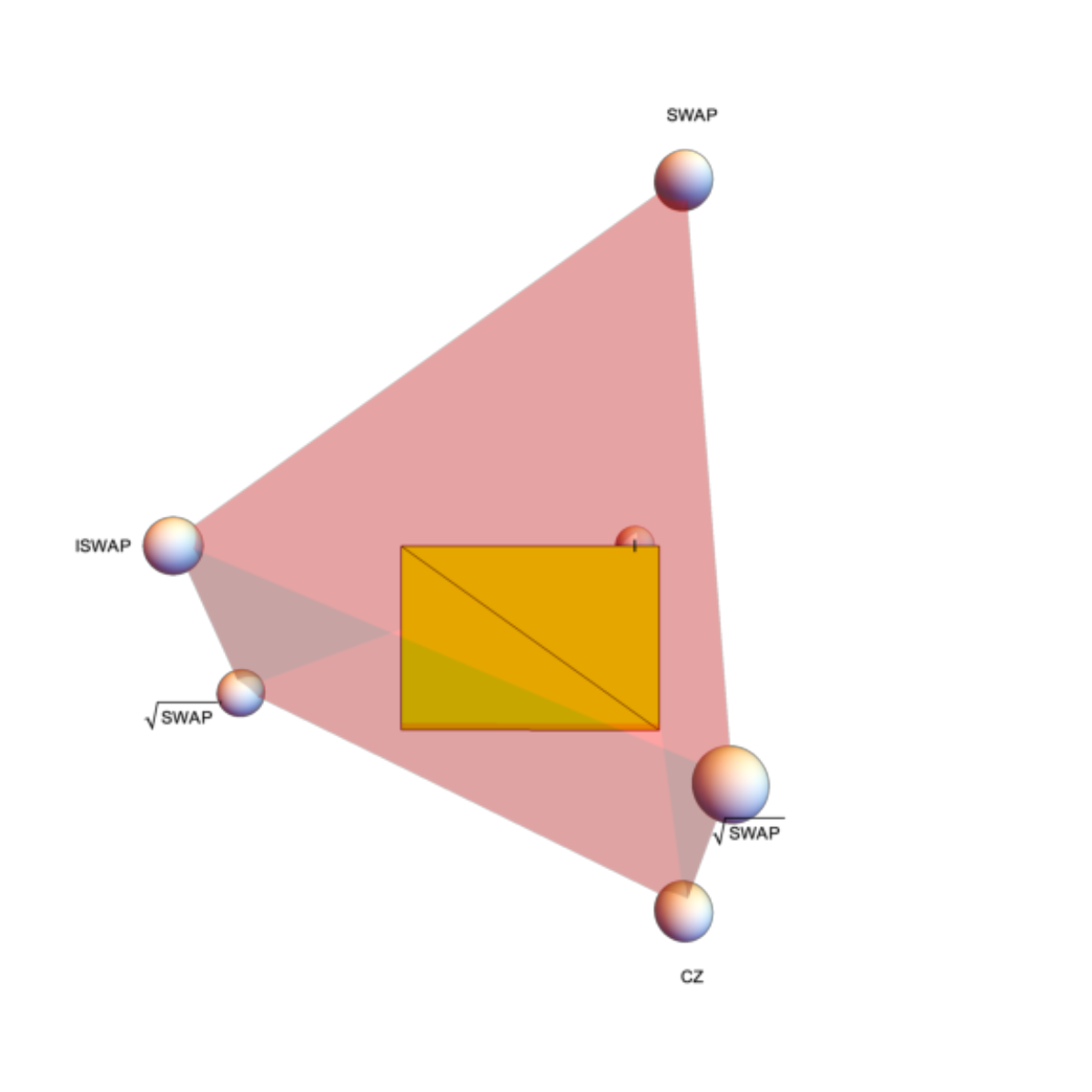}
    \includegraphics[width=0.2\textwidth,trim={1.2cm 2cm 5cm 2cm},clip]{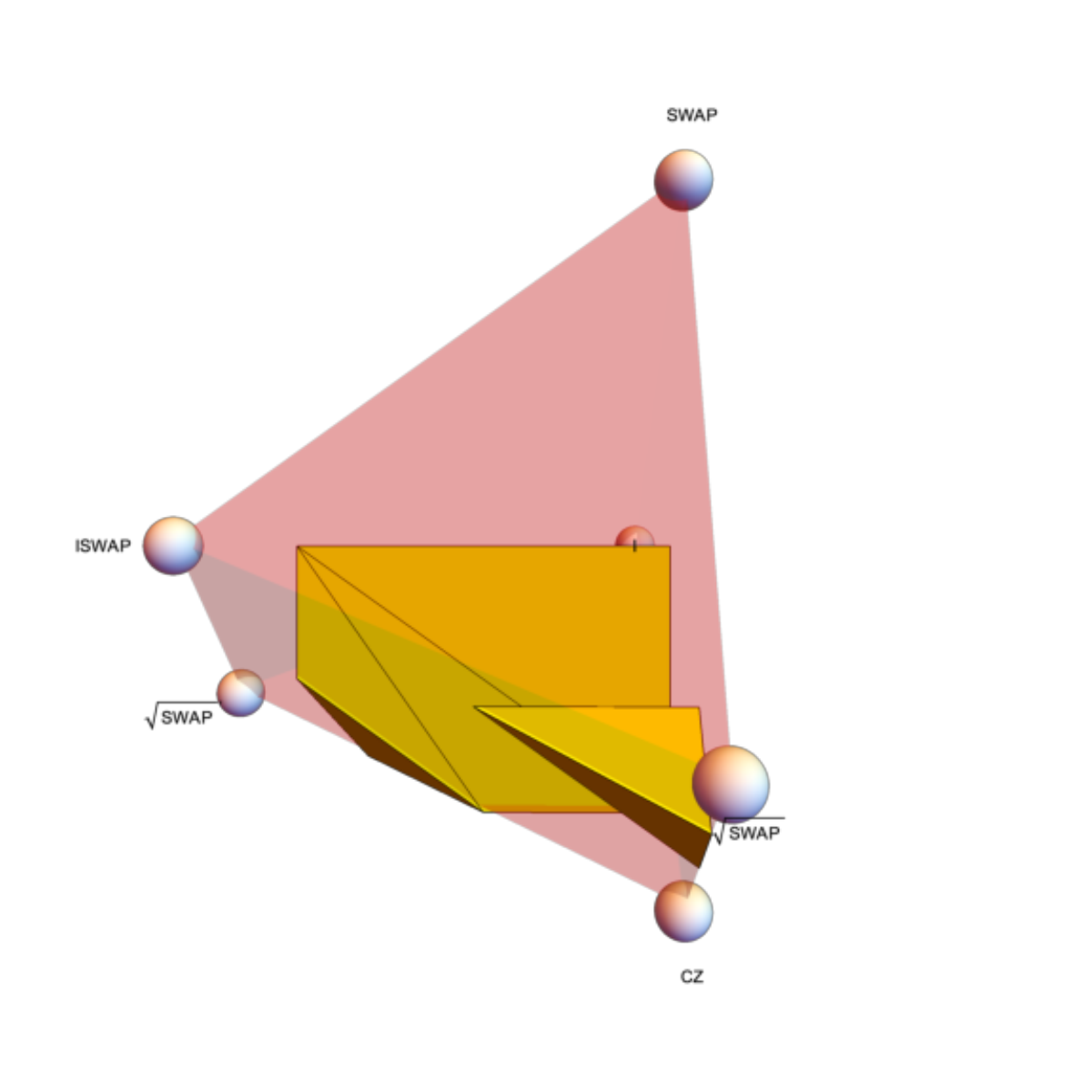}
    \includegraphics[width=0.2\textwidth,trim={1.2cm 2cm 5cm 2cm},clip]{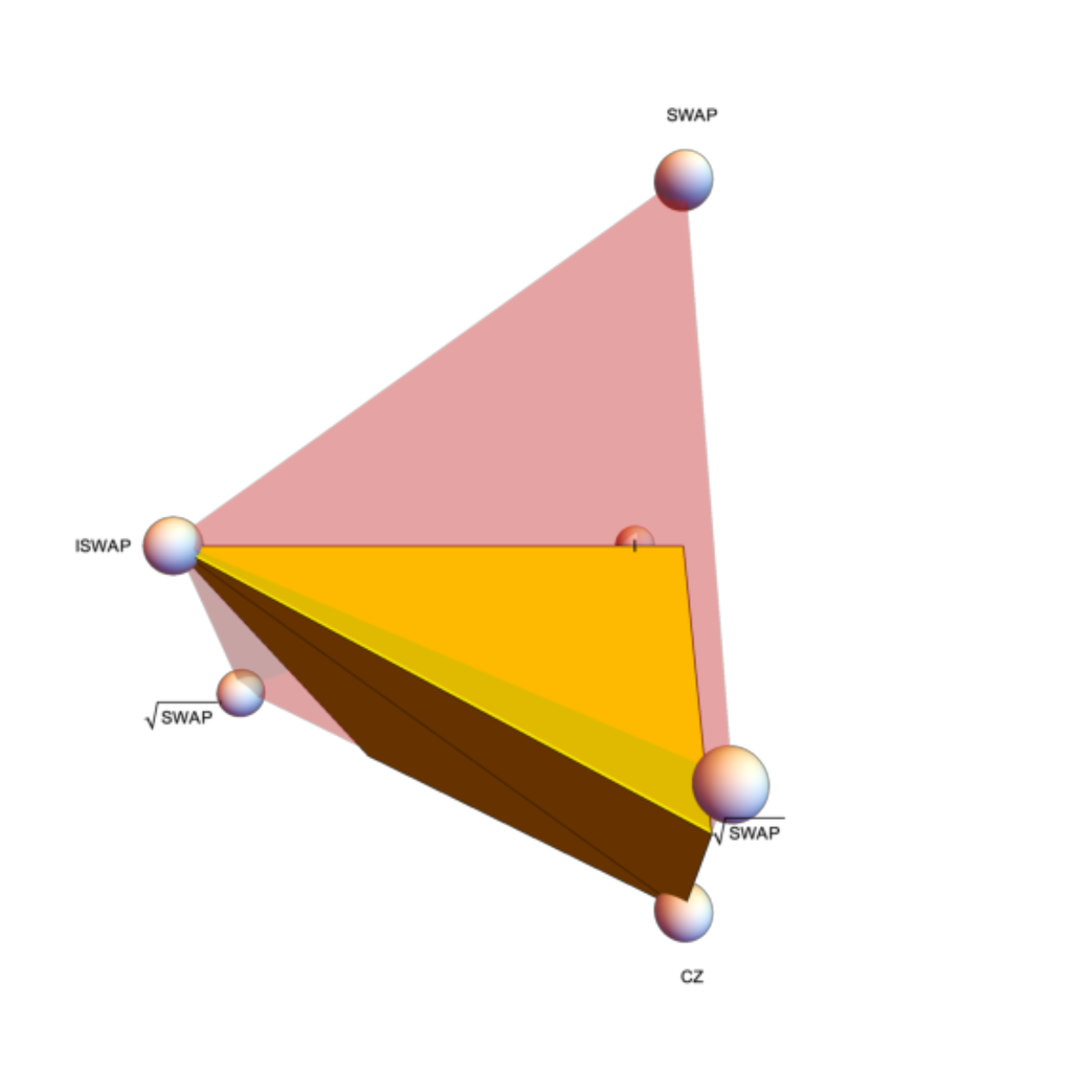}
    \includegraphics[width=0.2\textwidth,trim={1.2cm 2cm 5cm 2cm},clip]{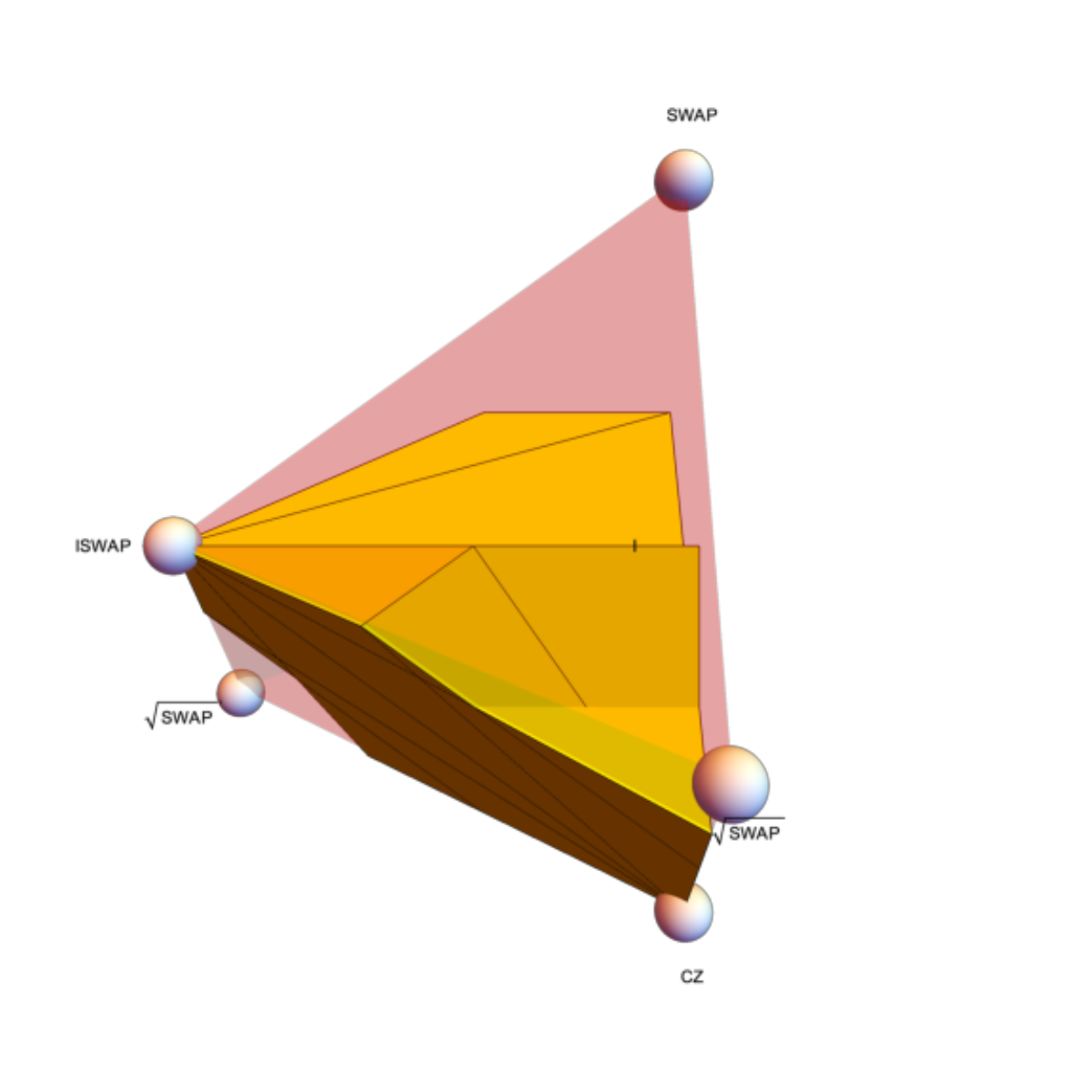}
    \includegraphics[width=0.2\textwidth,trim={1.2cm 2cm 5cm 2cm},clip]{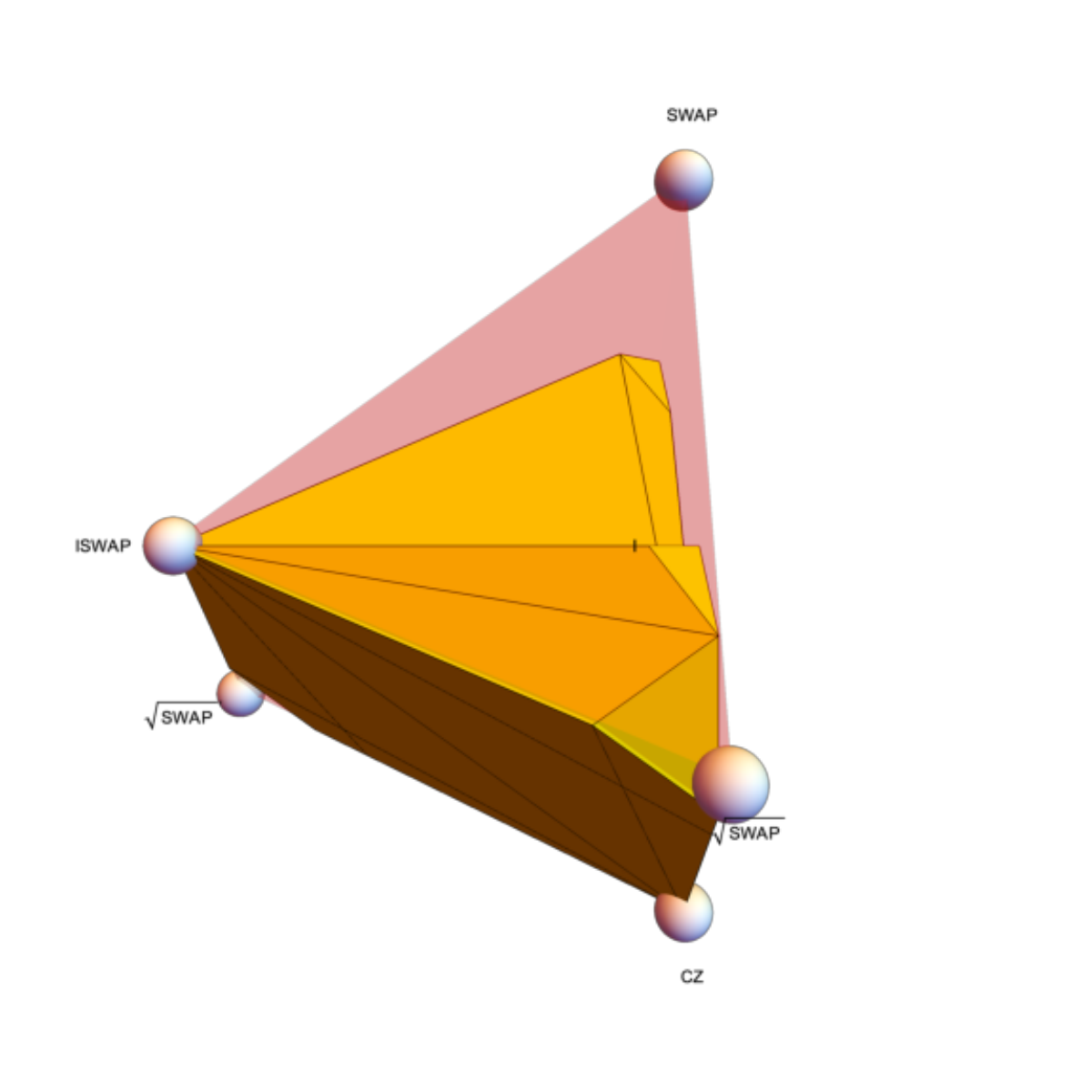}
    \includegraphics[width=0.2\textwidth,trim={1.2cm 2cm 5cm 2cm},clip]{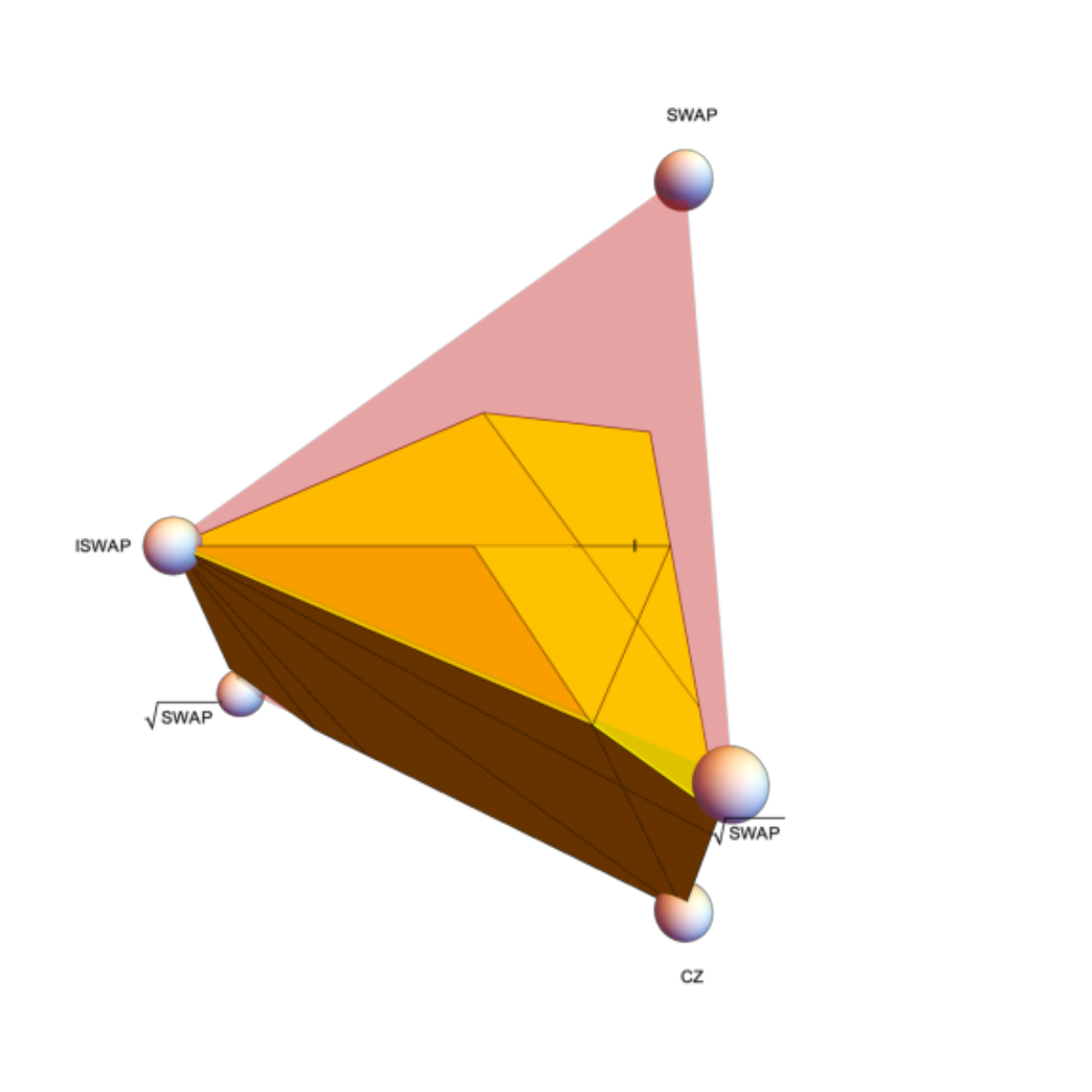}
    \includegraphics[width=0.2\textwidth,trim={1.2cm 2cm 5cm 2cm},clip]{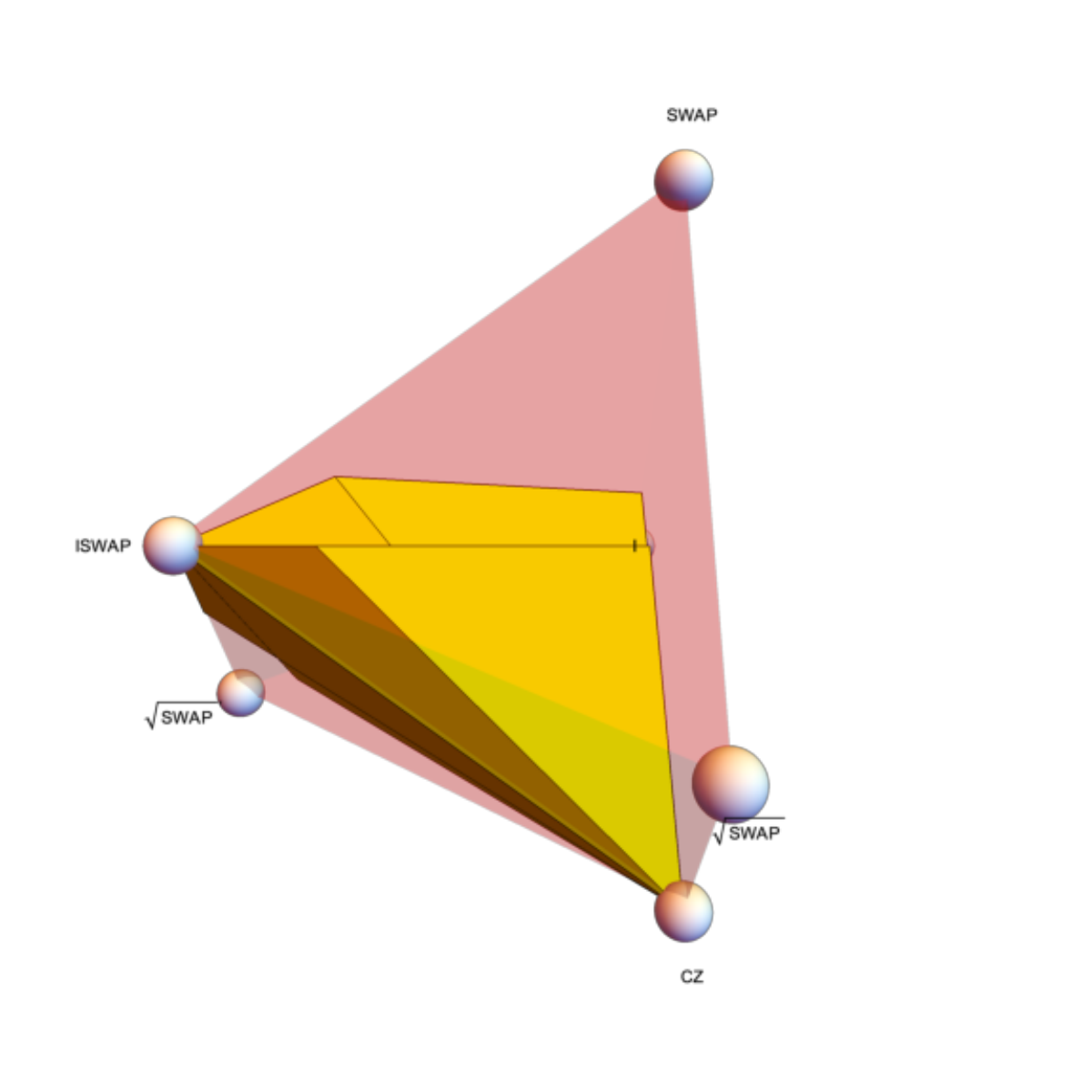}
    \caption{The solids $\LogCoords(P^2_{\XY_{\pi t}})$, as shown from a third perspective, for differing values of $t$: $2/10$, $3/10$, \ldots, $9/10$.}
    \label{P2XYFigure3}
\end{figure}

It is then of further interest to give a precise description of the polytope $\LogCoords(P^2_{\DB})$.

\begin{figure}
    \centering
    \includegraphics[width=0.3\textwidth, trim={3.5cm 3.5cm 3.5cm 3.5cm}]{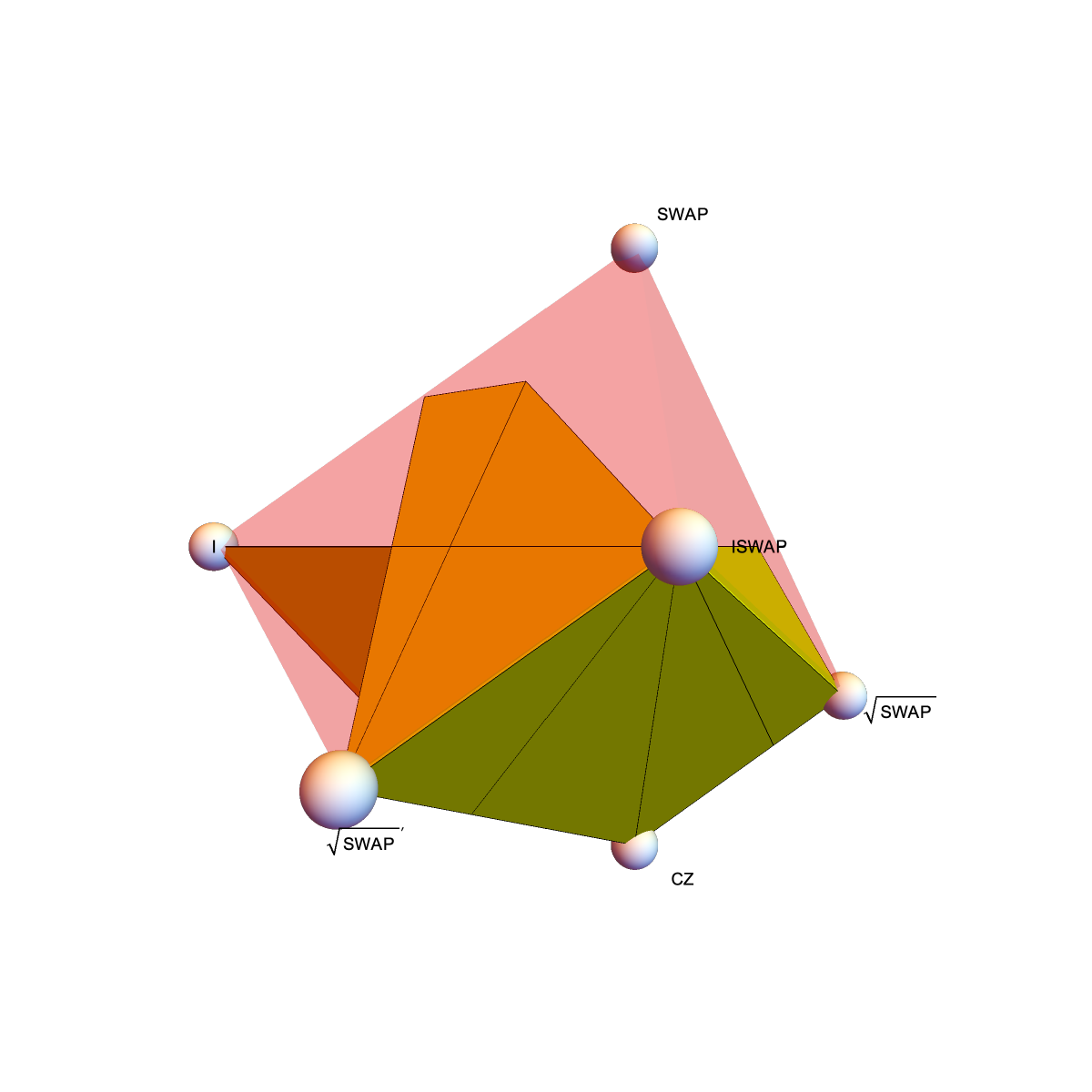}
    \includegraphics[width=0.35\textwidth,trim={3.5cm 3.5cm 3.5cm 3.5cm}]{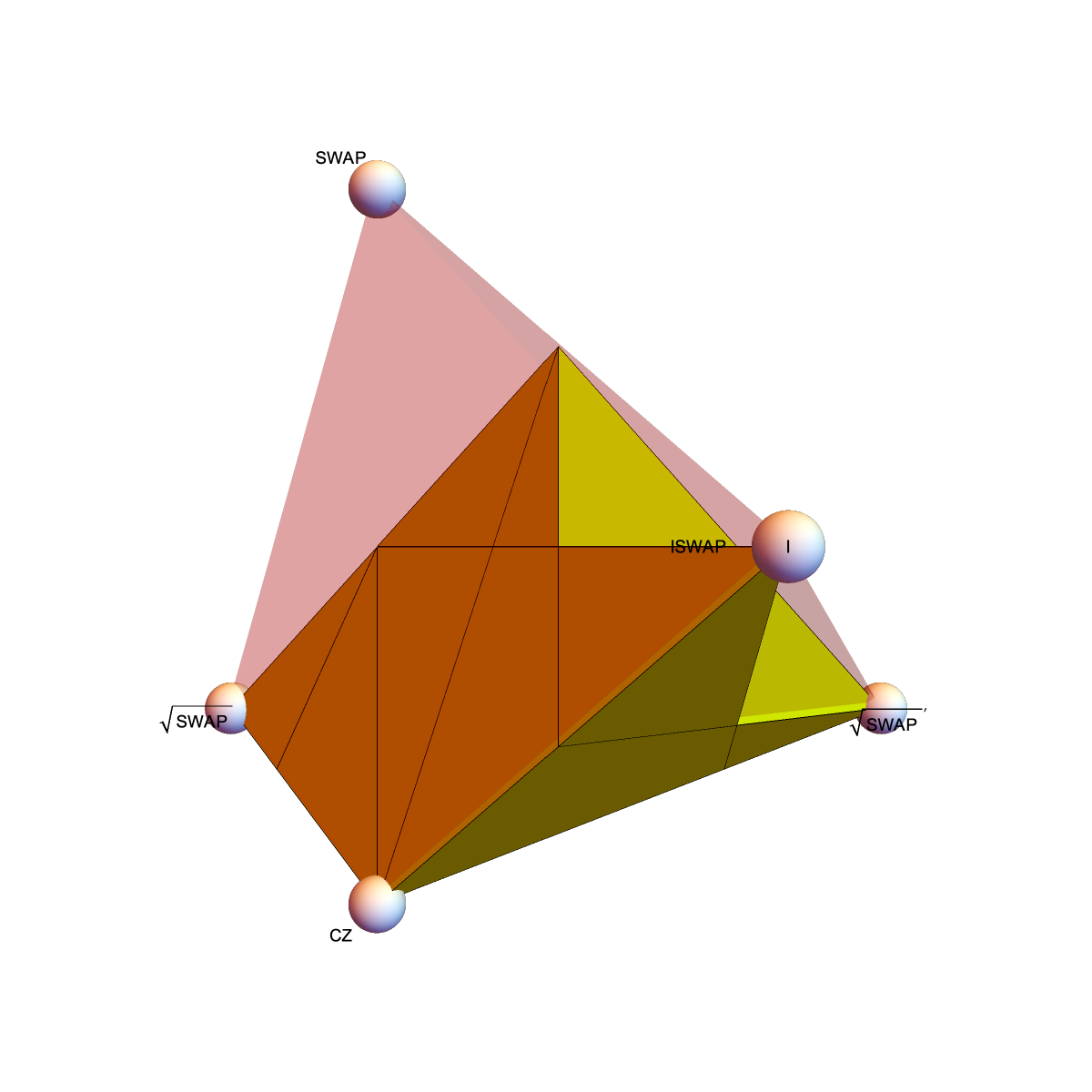}
    \includegraphics[width=0.3\textwidth, trim={3.5cm 3.5cm 3.5cm 3.5cm}]{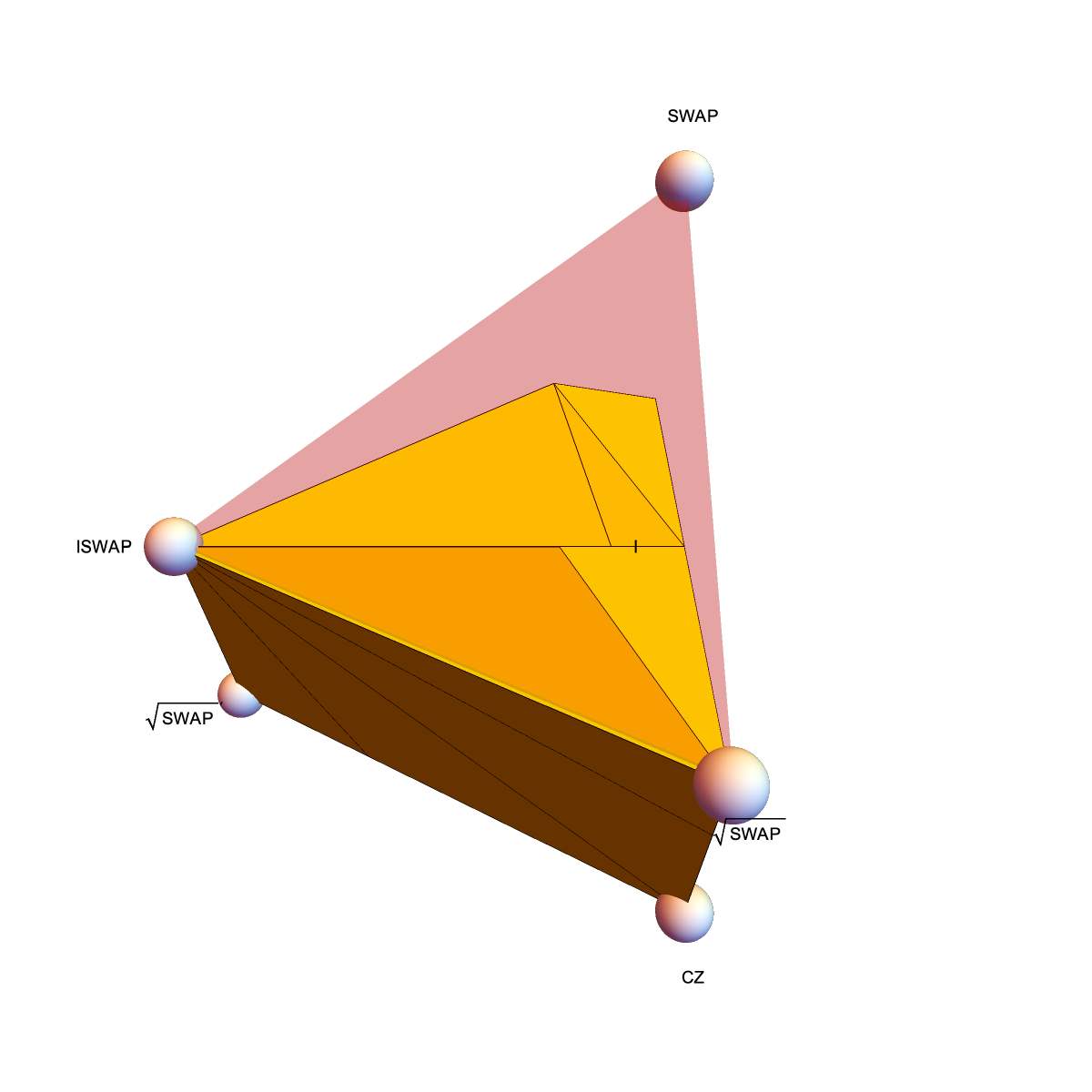} \vspace{\baselineskip}
    \caption{Three views of $\LogCoords(P^2_{\DB})$}
    \label{P2DBViews}
\end{figure}

\begin{lemma}\label{P2DBDescription}
$\LogCoords(P^2_{\DB})$ is a union of two convex polytopes, respectively described the following two families of inequalities:
\[
    \left\{\delta_2 \ge 0, \;\;
    \frac{1}{4} \ge |\delta_2 + \delta_3|, \;\;
    0 \ge \delta_3\right\},
\] \[
    \left\{\frac{1}{2} \ge \delta_1 + \delta_2 + \delta_3, \;\;
    -\frac{1}{4} + \delta_1 + \delta_2 \ge 0, \right.
\] \[
    \left. \frac{1}{4} \ge |\delta_2 + \delta_3|\right\}
\]
intersected with the fundamental alcove.  The extremal points can be found in \Cref{VerticesOfP2DB}.
\end{lemma}
\begin{proof}
Here we follow the method espoused by \Cref{ArduousMechanismRem}.  The family of inequalities comes directly from reducing the family supplied by \Cref{MainAWBTheorem}.  After calculating all of the points of intersection of the associated equalities and discarding those intersection points which do not satisfy all of the inequalities, the remainder is the set of extremal vertices, as listed above.
\end{proof}

\begin{remark}
The polytope $\LogCoords(P^2_{\DB})$ is pictured from three angles in \Cref{P2DBViews}.
\end{remark}

\begin{remark}
For the interested reader, we also include as \Cref{HeatMapFigure} a depiction a numerical sampling of $\vol \LogCoords(P^2_X)$ as $X$ ranges over (the facets of) the entire monodromy polytope.  Points shaded black correspond to those values of $X$ for which $\vol \LogCoords(P^2_X)$ is at $0\%$ of the total volume, and points shaded white correspond to $100\%$ of the total volume.  In the middle figure, the heat values along the line connecting the westernmost point, labeled $\I$, to the center point, labeled $\ISWAP$, correspond to the graph depicted in \Cref{SymmetricXYVolumeFig}.
\end{remark}

\begin{figure}
    \centering
    \includegraphics[width=0.3\textwidth, trim={3.5cm 3.5cm 3.5cm 3.5cm}]{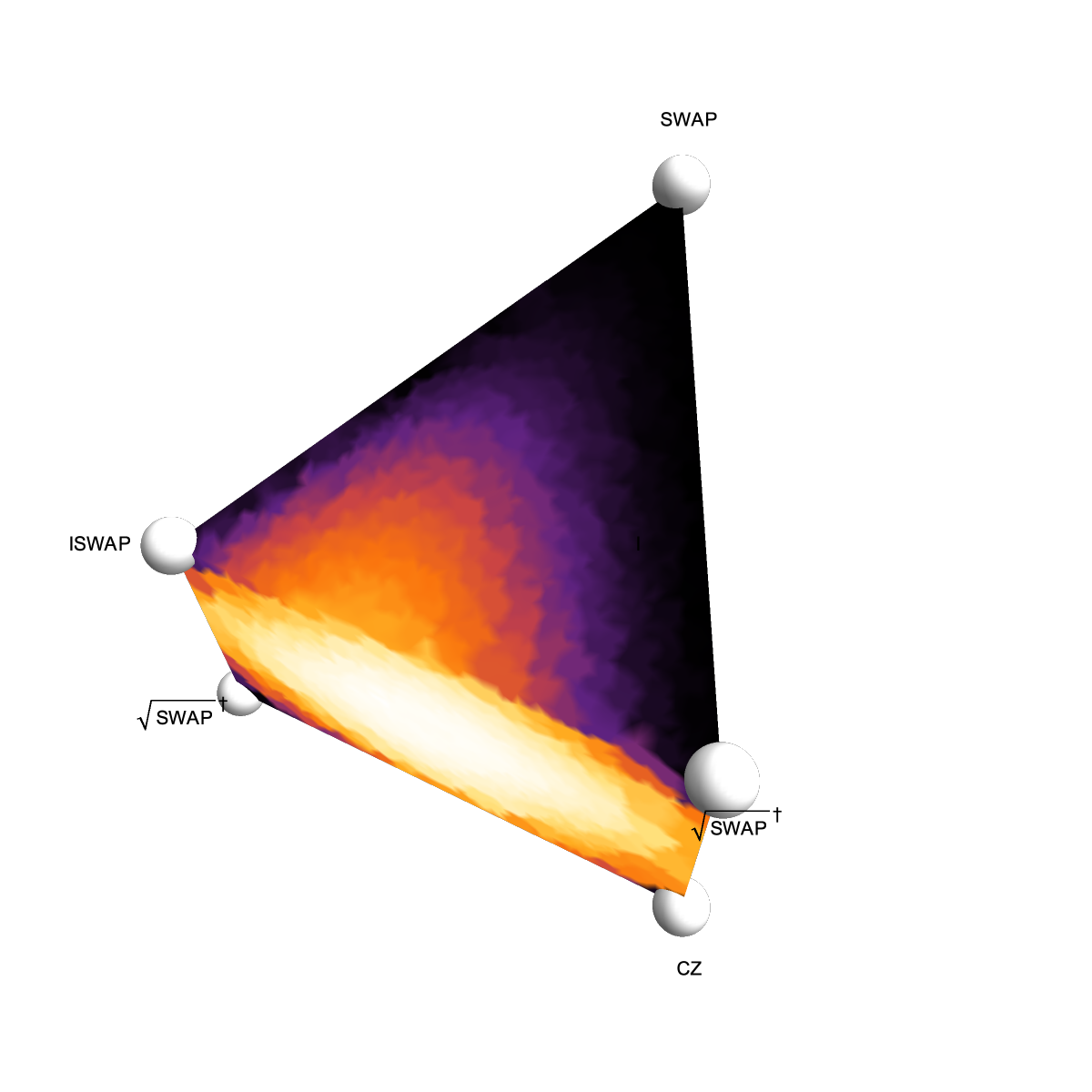}
    \includegraphics[width=0.35\textwidth,trim={3.5cm 3.5cm 3.5cm 3.5cm}]{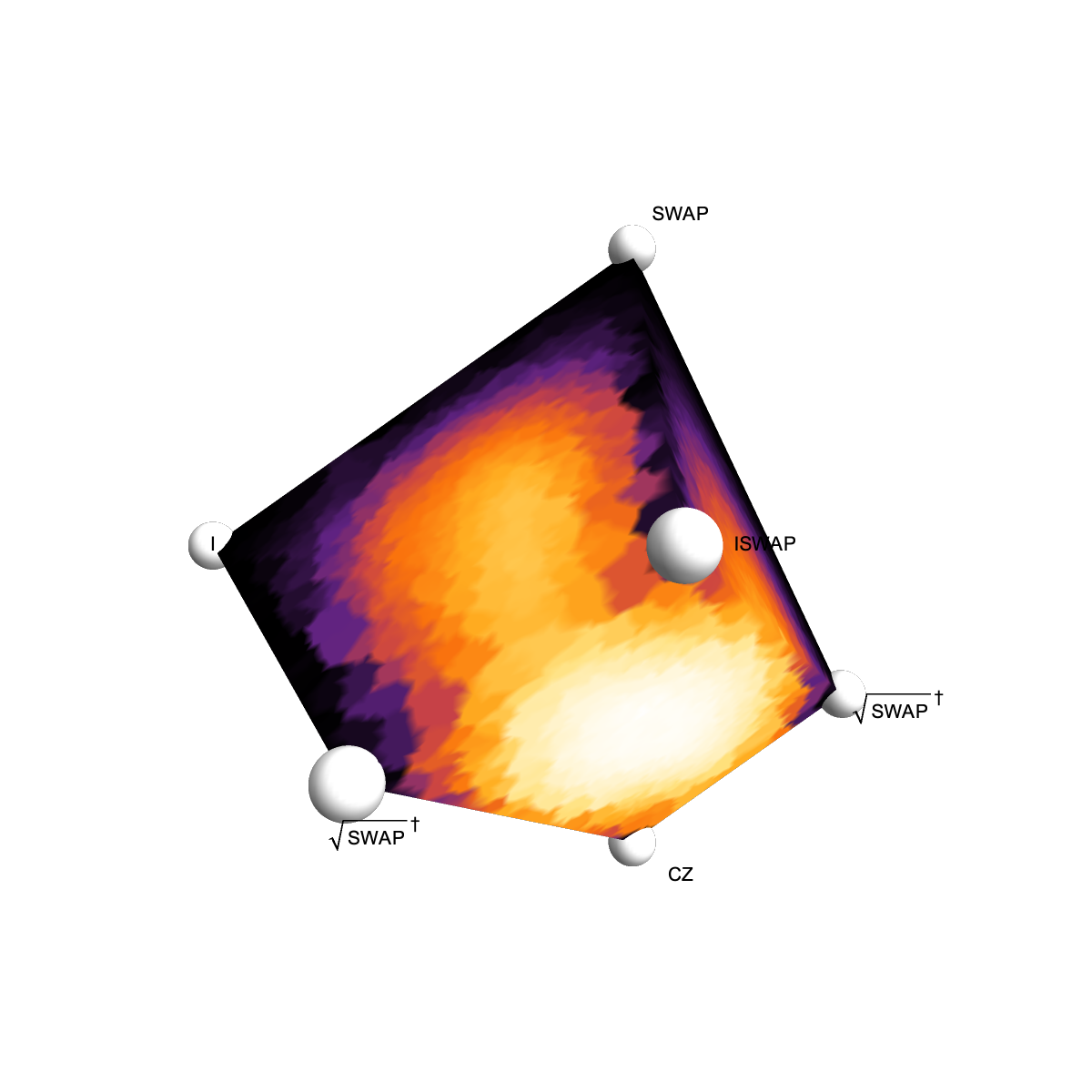}
    \includegraphics[width=0.3\textwidth, trim={3.5cm 3.5cm 3.5cm 3.5cm}]{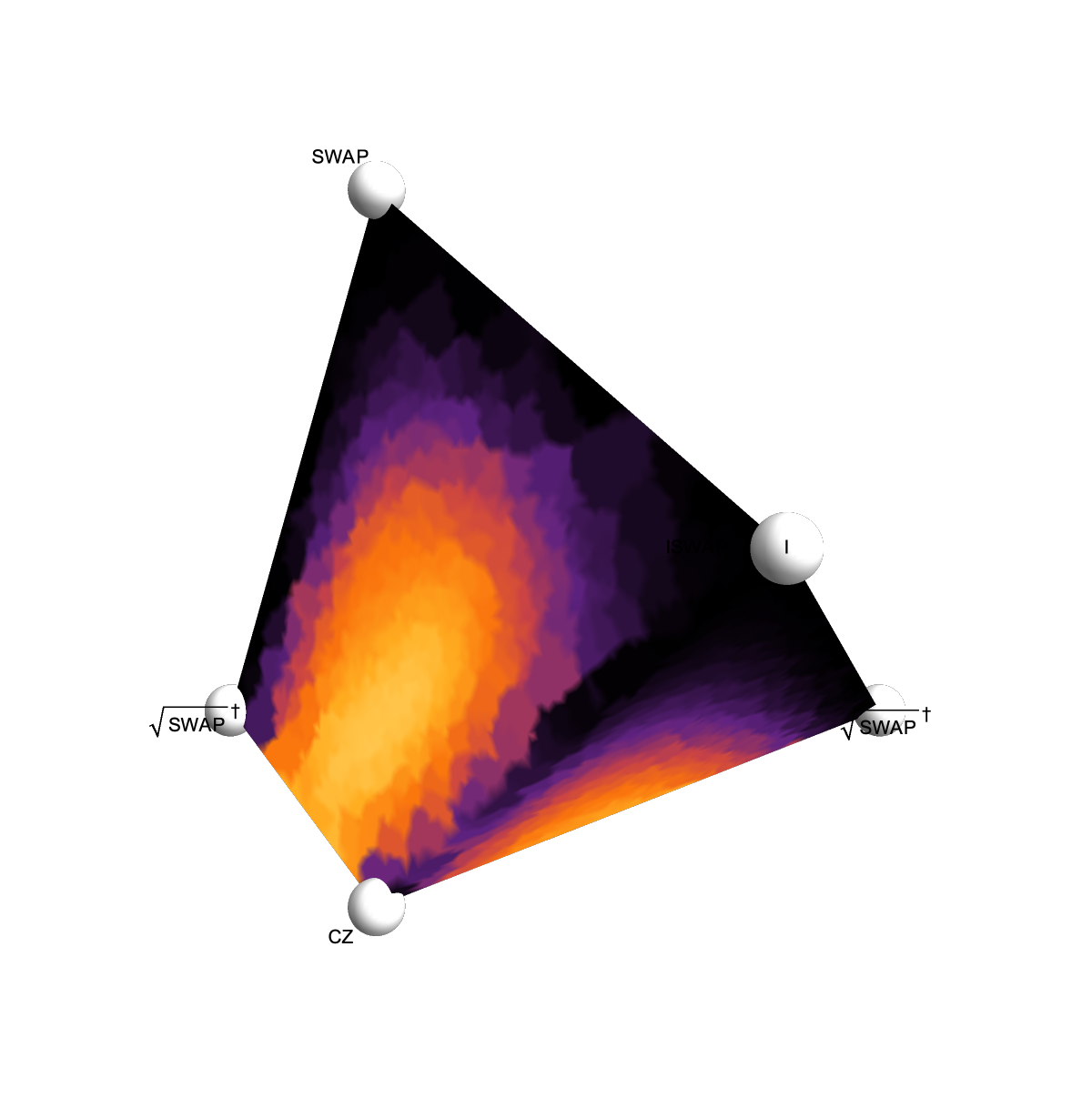}
    \caption{An approximate heat map of volumes of symmetric monodromy polytope slices.  The vertices of the figures are labeled $\I$, $\CZ$, $\ISWAP$, $\SWAP$, and $\sqrt{\SWAP}$.  A gate $U$ is shaded according to the volume of $\LogCoords(P^2_U)$: black is $0\%$ of the total volume of $\A_{C_2}$, and white is $100\%$.}
    \label{HeatMapFigure}
\end{figure}

\begin{remark}\label{NoLinearLunch}
In the course of our analysis of $\<\L_{\XY}\>^{\A_{C_2}}$, we have avoided giving an effective compilation routine for $P^2_{\XY}$ along the lines of \Cref{CZNailsCAN} for $P^3_{\CZ}$.  Indeed, one can show under mild hypotheses that such a formula cannot exist.  The canonical family forms a $3$--dimensional maximal torus in $PU(4)$, but the maximal tori in $PU(2)^{\otimes 2}$ are merely $2$--dimensional.  It is therefore impossible for any gate $S \in PU(4)$ to conjugate a family of local gates onto the canonical family.  If $S^\dagger$ is additionally locally equivalent to $S$, then this means that any analogue of \Cref{CZNailsCAN} for $P^2_S$ must instead be of the form
\begin{center}
\begin{tikzcd}[column sep=0.5em]
& \gate[2]{\CAN(\alpha, \beta, \delta)} & \qw \\
& \qw & \qw
\end{tikzcd}
=
\begin{tikzcd}[column sep=1.0em]
& \gate{A} & \gate[2]{S} & \gate{B} & \gate[2]{S^\dagger} & \gate{C} & \qw \\
& \gate{D} & \qw & \gate{E} & \qw & \gate{F} & \qw,
\end{tikzcd}
\end{center}
where the outer gates $A$, $C$, $D$, and $F$ are \emph{not all constant} in the parameters $\alpha$, $\beta$, $\delta$.  In practice, it seems that $B$ and $E$ cannot be made linearly dependent on the canonical parameters either, though we do not presently have a proof of this to offer.
\end{remark}

\section{Approximate compilation}\label{ApproxCompSection}

We now use the above descriptions of the polytopes $\LogCoords(P^n_\S)$ to address the problem of \textit{approximate compilation}:

\begin{problem}\label{ApproximationProblem}
Given a two-qubit program $U$ and a gate set $\S$ whose members $S \in \S$ have associated fidelity estimates $f_S$, what circuit drawn from $\S$ gives the greatest fidelity approximation to $U$?
\end{problem}

For instance, in this specific setting of $\S = \{\CZ\}$, \Cref{P3CZDescription} shows that every such $U$ can be written as a circuit involving three applications of $\CZ$, whereas \Cref{P2CZDescription} shows that almost no $U$ can be decomposed exactly using just two applications of $\CZ$.  Nonetheless, if there is an associated cost to each application of $\CZ$, it may be preferable to deliberately ``miss'' $U$ (and thereby incur deliberate error) if it affords an opportunity to avoid applying $\CZ$ a third time (and thereby avoid indeliberate error).  This idea of approximate compilation is not a new one~\cite[Appendix B]{CBSNG}, and we begin by recalling some useful results.

\begin{definition}[{\cite{Nielsen}, see also \cite{PMM}}]
Given a pair of two-qubit programs $G$ and $G'$, we define their \textit{average gate fidelity} to be
\begin{align*}
    F_{\avg}(G, G') & = \int_{\psi \in \mathbb P(V)} |\<\psi|G^\dagger G'|\psi\>|^2 \\
    & = \frac{4 + |\tr(G^\dagger G')|^2}{4 \cdot 5} \in [1/5, 1].
\end{align*}
\end{definition}

Because this comes down to a trace calculation, this value is especially easy to calculate for simultaneously diagonalizable gates, which includes pairs of canonical gates after conjugation by $Q$:

\begin{lemma}[{\cite[Equation B.8d]{CBSNG}}]\label{CanonicalFidelityCalcn}
Let $G = \CAN(\alpha, \beta, \delta)$ and $G' = \CAN(\alpha', \beta', \delta')$ be two canonical gates with parameter differences
\begin{align*}
    \Delta_\alpha & = \frac{\alpha' - \alpha}{2}, &
    \Delta_\beta  & = \frac{\beta'  - \beta}{2},  &
    \Delta_\delta & = \frac{\delta' - \delta}{2}.
\end{align*}
Their average gate fidelity is given by
\begin{align*}
    20 F_{\avg}(G, G') & = 4+16\left|\begin{array}{c}\cos \Delta_\alpha \cos \Delta_\beta \cos \Delta_\delta \\ + \\ i \sin \Delta_\alpha \sin \Delta_\beta \sin \Delta_\delta\end{array}\right|^2.
\end{align*}
\qed
\end{lemma}

In pursuit of \Cref{ApproximationProblem}, we are also interested in the effect of local gates on \Cref{CanonicalFidelityCalcn}.  Rewriting the trace in terms of $Q$--conjugates, we have
\begin{align*}
|\tr G^\dagger G'|^2 & = |\tr L_2^\dagger C^\dagger L_1^\dagger \cdot L_1' C' L_2' |^2 \\
& = |\tr D_1 O_1 D_2 O_2|^2,
\end{align*}
where $D_1 = (C^\dagger)^Q$ and $D_2 = (C')^Q$ are diagonal gates and $O_1 = (L_1^\dagger L_1')^Q$ and $O_2 = (L_2' L_2^\dagger)^Q$ are orthogonal gates.  Letting $L_2$ and $L_2'$ range, we see from \Cref{ReductionToOrthogonalLemma} that we are maximizing a quadratic functional over the monodromy polytope slice associated to $D_1$ and $D_2$.  Mercifully, one need not employ this heavy machinery to solve this optimization problem:

\begin{lemma}[{\cite[Section III.A]{WVMCWRGK}}]
Suppose that $C_1$, $C_2$ are fixed canonical gates and that $L_1$, $L_1'$ are fixed local gates.  Letting $L_2$ and $L_2'$ range over all local gates, the value $F_{\mathrm{avg}}(L_1 C_1 L_1', L_2 C_2 L_2')$ is maximized when taking $L_2 = L_1$ and $L_2' = L_1'$. \qed
\end{lemma}

\begin{corollary}
The spectrum of the gate which gives the best approximation to a two-qubit unitary $U$ depends only on $\LogCoords(U)$. \qed
\end{corollary}

By combining these results to our descriptions of $\LogCoords(P^n_\S)$ for our preferred gate sets $\S$, we produce the following protocol for approximate compilation.  In the following, we take $\S$ to be a gate-set with the nesting property of \Cref{NestingPolytopes} and $U$ to be a two-qubit program to be compiled.

\begin{enumerate}
    \item Calculate the canonical decomposition associated to $U$: $U = L C L'$.
    \item Let $n = 1$.
    \item \label{LoopStep} Use $\LogCoords(U)$ to calculate the point $\delta_*^n \in \LogCoords(P^n_\S)$ which maximizes $F_{\mathrm{avg}}(U, -)$.  Multiply this maximum value by $f_S^n$.\footnote{We are using $f_S^n$ as an \emph{approximation} for the fidelity of the depth $n$ circuit. However, as fidelity is not multiplicative, there is considerable room for the implementer to express their own preference here.}
    \item Is this fidelity value smaller than the previous fidelity value?  If not, increment $n$ and try Step \ref{LoopStep} again.  Otherwise, proceed to Step \ref{BuildStep}.
    \item \label{BuildStep} Find a realization $R$ of $\delta_*^{n-1}$ with canonical decomposition \[R = L_{\mathrm{approx}} \cdot C_{\mathrm{approx}} \cdot L'_{\mathrm{approx}}.\]
    \item Return \[L L_{\mathrm{approx}}^\dagger \cdot R \cdot (L'_{\mathrm{approx}})^\dagger L'.\]
\end{enumerate}

The first half of the protocol depends only on the structure of the polytopes $\LogCoords(P^n_\S)$, from which we may conclude the following result:

\begin{corollary}
The two-qubit gate sets $\{\CZ\}$, $\{\ISWAP\}$, $\{\CPHASE\}$, and $\{\PSWAP\}$ all do an equally effective job of approximating an arbitrary two-qubit program by a circuit with a pair of two-qubit gates.
\end{corollary}
\begin{proof}
This is a direct consequence of coupling the above ideas to \Cref{P2PSWAPDescription}.
\end{proof}

\begin{remark}[{\cite[Appendix B]{CBSNG}}]
Finding the nearest point in $P^n_\S$ to an arbitrary outside point is numerically accessible, but it does not seem to admit a closed-form solution in general.  An exception is the case of $P^2_{\CZ}$, where the nearest canonical gate to $\CAN(\alpha, \beta, \delta)$ is $\CAN(\alpha, \beta, 0)$.
\end{remark}

\begin{example}
The $\SWAP$ gate is of particular interest, and so we provide an analysis of its approximants as an example of these methods.  The nearest point to $\SWAP$ within $\LogCoords(P^2_{\XY})$ is $(1/3, 1/3, 0, -2/3)$, with an average gate infidelity of $3/20$.  The nearest point to $\SWAP$ within $\LogCoords(P^2_{\DB})$ is only slightly further: $(1/4, 1/8, 1/8, -1/2)$, with average gate infidelity of $1/6$.  For contrast, the nearest point to $\SWAP \equiv \CAN(\pi/2, \pi/2, \pi/2)$ within $\LogCoords(P^2_{\CZ})$ is given by $\CAN(\pi/2, \pi/2, 0)$, with an average gate infidelity of $2/5$.
\end{example}

\section{Open questions}\label{OpenQuestionsSection}

In closing the main thread of the paper, we list some follow-on projects where one would expect to find interesting results.

\subsection{Algorithmic effectiveness and circuit realization}

The single most important avenue left open by this work is the actual manufacture of a circuit in $P^2_\S$ from a point in $\LogCoords(P^2_\S)$ (i.e., \Cref{2QMultiDeckerProblem}.\ref{2QMultiDeckerProblem2}), which we refer to as the \textit{realization problem}.

Edelman et al.\ have presented a specialization of Newton's method on a curved Riemannian manifold to the orthogonal group with its natural metric~\cite{EAS}.  If one were able to provide approximate solutions to the realization problem, such an algorithm could be used to rapidly increase the accuracy of such a solution---but without approximate solutions, such methods have no guarantee of convergence.\footnote{In the particular case of $P^2_{\XY}$, M.\ Scheer has pointed out to us the commutation relation $[\mathrm{XX} + \mathrm{YY}, \mathrm{ZI} + \mathrm{IZ}] = 0$, from which it follows that the group of interest can be reduced from $PO(4)$ to a particular four-dimensional subset.  However, this subset is not closed under multiplication, which hinders the translation of the methods of Edelman et al.}  Additionally, this method would require foreknowledge of the gates $D_E$ and $D_F$ in \Cref{OrthogonalMEP}, limiting its applicability in parametric settings such as $P^2_{\XY}$.

From the perspective of \Cref{AWBSection}, a solution to the monodromy problem corresponds to a flat connection on the trivial $PU(4)$--bundle over a punctured Riemann sphere with prescribed monodromy values.  The data of an \emph{arbitrary} connection is easier to describe: it assigns to each path an element of $PU(4)$ via parallel transport, perhaps with some further smoothness conditions.  Rade's thesis~\cite{Rade} analyzes the Yang--Mills flow from an arbitrary such connection (with boundary conditions) to a flat representative, and for generic connections its convergence is rapid.  One might therefore try to discretize the punctured Riemann sphere and apply a numerical variant of Rade's method.  It is not immediately clear, however, how one would introduce the orthogonality constraints present in \Cref{OrthogonalMEP}~\cite{FalbelWentworth}.

Cole Franks et al.~\cite{Franks,FranksEtAl} have described effective numerical methods for solving the \emph{additive} analogue of the eigenvalue problem.  One might explore multiplicative variations on their methods (especially those with the orthogonality constraint kept in mind) which would then adapt to solve the problem posed here.

A separate concern of algorithmic effectiveness is the taming of the exponential upper bound on the size of the family of inequalities coming from \Cref{PnAreAllPolytopes}.  In practice, this bound appears to be a gross overestimate, and we are optimistic that a polynomial bound is possible.  As an aside, one of the advantages of Belkale's analysis of the monodromy polytope~\cite{Belkale} over that of Agnihotri and Woodward~\cite{AgnihotriWoodward} is that his set of inequalities is minimal---so similar considerations have already been taken up by the progenitors of the results shown here.

\subsection{Alternative interpretations of ``optimum''}

The particular metric by which we measured the utility of $\DB$ over other instances of $\XY_\alpha$ was the volume of the polytope $\LogCoords(P^2_{\DB})$.  It is not clear that this is the best such metric (nor that there is a best).  Here we list some alternative metrics that seem worth exploring.

Firstly, is there a value of $\alpha$ for which the average (or worst) value of average gate infidelity is minimized?  Against this metric, an ``elliptical'' polytope may be more valuable than a ``spherical'' one.  This analysis may also change when considering other approximation metrics than average gate infidelity, e.g., diamond distance.

The Haar volume (or, indeed, most any other natural volume) of a subset of $PU(4)$ is not perfectly related to the volume of its image as a subset of $\A_{C_2}$.  For any such volume $\vol'$ on $PU(4)$, it would also be of interest to maximize the analogous function $\vol' P^2_{\XY_t}$ over $t$.  (For the Haar volume, it appears that the maximum remains at $\alpha = 3\pi/4$, but we do not have a proof that this is so.)  Another discrepancy that is worth illuminating is the change-of-coordinates formula comparing $\<\L_\S\>^\Haar$ and $\<\L_\S\>^{\A_{C_2}}$.  These integrals vary at least by the Jacobian of $\LogCoords$, considered as a function on the canonical subgroup, which would already be worth computing.

It would also be of interest to understand the local behavior of any of these metrics with respect to small distances in $PU(4)$.  This is the domain of \textit{coherent unitary error}, and one might hope to leverage some of the results of this paper to tailor a compilation method for a coherently error-prone device.  Preliminary inspection of this for $P^2_{\CZ}$ indicates that derivatives conspire so that only large coherent unitary error gives rise to significant gain in volume.

Any more nuanced tracking of error has the potential to alter the analysis in \Cref{ApproxCompSection}.  In particular, separate tracking of unitary and nonunitary error is very likely to affect the outcome of the protocol described there.

C.\ Iancu has observed that potential native gate-sets are typically not uniform in quality: for instance, on certain hardware architectures it could be the case that $\sqrt{\CZ}$ takes half as long (and hence suffers half as much performance degradation) as $\CZ$, so that the ``performance metric'' of expected gate depth ought to be reweighted by an appropriate factor coming from the performances of the gates involved.  Indeed, in this toy example, one finds \[\frac{1}{2} \<\L_{\sqrt{\CZ}}\>^{\A_{C_2}} = 1.84375 < 3 = \<\L_{\CZ}\>^{\A_{C_2}},\] contrary to the apparent loss in gateset quality considered in \Cref{SqrtCZExample}.

\subsection{Unexplored polytopes}

The material presented here amounts to a toolkit for analyzing the space of programs available to a given native gate set.  We have used this as incentive to investigate a particular native gate set because of its depth-two behavior and its relevance to a particular sort of hardware, but this is hardly the only option.

As part of the project of exploring the structure of the space of gatesets, it would be of interest to describe those native gate sets $\S$ which enjoy $\LogCoords(P^2_\S) = \A_{C_2}$.  This set is nonempty, as the $\B$--gate has this property.  Are there other singletons?  Other finite sets?  Other exponential families?

\Cref{HeatMapFigure} could probably be smoothed using the same interpolation techniques as in \Cref{MaximizingXYThm}, yieliding an explicit piecewise formula for $\vol^{\A_{C_2}}(P^2_S)$ as $S$ ranges over the entire solid.

The hardware employed by Rigetti is not the only option, and one could re-run this same analysis for other designs.  As an example, calculating the subspace of programs efficiently available to Google's $\mathrm{fSim}$ gate would likely be a pleasant exercise in these techniques.

The methods of this paper may also bear indirect fruit at higher qubit counts.  For example, Wei and Di~\cite{WeiDi} have given a $\CZ$--based circuit decomposition for operators in the subgroup $PO(8)$ which improves over the best known general such circuit decomposition for operators in $PU(8)$.  They leave open the full reach of their result: namely, the local equivalence class of $PO(8)$ within $PU(8)$ belongs in the middle of the chain of inclusions \[PO(8) \subseteq PU(2)^{\otimes 3} \cdot PO(8) \cdot PU(2)^{\otimes 3} \subseteq PU(8),\]  and their techniques continue to apply to this middle term.  However, its structure is unknown, save that $PO(8)$ is generically of codimension $12$ within it.  Preliminary application of our techniques in this context shows that this local equivalence class is detected within $PU(8)$ by a set of linear constraints on the logarithmic spectrum, as well as some further constraint phenomena for which we are unable to account.  Giving concrete descriptions these remaining constraints would greatly widen the applicability of the techniques of Wei and Di.

Recent work by Glaudell, Ross, and Taylor leverages the accidental isomorphism $SU(4) \cong \mathrm{Spin}(6)$ to produce circuit decompositions into a certain discrete gate set~\cite{GRT}.  One might wonder whether the methods described here interact usefully with this alternative presentation of two-qubit operations.

\subsection{Leakiness}

The analysis of ``leaky gates'' in \Cref{LeakyEntanglersApdx} is not as thorough as it might be.  Here are some open questions and problems concerning that property:

\begin{itemize}
\item In \Cref{RearrangingForZ}, we argue that within the local equivalence class of a leaky entangler, there is one where the single-qubit gates involved in the leakiness relation are all $\Z$--gates.  However, their parameters may depend on each other in a nontrivial way.  Give a description of the possible ways this can happen.  The exponential family $\SWAP_\alpha$ (i.e., the $(\alpha/\pi)$\textsuperscript{th} root of $\SWAP$) is probably of interest here.
\item Every given example of a leaky gate is leaky on both coordinates.  Is this always the case?
\item Every given example of a leaky gate has a representative in its local equivalence class with it transpose-symmetric.  Is this always the case?
\item Our best guess is that the subspace of leaky entanglers coincides exactly with the edges of $\A_{C_2}$.  Is this true?
\end{itemize}

\section*{Acknowledgements}

We would like to thank Rigetti Computing for providing such a stimulating workplace, with difficult problems to solve and wonderful peers to work alongside.  In particular, M.\ Appleby, J.\ Combes, E.\ Davis, C.\ Hadfield, P.\ Karalekas, A.\ Papageorge, N.\ Rubin, C.\ Ryan, M.\ Scheer, M.\ P.\ da Silva, M.\ Skilbeck, and N.\ Tezak contributed a lot to our momentum, whether through pointers, proofreading, or generalized enthusiasm.  E.\ Davis deserves special credit, as he finally put us on the right path to move from numerical experiment to mathematical proof, and this project might not have come together if not for his crucial advice.  Although W.\ Zeng had been suggesting that we think about approximate compilation for some time, \Cref{ApproxCompSection} is a direct result of A.\ Javadi-Abhari's very pleasant talk at IWQC 2018.  S.\ Lin helpfully brought to our attention a cascade of mistakes in a previous version of this paper.  We additionally had the pleasure of speaking with the mathematicians W.\ C.\ Franks, B.\ Gammage, S.\ Kumar, E.\ Lerman, S.\ Lichak, P.\ Solis, C.\ Teleman, R.\ Wentworth, and C.\ Woodward, who freely offered their consultation and expertise on matters related to the monodromy polytope.  The first author would like to note the indirect but invaluable role that his PhD adviser, C.\ Teleman, played in this project: the material presented here is \emph{much} closer to Teleman's domain than the first author's thesis ever was, and acquiring the working knowledge to complete this project would have been much harder without a steady exposure to these ideas over the years.  An additional hearty thank you goes to R.\ Bryant for teaching the first author most of what he knows about Lie theory (and from the second author to the first for the same reason).  Finally, the anonymous referees enormously improved the quality of this paper and kindly pointed out a goodly number of errors, for which we are very grateful.

\bibliographystyle{habbrv}
\bibliography{coverage}

\onecolumn\newpage
\appendix

\section{The mathematics of the monodromy polytope}\label{AWBSection}

In this appendix, we produce some of the details (or, failing that, some soothing exposition) of the mathematics underlying our results in the main text.  This effort cleaves into two parts: some generic convexity results in symplectic geometry that give the qualitative solution to the multiplicative eigenvalue problem (and which merit the name ``monodromy polytope''), followed by some results around quantum cohomology that give the quantitative solution.

\subsection{Qualitative results}

Before getting involved with the multiplicative eigenvalue problem directly, we first give a slightly ahistorical account\footnote{The solution to this problem is strongly coupled to the solution of the corresponding ``linearized'' problem: given Hermitian matrices $H_1$ and $H_2$, what spectra can possibly arise as that of $\Ad_{U_1} H_1 + \Ad_{U_2} H_2$ for unitary operators $U_1$ and $U_2$?  A conjectural solution to this problem was set out by Horn~\cite{Horn}, which spurred the development of a great many results in symplectic geometry and representation theory in an effort to explain his findings, and these tools were ultimately used by Klyachko~\cite{Klyachko} to settle the matter.  Knutson~\cite{Knutson} gives a very pleasant overview of this body of work and its surroundings, and although he does not address the multiplicative problem, (generalizations of) these same tools reappear in this context.  We intend the word ``ahistorical'' only in the sense that the tools were developed \emph{in response to} the visible behavior of the (additive) eigenvalue problem, whereas our exposition presents the tools as generic ideas which we then apply \emph{post facto} to the eigenvalue problem---a significant misrepresentation of history.} of a generic qualitative result found in symplectic geometry.

\begin{definition}
A \textit{symplectic manifold} $M$ is an oriented $2n$--manifold equipped with a choice of \textit{symplectic form} $\omega$, i.e., an $n$\textsuperscript{th} root of the volume form (or, equivalently, an everywhere nondegenerate $2$--form).
\end{definition}

\begin{example}
Examples of such objects are rife in physics: all phase spaces are instances of symplectic manifolds.  For an ultra-simple but ultra-concrete example, we might take $M = T^* \R = \R^2$ with the symplectic form $\omega = dp \wedge dq$, or more generally $M = T^* \R^d = \R^{2d}$ with $\omega = \sum_j dp_j \wedge dq_j$.  These arise as the phase spaces associated to $d$ many non-interacting simple harmonic oscillators.  In general, a symplectic manifold has this as its local form.
\end{example}

\begin{definition}
Given an action on a symplectic manifold $M$ by a Lie group $G$ (so that $\tau_g^* \omega = \omega$), a \textit{moment map} is a $G$--equivariant function $\Phi\co M \to \mathfrak g^*$, where the target carries the coadjoint action.
\end{definition}

\begin{example}\label{GaugeGroupsInPhysicsEx}
Again, examples of such objects are rife in physics: a nontrivial gauge group gives rise to a $G$--action on a phase space, and a moment map can be used to describe a $G$--invariant physical quantity, such as the total energy of a system.  In the above example, $G = S^1$ acts on $\R^2$ by rotation, and the associated Lie algebra $\mathfrak{g}^*$ can be identified with $\R$ in such a way that a moment map is given by $\Phi(v) = \frac{1}{2} |v|^2$.  Similarly, the $d$--torus $G = (S^1)^{\times d}$ acts on $\R^{2d}$ by rotations of the component planes, and there is an associated moment map $\R^{2d} \to \mathfrak{g}^* \cong \R^d$ which sends each particle to its total energy.
\end{example}

A useful tool for manufacturing these objects comes in the form of the following theorem:

\begin{theorem}[Symplectic reduction]
Let $M$ be a symplectic $G$--manifold with associated proper moment map $\Phi_G$, and let $H \le G$ be a normal subgroup.  When $M \mmod H := \Phi_G^{-1}(0) / H$ is a manifold, it inherits both a symplectic form $\omega_{M \mmod H}$ (which pulls back to $\Phi_G^{-1}(0)$ to agree with the restriction of $\omega_M$), a compatible action by $G/H$, and a moment map $\Phi_{G/H}$. \qed
\end{theorem}

\noindent Our interest in these objects stems from the following family of convexity results:

\begin{theorem}\label{SymplecticConvexity}
Let $M$ be a connected symplectic manifold with an action by a Lie group $G$ through symplectomorphisms and a proper moment map $\Phi\co M \to \mathfrak{g}^*$.
\begin{itemize}
    \item (Atiyah~\cite{Atiyah}, Guillemin--Sternberg~\cite{GS1,GS2}:) Suppose that $G = T$ is a compact torus, and let $\A$ be a choice of fundamental alcove within $\mathfrak t^*$.  The restriction of the image of $\Phi$ to $\A$ then forms a convex polytope.
    \item (Kirwan~\cite{Kirwan}:) Suppose that $G$ is compact, let $T \le G$ be a choice of maximal torus with corresponding dual Cartan subalgebra $\pi\co \mathfrak g^* \to \mathfrak t^*$, and let $\A$ again be a fundamental alcove within $\mathfrak t^*$.  The restricted set $\A \cap \operatorname{im}(\pi \circ \Phi)$ is a convex polytope.
    \item (Meinrenken--Woodward~\cite[Theorem 3.13]{MeinrenkenWoodwardUnpublished}:) For $G'$ a Lie group, its \textit{loop group} is the infinite-dimensional Lie group $LG' = (G')^{S^1}$ of loops in $G'$ with pointwise multiplication.\footnote{The loops are commonly assumed to have a further technical property---for instance, some degree of differentiability.}  Suppose that $G = LG'$ for $G'$ a compact, connected, simply connected Lie group, let $T'$ be a choice of maximal torus within $G'$, and let $\A'$ be a choice of fundamental alcove within $(\mathfrak t')^*$.  The intersection $\A' \cap (\pi \circ \Phi)(M)$ is then a convex polytope.
    \pushQED\qed\qedhere\popQED
\end{itemize}
\end{theorem}

\begin{remark}
The most basic of this chain of results is somewhat believable: in \Cref{GaugeGroupsInPhysicsEx}, the image of the moment map is the positive orthant in $\mathfrak t^*$.  Since a general symplectic manifold is constructed locally from that example, the image of a general moment map is constructed locally out of such ``corners''---though amplifying this to an equivariant statement (and then to the nonabelian setting) is no trivial feat.  The final form of the theorem is \emph{considerably} harder to visualize, but it is the version that will concern us chiefly.
\end{remark}

\begin{example}[{\cite[p.\ 587]{AtiyahBott}, \cite{Donaldson}}]\label{AtiyahBottModuli}
We focus our attention on an example that physicists may recognize as an instance of Yang--Mills theory.  Let $G$ be a compact, connected, simply-connected Lie group (e.g., $SU(4)$), let $\Sigma$ be Riemann sphere with $b$ disks excised, and let $P$ be the trivial principal $G$--bundle over $\Sigma$.  The space $\mathcal A(\Sigma; \mathfrak g)$ of $\mathfrak g$--valued connections on $P$ may be identified with $\Omega^1(\Sigma; \mathfrak g)$, and it can be shown to carry the structure of a symplectic manifold using the Atiyah--Bott symplectic form \[\omega_{AB}(A_1, A_2) = \int_\Sigma \tr (A_1 \wedge A_2) .\]  This carries a compatible action by the gauge group $\mathcal G(\Sigma)$ of sections of $P$ (i.e., $G$--valued continuous functions on $\Sigma$), which has Lie algebra $\Omega^0(\Sigma; \mathfrak g)$, and this action moreover admits a moment map $\Phi_{AB}$ determined by \[\<\Phi_{AB}(A), \xi\> = \int_\Sigma F_A \cdot \xi + \int_{\partial \Sigma} \iota^*(A \cdot \xi),\] where $F_A$ is the curvature form associated to $A$ and $\iota\co \partial \Sigma \to \Sigma$ is the inclusion of the boundary components.  Writing $\mathcal G_\partial(\Sigma)$ for the term in the kernel sequence \[1 \to \mathcal G_\partial(\Sigma) \to \mathcal G(\Sigma) \xrightarrow{\iota^*} \mathcal G(\partial \Sigma) \to 1,\] the restricted action on $\mathcal A(\Sigma; \mathfrak g)$ inherits the moment map $\Phi_\partial(A) = F_A$, and so the symplectic reduction \[\M^\flat(\Sigma; \mathfrak g) = \mathcal A(\Sigma; \mathfrak g) \mmod \mathcal G_\partial(\Sigma) = \mathcal A^\flat(\Sigma; \mathfrak g) / \mathcal G_\partial(\Sigma),\] called the \textit{moduli of flat connections}, inherits an action by $\mathcal G(\partial \Sigma) \cong LG^b$ and a $\mathcal G(\partial \Sigma)$--equivariant moment map $\Phi^\flat(A) = \iota^* A$.
\end{example}

\begin{corollary}[{\cite[Theorem 3.2]{MeinrenkenWoodwardVerlindeFactorization}, \cite[Theorem 3.16]{MeinrenkenWoodwardUnpublished}}]\label{MEPIsAConvexPolytope}
The set \[\LogSpec \{ U_1, U_2, U_3 \in SU(4) \mid U_1 U_2 = U_3\} \subset \A^{\times 3}\] is a convex polytope.
\end{corollary}
\begin{proof}[Construction]
We set $\Sigma = \CP^1 \setminus \{1, 2, 3\}$ and $G = SU(4)$, then apply \Cref{AtiyahBottModuli} to conclude that $\M^\flat(\Sigma)$ is a symplectic $G$--manifold with associated moment map $A \mapsto \iota^* A$.  Fix the following auxiliary data:
\begin{itemize}
    \item Parametrizations $B_j\co S^1 \to \Sigma$ of the $j${\th} boundary component.
    \item Paths $\gamma_j\co B_1(0) \to B_j(0)$ begetting loops $b_j = \gamma_j^{-1} B_j \gamma_j$ which have the property \[\pi_1 \Sigma = \left\langle \begin{array}{c}b_1 = B_1, \\ b_2 = \gamma_2^{-1} B_2 \gamma_2, \\ b_3 = \gamma_3^{-2} B_3 \gamma_3 \end{array} \middle| 1 = b_1 b_2 b_3 \right\rangle.\]
\end{itemize}
To these data, a connection $A$ associates elements $B_j^* A \in L \mathfrak g^*$, the local behavior of $A$ near $B_j$, and elements $\Gamma(A, \gamma_j)_0^1 \in G$, the action of parallel transport along $\gamma_j$ from the (trivialized) fiber over $\gamma_j(0)$ to the (trivialized) fiber over $\gamma_j(1)$.

It is well-known that the moduli space of flat connections on a trivial $G$--bundle over a suitable Riemann surface $\Sigma$ is weakly equivalent to the space of $G$--representations of $\pi_1 \Sigma$.  The procedure for extracting such a representation is by sending a loop in the base to the monodromy of the connection around the loop.  One may promote this idea from a weak equivalence into a commuting square with horizontal arrows \emph{equivariant symplectomorphisms}:
\begin{center}
    \begin{tikzcd}[row sep=2em]
    \M^\flat(\Sigma; \mathfrak g) \arrow{r} \arrow["\Phi"]{d} & \displaystyle \left\{ \begin{array}{c} c_* \in \{1\} \times G^2 \\ \xi_* \in (L\mathfrak g^*)^{\times 3} \end{array} \middle| 1 = \prod_{j=1}^3 \Ad_{c_j} \Mon(\xi_j) \right\} \arrow["\Phi"]{d} \\
    (\operatorname{Lie} \mathcal{G}(\partial \Sigma))^* \arrow{r} & \left\{ (\xi_*) \in (L\mathfrak g^*)^{\times 3} \right\},
    \end{tikzcd}
\end{center}
where the first horizontal arrow is defined by
\begin{align*}
    c_j(A) & = \Gamma(A, \gamma_j)_0^1, &
    \xi_j(A) & = B_j^*(A),
\end{align*}
the monodromy operator is defined by \[\Mon(\xi_j) = \int_{S^1} B_j^*(A) \in G,\]and the action of $\mathcal G(\partial \Sigma) \cong LG^3$ on the top-right corner is given by
\begin{align*}
    g \cdot c_j & = g_j(0)^{-1} c_j g_1(0), &
    g \cdot \xi_j & = \Ad_{g_j} \xi_j - g_j^{-1} dg_j.
\end{align*}

Granting this, we find ourselves at the doorstep of the multiplicative eigenvalue problem.  Note first that the operator $\Mon$ enjoys two pleasant properties:
\begin{enumerate}
    \item After using the Killing form to identify $\mathfrak g^*$ with the subspace $\mathfrak g \subseteq L \mathfrak g$ of constant loops, for $h \in \mathfrak g$ we have $\Mon(h) = \exp(h)$, the usual Lie exponential.
    \item The $G$--action on $\xi_j$ is then arranged so that the following formula holds: \[\Mon(g \cdot \xi_j) = \Ad_{g_j(0)} \Mon(\xi_j).\]
\end{enumerate}
These properties combine to give the required link.  We apply \Cref{SymplecticConvexity}: take $\A \subset \mathfrak t^*$ to be real diagonal matrices whose entries obey the criteria set out by \Cref{LogSpecDefn}.  The image of the moment map then becomes those triples of diagonal matrices $(\xi_1, \xi_2, \xi_3) \in \A^{\times 3}$ for which there exist unitary operators $c_2$, $c_3$ satisfying \[e^{- 2 \pi i \xi_1} = c_2^{-1} e^{2 \pi i \xi_2} c_2 \cdot c_3^{-1} e^{2 \pi i \xi_3} c_3. \qedhere\]
\end{proof}

\begin{remark}\label{CentralExtensionsRemark}
Throughout the paper, there are two Lie groups of interest: $PU(4) = U(4) / \C^\times$, which participates in a nontrivial central extension \[1 \to C_4 \to SU(4) \to PU(4) \to 1,\] and the double cover $SU(4) / C_2$ of $PU(4)$, which also participates in a nontrivial central extension \[1 \to C_2 \to SU(4) \to SU(4) / C_2 \to 1.\]  Neither is simply connected, a necessary hypothesis of \Cref{MEPIsAConvexPolytope}, which we redress as follows.  In general, we may consider compact connected Lie groups $G$ whose universal cover $\widetilde G$ participates in a finite central extension \[1 \to F \to \widetilde G \xrightarrow\pi G \to 1.\]  The Lie algebras of $\widetilde G$ and $G$ may be identified by $\pi$, and the image of the moment map $\Phi_G$ considered in \Cref{MEPIsAConvexPolytope} is then given by the union over $f \in F$ of the images of the moment maps $\Phi_{\widetilde G, f}$, constructed analogously so as to detect products of the form $U_1 U_2 = f U_3$ with $U_1, U_2, U_3 \in \widetilde G$.
\end{remark}

\subsection{Quantitative results}

We now turn to quantitative results: given that the solution set to the multiplicative eigenvalue problem forms a convex polytope, what polytope is it?  As in the additive case, this problem passes through representation theory, and in the exposition about the qualitative problem we have already begun to make this contact: a flat connection on a trivial vector bundle is equivalent data to a representation of the fundamental group of the base, and flat connections modulo gauge equivalence correspond to representations up to choice of basis.  Theorems of Narasimhan and Seshadri and of Donaldson show that this has a kind of converse: unitary representations of the fundamental group of a compact Riemann surface correspond to ``stable'' holomorphic vector bundles over the surface~\cite{NarasimhanSeshadri}, and such bundles can be shown to admit a unique flat unitary connection~\cite{DonaldsonNS}.  A vector bundle $V$ is said to be stable when its slope, $\mu(V) = \operatorname{deg}(V) / \operatorname{rank}(V)$, decreases when passing to any subbundle.  Informally, a stable bundle is ``more ample'' than any of its subbundles.

However, our surface of interest, $\Sigma = \CP^1 \setminus \{1, 2, 3\}$, is a \emph{noncompact} Riemann surface.\footnote{In fact, the fundamental groups of compact Riemann surfaces are all known: the surface $\Sigma_g$ of genus $g$ has fundamental group the free group on letters $a_1, b_1, \ldots, a_g, b_g$ subject to the relation $1 = [a_1, b_1] \cdots [a_g, b_g]$.  There is no $g$ for which this looks like our desired free group on generators $a$, $b$, $c$ subject to $a b c = 1$.}  Work of Mehta and Seshadri extends the above correspondence to the noncompact case: a \textit{parabolic bundle} (on $\CP^1$) is a holomorphic vector bundle $E$, a choice of finite set $S \subset \CP^1$, a choice of flag $\{E_{s,i}\}$ for each $s \in S$, and a family of weights $\lambda_{s,i}$ satisfying the strings of inequalities \[\lambda_{s,1} \ge \cdots \ge \lambda_{s,n} > \lambda_{s,1} - 1\] as well as the equality $\deg E = -\sum_{s,i} \lambda_{s,i}$.  A parabolic bundle is additionally said to be \textit{semistable} when its \textit{parabolic slope}, a modification of the slope that is offset by the choice of parabolic weights, decreases when passing to any subbundle (and appropriately restricting the parabolic structure).  They then show the following result:

\begin{theorem}[{\cite{MehtaSeshadri}}]
Fix a set $S$ and a family of parabolic weights $\lambda_{s,i}$.\footnote{In fact, they assume that $\lambda_{s,i}$ are rational because they work with tools from algebraic geometry.  Since we are concerned with complex geometry, we may drop this assumption by interpolation.}  The moduli space of semistable parabolic bundles on $\CP^1$ with these weights is a normal, projective variety, homeomorphic to the moduli space of flat unitary connections on the trivial bundle over $\CP^1 \setminus S$ such that the monodromy operator $U_s$ at $s$ has $\LogSpec U_s = (\lambda_{s,i})_i$. \qed
\end{theorem}

What this theorem conspicuously does not assert is when the moduli of semistable parabolic bundles is \emph{nonempty}.  In order to assess this, Agnihotri and Woodward give a geometric interpretation of the semistability condition, then connect it to a complicated form of intersection theory known as \textit{quantum cohomology}.  Their ultimate theorem statement is as follows:

\begin{definition}
We make the following definitions:
\begin{itemize}
	\item For $r, k > 0$ be positive integers with $r + k = n$, let $\P_{r,k}$ be the set of partitions \[\P_{r,k} = \{I \subseteq \{1, 2, \ldots, n\} : |I| = r\}\] with duality operator $*I = \{n + 1 - i \mid i \in I\}$.
	\item Let $\Gr(r,k)$ be the Grassmannian of $k$--planes in $\C^n$.
	\item Let $\C^n = F_n \supset F_{n-1} \supset \cdots \supset F_0 = \{0\}$ be a complete flag in $\C^n$.
	\item For a partition $I \in \P_{r,k}$, its \textit{Schubert variety} is $\sigma_I = \left\{W \in \Gr(r,k) \mid \dim(W \cap F_{I_j}) \ge j\right\}$.
	\item The \textit{Schubert cell} $C_I \subset \sigma_I$ is the complement of all lower-dimensional Schubert varieties contained in $\sigma_I$: $C_I = \bigcap_{\sigma_J \subset \sigma_I} \sigma_I \setminus \sigma_J$.
	\item From these, we define Schubert cycles $[\sigma_I]$ and $[C_I]$ in $H_* \Gr(r,k)$, as well as dual cohomology classes $T_I \in H^* \Gr(r,k)$ and Poincar\'e dual homology classes $[\sigma_{*I}]$.
\end{itemize}
\end{definition}

\begin{theorem}[{\cite[Theorem 5.3, Lemma 5.5]{AgnihotriWoodward}}]\label{AWThmComplex}
The moduli of semistable parabolic bundles with prescribed weights $\lambda_{s,i}$ is non-empty if and only if, for all subsets $I_s$ and integers $d$ such that there exists a degree $d$ map sending $s \in S$ to a general translate of the Schubert cell $C_{I_s}$, \[\sum_{s \in S} \sum_{i \in I_s} \lambda_{s,i} \le d.\]

Moreover, the minimum such value $d$ is computable: for $d$ the lowest degree of any map \[\mu\co \CP^1 \to \Gr(r,k)\] sending $s \in S$ to a general translate of $\sigma_{I_s}$, $q^d$ is the maximal power of $q$ dividing $\prod_{s \in S} T_{I_s}$ in the ``small quantum cohomology ring'' of $\Gr(r, k)$. \qed
\end{theorem}

We now endeavor to explain the contents of this theorem and the connection to the ideas above.  The first half is relatively easy to see: suppose that we have a semistable parabolic structure on the trivial bundle over $\CP^1$.  A subbundle of rank $k$ is then classified by a map $\mu\co \CP^1 \to \Gr(r, k)$, and the inequality imposed by semistability on the parabolic weights is given by \[\sum_{s \in S} \sum_{i \in I_s} \lambda_{s, i} \le \deg(\mu),\] where here $I_s$ is the position of the subspace $\mu(s)$ inside of the parabolic flag at $s \in S$.  Taking the intersection of all such inequality families imposed by all such maps then gives the proposed description of the moduli.

The meat of the theorem is in the connection with quantum cohomology, which requires a much more elaborate explanation.  We follow a set of summary lectures by Fulton and Pandharipande~\cite{FultonPandharipande}.  Beginning with a sufficiently nice\footnote{The key property is called ``convexity''~\cite[Equation 2, Sections 1--6]{FultonPandharipande}.  Any homogeneous variety $X = G/P$ with $P$ a parabolic subgroup will do~\cite[pg.\ 6]{FultonPandharipande}, which includes Grassmannians.} space $X$ and for a choice of class $\beta \in H_2(X)$, one may construct a \textit{moduli space of nodal curves} $\mathcal M_{g, S}(X, \beta)$ populated by triples $(C, \Sigma, \mu)$ consisting of a projective connective nodal curve $C$ of genus $g$, a marking $\Sigma\co S \to C$ in the nonsingular locus, and a map $\mu\co C \to X$ such that $\mu_*[C] = \beta$ and such that $\mu$ admits finitely many automorphisms.  This moduli turns out to have a compactification $\overline{\mathcal M_{g, S}}(X, \beta)$, and hence its rational cohomology acquires Poincar\'e duality and an intersection form~\cite{KontsevichManin}.

Our interest in this construction stems from setting $X = \Gr(r, k)$ and $g = 0$, so that $\overline{\mathcal M_{0, S}}(\Gr(r, k), \beta)$ carries information about the available maps $\mu$ in \Cref{AWThmComplex}.  If we try to simultaneously prescribe the positions $I_s$ of $\mu(s)$ inside of the parabolic flag at $s \in S$, we will be further led to consider the classes \[[\operatorname{ev}_s^{-1}(\sigma_{I_s})] \in H_* \overline{\mathcal M_{0, S}}(\Gr(r, k), \beta)\] as well as their intersections.  To capture these intersections, we define the \textit{Gromov--Witten invariant} $I_\beta$ associated to an $S$--labled family of cohomology classes $(\gamma_s)_{s \in S} \in H^{2*}(X)$: \[I_\beta\left(\gamma_s\right)_{s \in S} := \int_{\overline{\mathcal M_{0,S}}(X, \beta)} \prod_{s \in S} \operatorname{ev}_s^*(\gamma_s).\]

\begin{example}[{\cite[Equation 44]{FultonPandharipande}}]
This definition contains the following special case: \[T_i T_j = \sum_e \left( \int_{\Gr(r,k)} T_i T_j T_e \right) T_{*e} = \sum_e I_0(T_i T_j T_e) T_{*e},\] i.e., $I_0(T_i T_j T_e)$ records the structure constants for the cup product.  These values are recognized in the literature as \textit{Littlewood--Richardson coefficients}, $N_{ij}^{*e}(r, k)$.
\end{example}

Motivated by this, we use nontrivial classes $\beta$ to define the following ``deformed product'':

\begin{theorem}[{\cite[Equation 66]{FultonPandharipande}}]\label{QLRCDefinition}
There is a commutative $\mathbb Z\llbracket q \rrbracket$--bilinear product on $\mathbb Z\llbracket q \rrbracket \otimes H^* \Gr(r, k)$ given by the formula
\begin{align*}
    T_i \ast T_j & = \sum_e \left( \sum_{\substack{\beta \in H^2(\Gr(r,k)) \\ \text{$\beta$ ``effective''}}} I_\beta(T_i T_j T_e) \cdot q^{\int_\beta T_1} \right) T_{*e} =: \sum_{d, e} N_{ij}^{*e,d}(r, k) \cdot q^d \cdot T_{*e}.
\end{align*}
The structure coefficients $N_{ij}^{*e,d}(r, k)$ are called \textit{quantum Littlewood--Richardson coefficients}. \qed
\end{theorem}

We now reconnect with the multiplicative eigenvalue problem: if the classes $[\operatorname{ev}_s^{-1}(\sigma_{I_s})]$ intersect the subspace of $\overline{\mathcal M_{0, S}}(\Gr(r, k), \beta)$ of curves of degree $d$ to produce a nontrivial homology class, they will induce a corresponding $q^i$--divisible cohomology class to appear in the quantum cohomology product of the classes $T_{I_s}$ for some $i \le d$.  In the case where $d$ is minimal, they show that the corresponding cohomology class is degree $0$, so that the $q^d$--divisibility is exact, and also conversely that a minimally $q$--divisible cohomology class belongs to such an intersection~\cite[Lemma 5.5]{AgnihotriWoodward}.

Finally, we may actually check this condition in cases of interest because the quantum Littlewood--Richardson coefficients are computable: they are connected to enumerative geometry~\cite[Section 9]{FultonPandharipande}, and one can use this to calculate them directly in small-index cases; they are connected to cohomology and so obey associativity-type relations~\cite[Theorem 4]{FultonPandharipande}; and it is possible to assemble both of these sources of information into an algorithm which recursively computes them~\cite{BCFF,Buch}.  Since we are specifically interested in the cases of $SU(2)$ and $SU(4)$, we produce a table of the quantum Littlewood--Richardson coefficients appearing in the products on the small quantum cohomology rings for $\Gr(1, 1)$ in \Cref{QLWCoeffsForC2} and for $\Gr(1, 3)$, $\Gr(2, 2)$, and $\Gr(3, 1)$ in \Cref{QuantumLittlewoodRichardsonCoeffs}.

Altogether, these results assemble into the following summary theorem.  The form presented here is not entirely standard in the literature, but fits well with Buch's software implementation of the algorithm for computing quantum Littlewood--Richardson coefficients~\cite{Buch}, as tabulated in \Cref{QLWCoeffsForC2} and \Cref{QuantumLittlewoodRichardsonCoeffs}.

\begin{theorem}[{\cite[Theorem 3.1]{AgnihotriWoodward}, \cite[Theorem 7]{Belkale}}]\label{AWBTheorem}\label{AWBTheoremRephrased}
For $r, k > 0$ be positive integers with $r + k = n$, define $\Q_{r,k}$ be the set of partitions \[\Q_{r,k} = \{(I_1, \ldots, I_r) \in \mathbb Z^r \mid k \ge I_1 \ge \cdots \ge I_r \ge 0\}.\]
Let $U_1$, $U_2$, $U_3 \in SU(n)$ satisfy $U_1 U_2  = U_3$, and let $\alpha_*, \beta_*, \delta_* \in \A$ be the fundamental alcove sequence respectively associated to these unitaries through $\LogSpec$.  Select $r, k > 0$ satisfying $r+k = n$, select $a, b, c \in \Q_{r,k}$, and take $d \ge 0$; then if $N_{ab}^{c,d}(r,k) = 1$, the following inequality must hold:
\[
d - \sum_{i=1}^r \alpha_{k+i-a_i} - \sum_{i=1}^r \beta_{k+i-b_i} + \sum_{i=1}^r \delta_{k+i-c_i} \ge 0.
\tag{*}\label{EqnAW}
\]
Conversely, given alcove sequences $\alpha_*, \beta_*, \delta_* \in \A$ for which $N_{ab}^{c,d}(r,k) = 1$ implies \Cref{EqnAW}, there exist $U_1, U_2, U_3 \in SU(n)$ with $U_1 U_2 = U_3$ and
\begin{align*}
    \alpha_* & = \LogSpec U_1, &
    \beta_* & = \LogSpec U_2, & 
    \delta_* & = \LogSpec U_3.
\end{align*}
\end{theorem}
\begin{proof}
We deduce this from Belkale's theorem, exactly as it is stated in the literature~\cite[Theorem 7]{Belkale}:%
\footnote{The original Mehta--Seshadri theorem concerns unitary flat connections.  Belkale also produced an alternative form of their theorem appropriate for flat connections which are special unitary~\cite[Appendix]{Belkale}.}

\begin{quote}
Let $I_1, I_2, I_3 \in \P_{r, k}$ be partitions such that the Gromov--Witten invariant associated to the product of Schubert cycles satisfies $N_{I_1 I_2}^{*I_3, d}(r, k) = 1$.
For special unitary operators $A_1$, $A_2$, $A_3$ satisfying $A_1 A_2 A_3 = 1$, the following inequality then holds: \[\sum_{j=1}^3 \sum_{i \in I_j} \LogSpec(A_j)_i \le d.\]
Conversely, if $\alpha_*, \beta_*, \delta_* \in \A$ satisfy all of the inequalities given above, there must exist special unitary operators $A_1, A_2, A_3 \in SU(n)$ with $A_1 A_2 A_3 = 1$ and
\begin{align*}
\LogSpec A_1 & = \alpha_*, &
\LogSpec A_2 & = \beta_*, &
\LogSpec A_3 & = \delta_*.
\end{align*}
\end{quote}

We show that these inequalities are equivalent, up to notation.
First, note that the sets $\P_{r, k}$ and $\Q_{r, k}$ are related by a bijection $\pi$, which is constructed as follows.
Beginning with a subset $I \in \P_{r, k}$ and using the reference subset $R = \{k+1, k+2, \ldots, k+r\}$ of size $r$, consider the distances between their $j${\th} elements: the distance $R_1 - I_1$ is bounded by $k$; the distance $R_2 - I_2$ is bounded by $R_1 - I_1$; and so on.
Altogether, the distances uniquely determine an element of $\Q_{r, k}$.
Conversely, given a sequence of displacements $D \in \Q_{r, k}$, we produce the subset \[\pi\co D \mapsto \{k + j - D_j \mid 1 \le j \le r\} \in \P_{r, k}.\]
Suppose, then, that $I_1, I_2, *I_3 \in \P_{r, k}$ are preimages for $a, b, c \in \Q_{r, k}$ through the bijection $\pi$.
We may use this to rewrite the sums in the inequality in our theorem statement: \[d - \sum_{i=1}^r \alpha_{k+i-a_i} - \sum_{i=1}^r \beta_{k+i-b_i} + \sum_{i=1}^r \delta_{k+i-c_i} = d - \sum_{i \in I_1} \alpha_i - \sum_{i \in I_2} \beta_i + \sum_{i \in *I_3} \delta_i.\]

Next, note the identity for the logarithmic spectrum of an inverse: \[\LogSpec(U^{-1})_i = -\LogSpec(U)_{n+1-i},\] which extends over a sum to \[\sum_{i \in I} \LogSpec(U^{-1})_i = - \sum_{i \in *I} \LogSpec(U)_i.\]
Set $A_1 = U_1$, $A_2 = U_2$, and $A_3 = U_3^{-1}$, so that $A_1 A_2 A_3 = 1$ is equivalent to $U_1 U_2 = U_3$.
Making this substitution in the inequality and moving the sums to the other side, we recover Belkale's formula \[\sum_{i \in I_1} (\LogSpec A_1)_i + \sum_{i \in I_2} (\LogSpec A_2)_i + \sum_{i \in I_3} (\LogSpec A_3)_i \le d. \qedhere\]
\end{proof}

\begin{figure}[h]
\hfill
\begin{tabular}[t]{ccccccc}
$r$ & $k$ & $a$ & $b$ & $c$ & $d$ & $N_{ab}^{c,d}(r,k)$ \\
\hline
$1$ & $3$ & $(0)$ & $(0)$ & $(0)$ & $0$ & $1$ \\
              & & & $(1)$ & $(1)$ & $0$ & $1$ \\
              & & & $(2)$ & $(2)$ & $0$ & $1$ \\
              & & & $(3)$ & $(3)$ & $0$ & $1$ \\
        & & $(1)$ & $(1)$ & $(2)$ & $0$ & $1$ \\
              & & & $(2)$ & $(3)$ & $0$ & $1$ \\
              & & & $(3)$ & $(0)$ & $1$ & $1$ \\
        & & $(2)$ & $(2)$ & $(0)$ & $1$ & $1$ \\
              & & & $(3)$ & $(1)$ & $1$ & $1$ \\
        & & $(3)$ & $(3)$ & $(2)$ & $1$ & $1$ \\
$3$ & $1$ & $(0,0,0)$ & $(0,0,0)$ & $(0,0,0)$ & $0$ & $1$ \\
                  & & & $(1,0,0)$ & $(1,0,0)$ & $0$ & $1$ \\
                  & & & $(1,1,0)$ & $(1,1,0)$ & $0$ & $1$ \\
                  & & & $(1,1,1)$ & $(1,1,1)$ & $0$ & $1$ \\
        & & $(1,0,0)$ & $(1,0,0)$ & $(1,1,0)$ & $0$ & $1$ \\
                  & & & $(1,1,0)$ & $(1,1,1)$ & $0$ & $1$ \\
                  & & & $(1,1,1)$ & $(0,0,0)$ & $1$ & $1$ \\
        & & $(1,1,0)$ & $(1,1,0)$ & $(0,0,0)$ & $1$ & $1$ \\
                  & & & $(1,1,1)$ & $(1,0,0)$ & $1$ & $1$ \\
        & & $(1,1,1)$ & $(1,1,1)$ & $(1,1,0)$ & $1$ & $1$
\end{tabular}
\hfill
\begin{tabular}[t]{ccccccc}
$r$ & $k$ & $a$ & $b$ & $c$ & $d$ & $N_{ab}^{c,d}(r,k)$ \\
\hline
$2$ & $2$ & $(0,0)$ & $(0,0)$ & $(0,0)$ & $0$ & $1$ \\
                & & & $(1,0)$ & $(1,0)$ & $0$ & $1$ \\
                & & & $(1,1)$ & $(1,1)$ & $0$ & $1$ \\
                & & & $(2,0)$ & $(2,0)$ & $0$ & $1$ \\
                & & & $(2,1)$ & $(2,1)$ & $0$ & $1$ \\
                & & & $(2,2)$ & $(2,2)$ & $0$ & $1$ \\
        & & $(1,0)$ & $(1,0)$ & $(1,1)$ & $0$ & $1$ \\
                        & & & & $(2,0)$ & $0$ & $1$ \\
                & & & $(1,1)$ & $(2,1)$ & $0$ & $1$ \\
                & & & $(2,0)$ & $(2,1)$ & $0$ & $1$ \\
                & & & $(2,1)$ & $(0,0)$ & $1$ & $1$ \\
                        & & & & $(2,2)$ & $0$ & $1$ \\
                & & & $(2,2)$ & $(1,0)$ & $1$ & $1$ \\
        & & $(1,1)$ & $(1,1)$ & $(2,2)$ & $0$ & $1$ \\
                & & & $(2,0)$ & $(0,0)$ & $1$ & $1$ \\
                & & & $(2,1)$ & $(1,0)$ & $1$ & $1$ \\
                & & & $(2,2)$ & $(2,0)$ & $1$ & $1$ \\
        & & $(2,0)$ & $(2,0)$ & $(2,2)$ & $0$ & $1$ \\
                & & & $(2,1)$ & $(1,0)$ & $1$ & $1$ \\
                & & & $(2,2)$ & $(1,1)$ & $1$ & $1$ \\
        & & $(2,1)$ & $(2,1)$ & $(1,1)$ & $1$ & $1$ \\
                        & & & & $(2,0)$ & $1$ & $1$ \\
                & & & $(2,2)$ & $(2,1)$ & $1$ & $1$ \\
        & & $(2,2)$ & $(2,2)$ & $(0,0)$ & $2$ & $1$
\end{tabular}
\hfill\null
    \caption{Structure constants in $qH^* \mathrm{Gr}(r, k)$ for $r + k = 4$. Note $N_{ab}^{c,d}(r,k) = N_{ba}^{c,d}(r,k).$}\label{QuantumLittlewoodRichardsonCoeffs}
\end{figure}


\section{Leaky entanglers}\label{LeakyEntanglersApdx}

There is also a differential-geometric proof that $\LogCoords(P^2_{\CZ})$ has vanishing volume which does not rely on first knowing the precise region.  The gate $\CZ$ commutes with $\Z$--rotations:
\begin{align*}
    \begin{tikzcd}[ampersand replacement=\&]
    \& \ctrl{1} \& \gate{\Z_\alpha} \& \qw \\
    \& \control{} \& \qw \& \qw
    \end{tikzcd} & = 
    \begin{tikzcd}[ampersand replacement=\&]
    \& \gate{\Z_\alpha} \& \ctrl{1} \& \qw \\
    \& \qw \& \control{} \& \qw
    \end{tikzcd}, &
    \begin{tikzcd}[ampersand replacement=\&]
    \& \ctrl{1} \& \qw \& \qw \\
    \& \control{} \& \gate{\Z_\alpha} \& \qw
    \end{tikzcd} & = 
    \begin{tikzcd}[ampersand replacement=\&]
    \& \qw \& \ctrl{1} \& \qw \\
    \& \gate{\Z_\alpha} \& \control{} \& \qw
    \end{tikzcd},
\end{align*}
from which we may conclude the following for a generic pair of one-qubit gates $K_1$ and $K_2$:
\begin{align*}
    \begin{tikzcd}[ampersand replacement=\&]
    \& \ctrl{1} \& \gate{K_1} \& \ctrl{1} \& \qw \\
    \& \control{} \& \gate{K_2} \& \control{} \& \qw
    \end{tikzcd} &
    = \begin{tikzcd}[ampersand replacement=\&]
    \& \ctrl{1} \& \gate{\Z_\alpha} \& \gate{\Y_\beta} \& \gate{\Z_\delta} \& \ctrl{1} \& \qw \\
    \& \control{} \& \gate{\Z_\epsilon} \& \gate{\Y_\zeta} \& \gate{\Z_\eta} \& \control{} \& \qw
    \end{tikzcd} 
    = \begin{tikzcd}[ampersand replacement=\&]
    \& \gate{\Z_\alpha} \& \ctrl{1} \& \gate{\Y_\beta} \& \ctrl{1} \& \gate{\Z_\delta} \& \qw \\
    \& \gate{\Z_\epsilon} \& \control{} \& \gate{\Y_\zeta} \& \control{} \& \gate{\Z_\eta} \& \qw
    \end{tikzcd} \\
    & \equiv \begin{tikzcd}[ampersand replacement=\&]
    \& \ctrl{1} \& \gate{\Y_\beta} \& \ctrl{1} \& \qw \\
    \& \control{} \& \gate{\Y_\zeta} \& \control{} \& \qw
    \end{tikzcd},
\end{align*}
where by $\equiv$ we intend local equivalence.  This circuit therefore traces out at most a two-parameter subfamily of gates within $\A_{C_2}$, which cannot be the image of a top dimensional set in $PU(4)$ and hence cannot have positive Haar volume.

This kind of argument turns out to be flexible enough that the commutation property powering it deserves its own name:

\begin{definition}
A two-qubit gate $U$ is said to \textit{leak (on the first qubit wire)}%
\footnote{Added in post-publication revision: ``leakiness'' was previously considered by Koponen, Bergholm, and Salomaa~\cite{KBS}, who achieve better results than we present here whenever there is overlap.}
when there are exponential families $A_\theta$, $B_\theta$, and $C_\theta$ such that
\[
\begin{tikzcd}
& \gate[2]{U} & \gate{A_\theta} & \qw \\
& \qw & \qw & \qw
\end{tikzcd} =
\begin{tikzcd}
& \gate{B_\theta} & \gate[2]{U} & \qw \\
& \gate{C_\theta} & \qw & \qw.
\end{tikzcd}
\]
\end{definition}

In fact, the other two-qubit gates in the Quil standard library are also leaky, as portrayed in \Cref{LeakyFigure}.  This table has two remarkable features: first, that there are so many such relations, and second, that the single-qubit rotation is always a $\Z$.  We now show that at least the second of these is to be expected:

\begin{figure}
\begin{align*}
    \begin{tikzcd}[ampersand replacement=\&, row sep=2em, column sep=0em]
    \arrow[dash,thick]{rr} \& \gate[2, style={rectangle, inner xsep=-2em}, label style={rotate=90, outer sep=-4em, inner sep=-4em}]{\ISWAP} \& \gate{\Z_\alpha} \& \qw \\
    \& \qw \& \qw \& \qw
    \end{tikzcd} & =
    \begin{tikzcd}[ampersand replacement=\&, row sep=2em, column sep=0em]
    \& \qw \arrow[dash,thick]{rr} \& \gate[2, style={rectangle, inner xsep=-2em}, label style={rotate=90, outer sep=-4em, inner sep=-4em}]{\ISWAP} \& \qw \\
    \& \gate{\Z_{\alpha}} \& \qw \& \qw
    \end{tikzcd}, &
    \begin{tikzcd}[ampersand replacement=\&]
    \& \ctrl{1} \& \gate{\Z_\alpha} \& \qw \\
    \& \gate{\Z_\theta} \& \qw \& \qw
    \end{tikzcd} & =
    \begin{tikzcd}[ampersand replacement=\&]
    \& \gate{\Z_\alpha} \& \ctrl{1} \& \qw \\
    \& \qw \& \gate{\Z_\theta} \& \qw
    \end{tikzcd}, &
    \begin{tikzcd}[ampersand replacement=\&, row sep=2em, column sep=0em]
    \arrow[dash,thick]{rr} \& \gate[2, style={rectangle, inner xsep=-2em}, label style={rotate=90, outer sep=-4em, inner sep=-4em}]{\PSWAP_\theta} \& \gate{\Z_\alpha} \& \qw \\
    \& \qw \& \qw \& \qw
    \end{tikzcd} & =
    \begin{tikzcd}[ampersand replacement=\&, row sep=2em, column sep=0em]
    \& \qw \arrow[dash,thick]{rr} \& \gate[2, style={rectangle, inner xsep=-2em}, label style={rotate=90, outer sep=-4em, inner sep=-4em}]{\PSWAP_\theta} \& \qw \\
    \& \gate{\Z_\alpha} \arrow[dash,thick]{rr} \& \qw \& \qw
    \end{tikzcd}, \\
    \begin{tikzcd}[ampersand replacement=\&, row sep=2em, column sep=0em]
    \arrow[dash,thick]{rr} \& \gate[2, style={rectangle, inner xsep=-2em}, label style={rotate=90, outer sep=-4em, inner sep=-4em}]{\ISWAP} \& \qw \& \qw \\
    \& \qw \& \gate{\Z_\alpha} \& \qw
    \end{tikzcd} & =
    \begin{tikzcd}[ampersand replacement=\&, row sep=2em, column sep=0em]
    \& \gate{\Z_{\alpha}} \arrow[dash,thick]{rr} \& \gate[2, style={rectangle, inner xsep=-2em}, label style={rotate=90, outer sep=-4em, inner sep=-4em}]{\ISWAP} \& \qw \\
    \& \qw \& \qw \& \qw
    \end{tikzcd}, &
    \begin{tikzcd}[ampersand replacement=\&]
    \& \ctrl{1} \& \qw \& \qw \\
    \& \gate{\Z_\theta} \& \gate{\Z_\alpha} \& \qw
    \end{tikzcd} & =
    \begin{tikzcd}[ampersand replacement=\&]
    \& \qw \& \ctrl{1} \& \qw \\
    \& \gate{\Z_\alpha} \& \gate{\Z_\theta} \& \qw
    \end{tikzcd}, &
    \begin{tikzcd}[ampersand replacement=\&, row sep=2em, column sep=0em]
    \arrow[dash,thick]{rr} \& \gate[2, style={rectangle, inner xsep=-2em}, label style={rotate=90, outer sep=-4em, inner sep=-4em}]{\PSWAP_\theta} \& \qw \& \qw \\
    \& \qw \& \gate{\Z_\alpha} \& \qw
    \end{tikzcd} & =
    \begin{tikzcd}[ampersand replacement=\&, row sep=2em, column sep=0em]
    \& \gate{\Z_\alpha} \arrow[dash,thick]{rr} \& \gate[2, style={rectangle, inner xsep=-2em}, label style={rotate=90, outer sep=-4em, inner sep=-4em}]{\PSWAP_\theta} \& \qw \\
    \& \qw \& \qw \& \qw
    \end{tikzcd}.
\end{align*}
\caption{Leakiness relations for the other standard Quil gates}\label{LeakyFigure}
\end{figure}

\begin{lemma}\label{LeakinessIsInvariant}
Leakiness is invariant under $\equiv_L$.  Specifically, if $U$ satisfies
\[
\begin{tikzcd}
& \gate[2]{U} & \gate{A_\theta} & \qw \\
& \qw & \qw & \qw
\end{tikzcd} =
\begin{tikzcd}
& \gate{B_\theta} & \gate[2]{U} & \qw \\
& \gate{C_\theta} & \qw & \qw.
\end{tikzcd}
\]
for some single-qubit exponential families $A$, $B$, $C$, and if $V$ is given by
\[
\begin{tikzcd}
& \gate[2]{V} & \qw \\
& \qw & \qw
\end{tikzcd}
=
\begin{tikzcd}
& \gate{M} & \gate[2]{U} & \gate{L} & \qw \\
& \gate{N} & \qw & \gate{R'} & \qw
\end{tikzcd}
\equiv
\begin{tikzcd}
& \gate[2]{U} & \qw \\
& \qw & \qw
\end{tikzcd},
\]
then we also have
\[
\begin{tikzcd}
& \gate[2]{V} & \gate{D_\theta} & \qw \\
& \qw & \qw & \qw
\end{tikzcd} =
\begin{tikzcd}
& \gate{E_\theta} & \gate[2]{V} & \qw \\
& \gate{F_\theta} & \qw & \qw.
\end{tikzcd}
\]
for $D_\theta = A_\theta^L$, $E_\theta = B_\theta^M$, and $F_\theta = C_\theta^N$.
\end{lemma}
\begin{proof}
This is a direct calculation:
\begin{align*}
    \begin{tikzcd}[ampersand replacement=\&]
    \& \gate[2]{V} \& \gate{D_\theta} \& \qw \\
    \& \qw \& \qw \& \qw
    \end{tikzcd} & = 
    \begin{tikzcd}[ampersand replacement=\&]
    \& \gate{M} \& \gate[2]{U} \& \gate{L} \& \gate{D_\theta} \& \qw \\
    \& \gate{N} \& \qw \& \gate{L'} \& \qw \& \qw
    \end{tikzcd} \\
    & =
    \begin{tikzcd}[ampersand replacement=\&]
    \& \gate{M} \& \gate[2]{U} \& \gate{A_\theta} \& \gate{L} \& \qw \\
    \& \gate{N} \& \qw \& \qw \& \gate{L'} \& \qw
    \end{tikzcd} \\
    & =
    \begin{tikzcd}[ampersand replacement=\&]
    \& \gate{M} \& \gate{B_\theta} \& \gate[2]{U} \& \gate{L} \& \qw \\
    \& \gate{N} \& \gate{C_\theta} \& \qw \& \gate{L'} \& \qw
    \end{tikzcd} \\
    & =
    \begin{tikzcd}[ampersand replacement=\&]
    \& \gate{E_\theta} \& \gate{M} \& \gate[2]{U} \& \gate{L} \& \qw \\
    \& \gate{F_\theta} \& \gate{N} \& \qw \& \gate{L'} \& \qw
    \end{tikzcd} \\
    & =
    \begin{tikzcd}[ampersand replacement=\&]
    \& \gate{E_\theta} \& \gate[2]{V} \& \qw \\
    \& \gate{F_\theta} \& \qw \& \qw
    \end{tikzcd} . \qedhere
\end{align*}
\end{proof}

\begin{remark}\label{RearrangingForZ}
Suppose $U$ is as in \Cref{LeakinessIsInvariant}.  By picking single-qubit operators $L$, $M$, and $N$ with
\begin{align*}
    A_{k \theta}^L & = \Z_\theta, &
    B_{k  \theta}^M & = \Z_{\ell \theta}, &
    C_{k  \theta}^N & = \Z_{\ell' \theta},
\end{align*}
we may replace $U$ by $V$, also as in \Cref{LeakinessIsInvariant}, for which we then have
\[
\begin{tikzcd}
& \gate[2]{V} & \gate{\Z_\theta} & \qw \\
& \qw & \qw & \qw
\end{tikzcd} =
\begin{tikzcd}
& \gate{\Z_{\ell \theta}} & \gate[2]{V} & \qw \\
& \gate{\Z_{\ell' \theta}} & \qw & \qw.
\end{tikzcd}
\]
In fact, if $U$ has a second leakiness relation on the other qubit wire, transforming the single-qubit gate $A'_\theta$ into the local operator $B'_\theta \otimes C'_\theta$, one may reuse $M$ and $N$: because $A \otimes 1$ and $1 \otimes A'$ commute, $B \otimes C$ and $B' \otimes C'$ must also commute, which forces $B$ and $B'$ (hence $B^M = \Z$ and $(B')^M$) to lie in the same one-parameter family and the same for $C$ and $C'$ (hence $C^N = \Z$ and $(C')^N$).  
\end{remark}

However, the first observation---that Smith, Curtis, and Zeng's standard library contains so many leaky gates---is much more of an accident.

\begin{lemma}\label{GenericEntanglerDoesNotLeak}
A generic two-qubit gate does not leak.
\end{lemma}
\begin{proof}
At the level of Lie algebras, a leaky gate $U$ satisfies \[(\su(2) \oplus \su(2)) \cap \Ad_U (\su(2) \oplus 0) \ne \emptyset,\] witnessed by anti-Hermitian matrices 
\[h = \left( \begin{array}{cc} c & a + b i \\ - a + bi & d \end{array} \right), \]
and
\[ \small h' \otimes h'' = \left( \begin{array}{cccc}
(c_0 + c_1) i & a_0 + b_0 i & a_0 + b_1 i & 0 \\
-a_0 + b_0 i & (-c_0 + c_1) i & 0 & a_1 + b_1 i \\
-a_1 + b_1 i & 0 & (c_0 - c_1) i & a_0 + b_0 i \\
0 & -a_1 + b_1 i & -a_0 + b_0 i & -(c_0 + c_1) i
\end{array}\right)\]
which in particular satisfy
\[h' \otimes h'' = U^\dagger \left( \begin{array}{c|c} h & 0 \\ \hline 0 & h \end{array} \right) U.\]
Elements of this product are computed by \[(h' \otimes h'')_{i\ell} = \sum_{p=0}^1 \sum_{j, k=1}^2 \overline{U_{(2p+j)i}} h_{(2p+j)(2p+k)} U_{(2p+k)\ell},\] which for a fixed value of $U$ gives a linear system of \emph{real} equations in the \emph{real} unknowns specifying the elements of $\su(2) \otimes 1$ and $\su(2) \otimes \su(2)$.

We claim that this system is generically of full rank, i.e., there is no solution but the trivial one.  This system drops rank only when all determinants of all maximal subminors of the system vanish.  As each determinant is an algebraic function on the real algebraic variety determined by $SU(4)$, if these do not all simultaneously vanish everywhere, then they generically do not simultaneously vanish.  We therefore need only exhibit a point where the system has full rank for the conclusion to follow.  Selecting $g = \sqrt{\ISWAP}$, we make the manual calculation that the above system of equations is satisfied only for $h = 0$, $h' = h'' = 0$.\footnote{Alternatively, using \Cref{MaximizingXYThm}, $\LogCoords(P^2_{\DB})$ has positive volume, hence $\DB$ cannot be leaky, hence the linear system studied here has no solutions.}
\end{proof}

\begin{remark}
The above mode of proof can be adapted to show that a generic entangler $U$ has associated set $P^2_U$ of positive volume.  We rely on the following pair of geometric facts:

\begin{lemma}[Sard's theorem, \cite{Sard}]\label{SardsTheorem}
Let $U \subseteq \R^n$ be an open set, and let $f\co U \to \R^m$ be differentiable.  The image of $f$ contains an open ball (hence is of positive volume) if and only if there exists a point $u \in U$ so that the derivative $D_u f$ is surjective. \qed
\end{lemma}

\begin{lemma}[{\cite[Section I.1]{Hartshorne}}]\label{ZeroLociDropDimension}
Let $f\co V \to W$ be an algebraic morphism of connected algebraic varieties.  For $w \in W$ in the image of $f$, the set $f^{-1}(w)$ is either all of $W$ or of positive codimension (hence of zero volume in the real case). \qed
\end{lemma}

\noindent We use these tools to analyze the algebraic function
\begin{align*}
    \mathrm{cov}\co SU(4) \times (SU(2)^{\otimes 2})^{\times 3} & \to SU(4) \\
    (U, A, B, C) & \mapsto A U^{-1} B U C.
\end{align*}
Fixing $U$ to be the $\B$--gate, it is known that $\mathrm{cov}|_{U=\B}$ is surjective~\cite{ZVSWMinimum}, hence \Cref{SardsTheorem} shows that there is a point $(X, Y, Z)$ in the domain at which the derivative $D_{(X, Y, Z)} \mathrm{cov}|_{U = \B}$ is surjective.  This is equivalent to the claim that the family of determinants $\det M_i D_{(X, Y, Z)} \mathrm{cov}|_{U=\B}$ do not all simultaneously vanish, where $M_i$ ranges over the $\binom{18}{15}$ different $(15 \times 15)$--minors.  Each of these determinants can be thought of as an algebraic function in the gate used to restrict $\mathrm{cov}$: \[f_i(V) = \det M_i D_{(X, Y, Z)} \mathrm{cov}|_{U=V},\] and the above claim becomes the claim that $(f_i)|_{V = \B}$ is not equal to the origin in $\R^{\binom{18}{15}}$.  In turn, \Cref{ZeroLociDropDimension} says that $(f_i)^{-1}(0)$ has zero volume, so that $(f_i)(V)$ is nonvanishing generically in $V$, $D_{(X, Y, Z)} \mathrm{cov}|_{U = V}$ is surjective generically in $V$, and the image of $\mathrm{cov}|_{U = V}$ has positive volume generically in $V$.
\end{remark}

\begin{remark}
On the other hand, leakiness is an essential part of quantum error correction codes: the very definition of a nonleaky multi-qubit gate means that a locally correctable error becomes a locally uncorrectable error after application of the entangler.  This can severely dampen the functionality of stabilizer-type codes which rely on an understanding of the rate of error propagation.
\end{remark}


\begin{figure}[p]
\begin{align*}
    \LogCoords\left( \begin{tikzcd}[ampersand replacement=\&] \& \ctrl{1} \& \ctrl{1} \& \qw \\ \& \control{} \& \control{} \& \qw \end{tikzcd} \right) & = e_1, &
    \LogCoords\left( \begin{tikzcd}[ampersand replacement=\&] \& \ctrl{1} \& \qw \& \ctrl{1} \& \qw \\ \& \control{} \& \gate{\X_{\frac{\pi}{2}}} \& \control{} \& \qw \end{tikzcd} \right) & = e_2, &
    \LogCoords\left( \begin{tikzcd}[ampersand replacement=\&] \& \ctrl{1} \& \gate{\X_{\frac{\pi}{2}}} \& \ctrl{1} \& \qw \\ \& \control{} \& \gate{\X_{\frac{\pi}{2}}} \& \control{} \& \qw \end{tikzcd} \right) & = e_3.
\end{align*}
\caption{Realizations for the extremal vertices of $\LogCoords(P^2_{\CZ})$ as circuits.}\label{RealizationsP2CZ}
\end{figure}

\begin{figure}[p]
\begin{align*}
    e_1 & = \LogCoords\left( \begin{tikzcd}[ampersand replacement=\&] \& \ctrl{1} \& \qw \& \ctrl{1} \& \qw \& \ctrl{1} \& \qw \\ \& \control{} \& \gate{\X_{\frac{\pi}{2}}} \& \control{} \& \gate{\Y_{\frac{\pi}{2}}} \& \control{} \& \qw \end{tikzcd} \right), &
    e_2 & = \LogCoords\left( \begin{tikzcd}[ampersand replacement=\&] \& \ctrl{1} \& \qw \& \ctrl{1} \& \qw \& \ctrl{1} \& \qw \\ \& \control{} \& \gate{\X_\pi} \& \control{} \& \gate{\Y_{\frac{\pi}{2}}} \& \control{} \& \qw \end{tikzcd} \right), \\
    e_3 & = \LogCoords\left( \begin{tikzcd}[ampersand replacement=\&] \& \ctrl{1} \& \gate{\X_{\frac{\pi}{2}}} \& \ctrl{1} \& \qw \& \ctrl{1} \& \qw \\ \& \control{} \& \gate{\X_\pi} \& \control{} \& \gate{\Y_{\frac{\pi}{2}}} \& \control{} \& \qw \end{tikzcd} \right), &
    e_4 & = \LogCoords \left( \begin{tikzcd}[ampersand replacement=\&] \& \ctrl{1} \& \gate{\X_{\frac{\pi}{2}}} \& \ctrl{1} \& \gate{\X_{\frac{\pi}{2}}} \& \ctrl{1} \& \qw \\ \& \control{} \& \gate{\X_{\frac{\pi}{2}}} \& \control{} \& \gate{\X_{\frac{\pi}{2}}} \& \control{} \& \qw \end{tikzcd} \right), \\
    e_5 & = \LogCoords \left( \begin{tikzcd}[ampersand replacement=\&] \& \ctrl{1} \& \gate{\Y_{-\frac{\pi}{4}}} \& \ctrl{1} \& \gate{\Y_{-\frac{3\pi}{4}}} \& \ctrl{1} \& \qw \\ \& \control{} \& \gate{\Z_{-\frac{\pi}{4}} \X_{\frac{\pi}{2}}} \& \control{} \& \gate{\X_{-\frac{\pi}{2}}} \& \control{} \& \qw
    \end{tikzcd} \right).
\end{align*}
\caption{Realizations for the extremal vertices of $\LogCoords(P^3_{\CZ})$ as circuits.}\label{RealizationsP3CZ}
\end{figure}

\begin{figure}[p]
\begin{align*}
    \LogCoords\left( \begin{tikzcd}[ampersand replacement=\&, row sep=2em, column sep=0em] \arrow[dash,thick]{rr} \& \gate[2, style={rectangle, inner xsep=-2em}, label style={rotate=90, outer sep=-4em, inner sep=-4em}]{\ISWAP} \arrow[dash,thick]{rr} \& \gate[2, style={rectangle, inner xsep=-2em}, label style={rotate=90, outer sep=-4em, inner sep=-4em}]{\ISWAP} \& \qw \\ \& \qw \& \qw \& \qw \end{tikzcd} \right) & = e_1, &
    \LogCoords\left( \begin{tikzcd}[ampersand replacement=\&, row sep=2em, column sep=0em] \arrow[dash,thick]{rr} \& \gate[2, style={rectangle, inner xsep=-2em}, label style={rotate=90, outer sep=-4em, inner sep=-4em}]{\ISWAP} \& \qw \arrow[dash,thick]{rr} \& \gate[2, style={rectangle, inner xsep=-2em}, label style={rotate=90, outer sep=-4em, inner sep=-4em}]{\ISWAP} \& \qw \\ \& \qw \& \gate{\X_{\frac{\pi}{2}}} \& \qw \& \qw \end{tikzcd} \right) & = e_2, &
    \LogCoords\left( \begin{tikzcd}[ampersand replacement=\&, row sep=2em, column sep=0em] \arrow[dash,thick]{rr} \& \gate[2, style={rectangle, inner xsep=-2em}, label style={rotate=90, outer sep=-4em, inner sep=-4em}]{\ISWAP} \& \gate{\X_{\frac{\pi}{2}}} \arrow[dash,thick]{rr} \& \gate[2, style={rectangle, inner xsep=-2em}, label style={rotate=90, outer sep=-4em, inner sep=-4em}]{\ISWAP} \& \qw \\ \& \qw \& \gate{\X_{\frac{\pi}{2}}} \& \qw \& \qw \end{tikzcd} \right) & = e_3.
\end{align*}
\caption{Realizations for the extremal vertices of $\LogCoords(P^2_{\ISWAP})$ as circuits.}\label{RealizationsP2ISWAP}
\end{figure}

\begin{figure}[p]
\begin{align*}
    e_1 & = \LogCoords\left( \begin{tikzcd}[ampersand replacement=\&, row sep=2em, column sep=0em] \arrow[dash,thick]{rr} \& \gate[2, style={rectangle, inner xsep=-2em}, label style={rotate=90, outer sep=-4em, inner sep=-4em}]{\ISWAP} \& \gate{\Y_{\frac{\pi}{2}}} \arrow[dash,thick]{rr} \& \gate[2, style={rectangle, inner xsep=-2em}, label style={rotate=90, outer sep=-4em, inner sep=-4em}]{\ISWAP} \& \gate{\X_{\frac{\pi}{2}}} \arrow[dash,thick]{rr} \& \gate[2, style={rectangle, inner xsep=-2em}, label style={rotate=90, outer sep=-4em, inner sep=-4em}]{\ISWAP} \& \qw \\ \& \qw \& \gate{\Y_{\frac{\pi}{2}}} \& \qw \& \gate{\X_{\frac{\pi}{2}}} \& \qw \& \qw \end{tikzcd} \right), &
    e_2 & = \LogCoords\left( \begin{tikzcd}[ampersand replacement=\&, row sep=2em, column sep=0em] \arrow[dash,thick]{rr} \& \gate[2, style={rectangle, inner xsep=-2em}, label style={rotate=90, outer sep=-4em, inner sep=-4em}]{\ISWAP} \& \gate{\X_{\frac{\pi}{2}}} \arrow[dash,thick]{rr} \& \gate[2, style={rectangle, inner xsep=-2em}, label style={rotate=90, outer sep=-4em, inner sep=-4em}]{\ISWAP} \& \gate{\X_{\frac{\pi}{2}}} \arrow[dash,thick]{rr} \& \gate[2, style={rectangle, inner xsep=-2em}, label style={rotate=90, outer sep=-4em, inner sep=-4em}]{\ISWAP} \& \qw \\ \& \qw \& \gate{\Y_{\frac{\pi}{2}}} \& \qw \& \gate{\X_{\frac{\pi}{2}}} \& \qw \& \qw \end{tikzcd} \right), \\
    e_3 & = \LogCoords\left( \begin{tikzcd}[ampersand replacement=\&, row sep=2em, column sep=0em] \arrow[dash,thick]{rr} \& \gate[2, style={rectangle, inner xsep=-2em}, label style={rotate=90, outer sep=-4em, inner sep=-4em}]{\ISWAP} \& \gate{\X_{\frac{\pi}{2}}} \arrow[dash,thick]{rr} \& \gate[2, style={rectangle, inner xsep=-2em}, label style={rotate=90, outer sep=-4em, inner sep=-4em}]{\ISWAP} \& \gate{\X_{\frac{\pi}{2}}} \arrow[dash,thick]{rr} \& \gate[2, style={rectangle, inner xsep=-2em}, label style={rotate=90, outer sep=-4em, inner sep=-4em}]{\ISWAP} \& \qw \\ \& \qw \& \gate{\X_{\frac{\pi}{2}}} \& \qw \& \gate{\X_{\frac{\pi}{2}}} \& \qw \& \qw \end{tikzcd} \right), &
    e_4 & = \LogCoords \left( \begin{tikzcd}[ampersand replacement=\&, row sep=2em, column sep=0em] \arrow[dash,thick]{rr} \& \gate[2, style={rectangle, inner xsep=-2em}, label style={rotate=90, outer sep=-4em, inner sep=-4em}]{\ISWAP} \& \gate{\X_{-\frac{\pi}{2}}} \arrow[dash,thick]{rr} \& \gate[2, style={rectangle, inner xsep=-2em}, label style={rotate=90, outer sep=-4em, inner sep=-4em}]{\ISWAP} \& \qw \arrow[dash,thick]{rr} \& \gate[2, style={rectangle, inner xsep=-2em}, label style={rotate=90, outer sep=-4em, inner sep=-4em}]{\ISWAP} \& \qw \\ \& \qw \& \qw \& \qw \& \gate{\X_{\frac{\pi}{2}}} \& \qw \& \qw \end{tikzcd} \right), \\
    e_5 & = \LogCoords \left( \begin{tikzcd}[ampersand replacement=\&, row sep=2em, column sep=0em] \arrow[dash,thick]{rr} \& \gate[2, style={rectangle, inner xsep=-2em}, label style={rotate=90, outer sep=-4em, inner sep=-4em}]{\ISWAP} \& \gate{\Z_{-\frac{\pi}{4}} \X_{-\frac{\pi}{2}}} \arrow[dash,thick]{rr} \& \gate[2, style={rectangle, inner xsep=-2em}, label style={rotate=90, outer sep=-4em, inner sep=-4em}]{\ISWAP} \& \gate{\Y_{\frac{3\pi}{4}}} \arrow[dash,thick]{rr} \& \gate[2, style={rectangle, inner xsep=-2em}, label style={rotate=90, outer sep=-4em, inner sep=-4em}]{\ISWAP} \& \qw \\ \& \qw \& \gate{\Y_{\frac{\pi}{4}}} \& \qw \& \gate{\X_{\frac{\pi}{2}}} \& \qw \& \qw \end{tikzcd} \right).
\end{align*}
\caption{Realizations for the extremal vertices of $\LogCoords(P^3_{\ISWAP})$ as circuits.}\label{RealizationsP3ISWAP}
\end{figure}

\begin{figure}[p]
$n = 0$: \[\left\{\left(0, 0, 0, 0\right)\right\}.\]
$n = 1$: \[\left\{\left(\frac{1}{8}, \frac{1}{8}, -\frac{1}{8}, -\frac{1}{8}\right)\right\}.\]
$n = 2$:
\begin{align*}
\left\{\left(\frac{1}{4}, \frac{1}{4}, -\frac{1}{4}, -\frac{1}{4}\right)\right., & &
\left(0, 0, 0, 0\right), & & 
\left.\left(\frac{1}{4}, 0, 0, -\frac{1}{4}\right)\right\}.
\end{align*}
$n = 3$:
\begin{align*}
\left\{\left(\frac{3}{8}, \frac{1}{8}, -\frac{1}{8}, -\frac{3}{8}\right),\right. & & 
\left(\frac{3}{8}, -\frac{1}{8}, -\frac{1}{8}, -\frac{1}{8}\right), & &
\left(0, 0, 0, 0\right), & &
\left(\frac{7}{24}, \frac{7}{24}, -\frac{5}{24}, -\frac{9}{24}\right), \\
\left(\frac{1}{4}, \frac{1}{4}, -\frac{1}{4}, -\frac{1}{4}\right), & &
\left(\frac{1}{8}, \frac{1}{8}, \frac{1}{8}, -\frac{3}{8}\right), & &
\left.\left(\frac{3}{8}, 0, 0, -\frac{3}{8}\right)\right\}, \\ \\
\left\{\left(\frac{3}{8}, \frac{3}{8}, -\frac{1}{8}, -\frac{5}{8}\right),\right. & &
\left(\frac{3}{8}, \frac{1}{8}, -\frac{1}{8}, -\frac{3}{8}\right), & &
\left(\frac{7}{24}, \frac{1}{8}, -\frac{5}{24}, -\frac{5}{24}\right), & &
\left(\frac{1}{4}, \frac{1}{4}, -\frac{1}{4}, -\frac{1}{4}\right), \\
\left.\left(\frac{1}{8}, \frac{1}{8}, -\frac{1}{8}, -\frac{1}{8}\right)\right\}.
\end{align*}
$n = 4$:
\begin{align*}
\left\{\left(\frac{1}{3}, \frac{1}{3}, -\frac{1}{6}, -\frac{1}{2}\right),\right. & &
\left(\frac{1}{2}, 0, 0, -\frac{1}{2}\right), & &
\left(\frac{1}{6}, \frac{1}{6}, \frac{1}{6}, -\frac{1}{2}\right), & &
\left(\frac{1}{4}, \frac{1}{4}, -\frac{1}{4}, -\frac{1}{4}\right), \\
\left(0, 0, 0, 0\right), & &
\left.\left(\frac{3}{8}, -\frac{1}{8}, -\frac{1}{8}, -\frac{1}{8}\right)\right\}, \\ \\
\left\{\left(\frac{1}{4}, 0, 0, -\frac{1}{4}\right),\right. & &
\left(\frac{1}{4}, \frac{1}{4}, 0, -\frac{1}{2}\right), & &
\left(\frac{1}{3}, \frac{1}{3}, 0, -\frac{2}{3}\right), & &
\left(\frac{1}{6}, 0, -\frac{1}{12}, -\frac{1}{12}\right), \\
\left(\frac{1}{8}, \frac{1}{8}, -\frac{1}{8}, -\frac{1}{8}\right), & &
\left(\frac{1}{2}, 0, 0, -\frac{1}{2}\right), & &
\left(\frac{3}{8}, \frac{3}{8}, -\frac{1}{8}, -\frac{5}{8}\right), & &
\left(\frac{1}{3}, 0, -\frac{1}{6}, -\frac{1}{6}\right), \\
\left. \left(\frac{1}{4}, \frac{1}{4}, -\frac{1}{4}, -\frac{1}{4}\right)\right\}.
\end{align*}
$n = 5$:
\begin{align*}
\left\{\left(\frac{3}{8}, \frac{3}{8}, -\frac{1}{8}, -\frac{5}{8}\right),\right. & &
\left(\frac{1}{4}, \frac{1}{4}, -\frac{1}{4}, -\frac{1}{4}\right), & &
\left(\frac{3}{8}, -\frac{1}{8}, -\frac{1}{8}, -\frac{1}{8}\right), & &
\left(0, 0, 0, 0\right), \\
\left(\frac{3}{8}, \frac{1}{8}, \frac{1}{8}, -\frac{5}{8}\right), & &
\left(\frac{1}{2}, 0, 0, -\frac{1}{2}\right), & &
\left.\left(\frac{5}{24}, \frac{5}{24}, \frac{5}{24}, -\frac{5}{8}\right)\right\}, \\ \\
\left\{\left(\frac{3}{8}, \frac{1}{8}, \frac{1}{8}, -\frac{5}{8}\right),\right. & &
\left(\frac{1}{8}, \frac{1}{8}, \frac{1}{8}, -\frac{3}{8}\right), & &
\left(0, 0, 0, 0\right), & &
\left(\frac{7}{24}, \frac{7}{24}, \frac{1}{8}, -\frac{17}{24}\right), \\
\left(\frac{1}{2}, 0, 0, -\frac{1}{2}\right), & &
\left(\frac{3}{8}, -\frac{1}{8}, -\frac{1}{8}, -\frac{1}{8}\right), & &
\left(\frac{1}{4}, \frac{1}{4}, -\frac{1}{4}, -\frac{1}{4}\right), & &
\left. \left(\frac{3}{8}, \frac{3}{8}, -\frac{1}{8}, -\frac{5}{8}\right)\right\}.
\end{align*}
\caption{The extermal vertices of the polytopes making up the sets $P^n_{\sqrt{\CZ}}$, $n \le 5$.  The set $P^6_{\sqrt{\CZ}}$ exhausts the complement.}\label{sqrtCZPoints}
\end{figure}

\begin{figure}[p]
\begin{align*}
\left\{\left(\frac{1}{2}, 0, 0, -\frac{1}{2}\right),\right. & &
\left(\frac{3}{8}, \frac{3}{8}, -\frac{1}{8}, -\frac{5}{8}\right), & &
\left(\frac{1}{3}, \frac{1}{3}, 0, -\frac{2}{3}\right), & &
\left(\frac{1}{3}, 0, -\frac{1}{6}, -\frac{1}{6}\right), \\
\left(0, 0, 0, 0\right), & &
\left.\left(\frac{1}{4}, \frac{1}{4}, -\frac{1}{4}, -\frac{1}{4}\right)\right\}, \\ \\
\left\{\left(\frac{1}{2}, 0, 0, -\frac{1}{2}\right),\right. & &
\left(\frac{3}{8}, -\frac{1}{8}, -\frac{1}{8}, -\frac{1}{8}\right), & &
\left(\frac{1}{4}, \frac{1}{4}, -\frac{1}{4}, -\frac{1}{4}\right), & &
\left(\frac{1}{6}, \frac{1}{6}, \frac{1}{6}, -\frac{1}{2}\right), \\
\left(0, 0, 0, 0\right), & &
\left.\left(\frac{1}{3}, \frac{1}{3}, -\frac{1}{6}, -\frac{1}{2}\right)\right\}.
\end{align*}
\caption{The extremal points of the two polytopes comprising $\LogCoords(P^2_{\XY})$.}\label{RealizationsP2XY}
\end{figure}

\begin{figure}[p]
\begin{align*}
\left\{\left(\frac{1}{4}, \frac{1}{4}, -\frac{1}{4}, -\frac{1}{4}\right),\right. & &
\left(\frac{1}{3}, 0, -\frac{1}{6}, -\frac{1}{6}\right), & &
\left(0, 0, 0, 0\right), & &
\left(\frac{3}{8}, \frac{3}{8}, -\frac{1}{8}, -\frac{5}{8}\right),  \\
\left(\frac{1}{4}, \frac{1}{4}, 0, -\frac{1}{2}\right), & &
\left(\frac{1}{2}, 0, 0, -\frac{1}{2}\right), & & 
\left.\left(\frac{3}{8}, \frac{1}{4}, 0, -\frac{5}{8}\right)\right\}, \\ \\
\left\{\left(\frac{1}{8}, \frac{1}{8}, -\frac{1}{8}, -\frac{1}{8}\right),\right. & &
\left(\frac{1}{8}, \frac{1}{8}, \frac{1}{8}, -\frac{3}{8}\right), & &
\left(\frac{1}{4}, \frac{1}{8}, \frac{1}{8}, -\frac{1}{2}\right), & &
\left(\frac{1}{4}, \frac{1}{4}, -\frac{1}{4}, -\frac{1}{4}\right), \\
\left(\frac{1}{4}, \frac{1}{4}, 0, -\frac{1}{2}\right), & &
\left(\frac{1}{3}, \frac{1}{3}, -\frac{1}{6}, -\frac{1}{2}\right), & &
\left(\frac{3}{8}, -\frac{1}{8}, -\frac{1}{8}, -\frac{1}{8}\right), & &
\left.\left(\frac{1}{2}, 0, 0, -\frac{1}{2}\right)\right\}.
\end{align*}
\caption{The extremal points of the two polytopes comprising $\LogCoords(P^2_{\DB})$.}\label{VerticesOfP2DB}
\end{figure}

\end{document}